\documentclass[10pt, journal,compsoc]{IEEEtran}

\ifCLASSOPTIONcompsoc
  \usepackage[nocompress]{cite}
\else
  \usepackage{cite}
\fi

\ifCLASSINFOpdf
\else
\fi
\pdfoutput=1
\usepackage{ifpdf}
\usepackage{cite}
\usepackage{latexsym}
\usepackage{amsmath}
\usepackage{amsthm}
\usepackage{amssymb}
\usepackage{mathrsfs}
\usepackage[dvips,dvipdf]{graphicx}
\usepackage{subfigure}
\usepackage{algorithm}
\usepackage{algorithmicx}
\usepackage{multirow}
\usepackage{bbm}
\usepackage{arydshln}
\usepackage{url}
\usepackage{pict2e,picture} 
\usepackage{tikz}   
\usepackage{pifont} 
\usepackage{enumitem}
\usepackage{mathtools}
\usepackage{algpseudocode}
\usepackage{comment}
\usepackage{xcolor}
\usepackage{caption}

\DeclarePairedDelimiter{\ceil}{\lceil}{\rceil}

\newcommand{\A}{\mathbf{A}}
\renewcommand{\a}{\mathbf{a}}
\newcommand{\B}{\mathbf{B}}

\renewcommand{\b}{\mathbf{b}}

\renewcommand{\c}{\mathbf{c}}

\newcommand{\g}{\mathbf{g}}

\newcommand{\p}{\mathbf{p}}

\newcommand{\s}{\mathbf{s}}

\newcommand{\w}{\mathbf{w}}

\newcommand{\x}{\mathbf{x}}

\newcommand{\vlambda}{\boldsymbol{\lambda}}

\newtheorem{thm}{Theorem}

\newtheorem{lem}[thm]{Lemma}

\newtheorem{defn}{Definition}

\newtheorem{rem}{Remark}

\newtheorem{fact}{Fact}

\algtext*{EndWhile}
\algtext*{EndIf}
\algtext*{EndFor}
\algtext*{EndFunction}
\algtext*{EndProcedure}

\begin{document}

\title{Toward Efficient Online Scheduling for Distributed Machine Learning Systems}

\author{Menglu~Yu, 
        Jia~Liu~\IEEEmembership{Senior Member,~IEEE},
        Chuan~Wu~\IEEEmembership{Senior Member,~IEEE},
        Bo~Ji~~\IEEEmembership{Senior Member,~IEEE}
        Elizabeth~S.~Bentley~\IEEEmembership{Member,~IEEE}
\thanks{
Menglu Yu is with the Department of Computer Science, Iowa State University, Ames, IA 50011, USA (e-mail: mengluy@iastate.edu).}
\thanks{
Jia Liu is with the Department of Electrical and Computer Engineering, The Ohio State University, Columbus, OH 43210, USA (e-mail: liu@ece.osu.edu).}
\thanks{
Bo Ji is with the Department of Computer Science, Virginia Tech, Blacksburg, VA 24061, USA (e-mail: boji@vt.edu).}
\thanks{
Chuan Wu is with the Department of Computer Science, The University of Hong Kong, Hong Kong (e-mail: cwu@cs.hku.hk).}
\thanks{
Elizabeth S. Bentley is with the Air Force Research Laboratory, Information Directorate, Rome, NY 13441, USA (e-mail: elizabeth.bentley.3@us.af.mil).}
\thanks{
Digital Object Identifier 10.1109/TNSE.2021.3104513.}%
}

\IEEEtitleabstractindextext{%
\begin{abstract}
Recent years have witnessed a rapid growth of distributed machine learning (ML) frameworks, which exploit the massive parallelism of computing clusters to expedite ML training.
However, the proliferation of distributed ML frameworks also introduces many unique technical challenges in computing system design and optimization. 
In a networked computing cluster that supports a large number of training jobs, a key question is how to design efficient scheduling algorithms to allocate workers and parameter servers across different machines to minimize the overall training time. 
Toward this end, in this paper, we develop an online scheduling algorithm that jointly optimizes resource allocation and locality decisions. 
Our main contributions are three-fold: i) We develop a new analytical model that considers both resource allocation and locality; 
ii) Based on an equivalent reformulation and observations on the worker-parameter server locality configurations, we transform the problem into a mixed packing and covering integer program, which enables approximation algorithm design; 
iii) We propose a meticulously designed approximation algorithm based on randomized rounding and rigorously analyze its performance.
Collectively, our results contribute to the state of the art of distributed ML system optimization and algorithm design.
\end{abstract}

}

\maketitle
\IEEEdisplaynontitleabstractindextext
\IEEEpeerreviewmaketitle

\section{Introduction} \label{sec:intro}

Fueled by the rapid growth of data analytics and machine learning (ML) applications, recent years have witnessed an ever-increasing hunger for computing power.
However, with hardware capability no longer advancing at the pace of the Moore's law, it has been widely recognized that a viable solution to sustain such computing power needs is to exploit {\em parallelism} at both machine and chip scales.
Indeed, the recent success of deep neural networks (DNN)
is enabled by the use of distributed ML frameworks, which exploit the massive parallelism over computing clusters with a large number of GPUs.
These distributed ML frameworks have significantly accelerated the training of  DNN for many applications (e.g., image and voice recognition, natural language processing, etc.).
To date, prevailing distributed ML frameworks include TensorFlow~\cite{TensorFlow}, MXNet~\cite{MXNet}, PyTorch~\cite{Paszke19:Pytorch}, Caffe~\cite{Caffe}, to name just a few.

However, the proliferation of distributed ML frameworks also introduces many unique technical challenges on large-scale computing system design and network resource optimization. 
Particularly, due to the decentralized nature, at the heart of distributed learning system optimization lies the problem of scheduling ML jobs and resource provisioning across different machines to minimize the total training time.
Such scheduling problems involve dynamic and combinatorial 
worker and parameter server allocations, which are inherently NP-hard.
Also, the allocations of workers and parameter servers should take {\em locality} into careful consideration, since co-located workers and parameter servers can avoid costly network communication overhead.
However, locality optimization adds yet another layer of difficulty in scheduling algorithm design.
Exacerbating the problem is the fact that the future arrival times of training jobs at an ML computing cluster are hard to predict, 
which necessitates {\em online} algorithm design without the knowledge of future job arrivals.
So far in the literature, there remains a lack of holistic theoretical studies that address all the aforementioned challenges.
Most of the existing scheduling schemes are based on simple heuristics without performance guarantee (see Section~\ref{sec:Related} for detailed discussions).
This motivates us to fill this gap and pursue efficient online scheduling designs for distributed ML resource optimization, which offer {\em provable} performance guarantee.

The main contribution of this paper is that we develop an online scheduling algorithmic framework that {\em jointly} yields resource scheduling and locality optimization decisions with strong competitive ratio performance. 
Further, we reveal interesting insights on how distributed ML frameworks affect online resource scheduling optimization.
Our main technical results are summarized as follows:

\vspace{-.01in}
\begin{list}{\labelitemi}{\leftmargin=1em \itemindent=-0.0em \itemsep=.2em}
\item By abstracting the architectures of prevailing distributed ML frameworks, we formulate an online resource scheduling optimization problem that: i) models the training of ML jobs based on the parameter server (PS) architecture and stochastic gradient descent (SGD) method; and ii) explicitly takes {\em locality} optimization into consideration.
We show that, due to the heterogeneous internal (between virtual machines or containers) and external (between physical machines) communications, the locality-aware scheduling problem contains {\em non-deterministic} constraints and is far more complex compared to the existing works that are locality-oblivious (see, e.g., ~\cite{Chun16:Dolphin,Bao18:ML_INFOCOM}). 

\item To solve the locality-aware scheduling problem, we
develop an equivalent problem reformulation to enable subsequent developments of online approximation algorithms.
Specifically, upon carefully examining the locality configurations of worker-server relationships, we are able to transform the original problem to a special-structured integer nonlinear program with mixed cover/packing-type constraints, and
the low-complexity approximation algorithm design with provable performance can be further entailed.

\item To tackle the integer nonlinear problem with mixed cover/packing-type constraints, we propose an approximation algorithm based on a meticulously designed randomized rounding scheme and then rigorously prove its performance.
We note that the results of our randomized rounding scheme are general and could be of independent theoretical interest.
Finally, by putting all algorithmic designs together, we construct a primal-dual online resource scheduling (PD-ORS) scheme, which has an overall competitive ratio that only {\em logarithmically} depends on ML job characteristics (e.g., required epochs, training samples).
\end{list}

Collectively, our results contribute to a comprehensive and fundamental understanding of distributed machine learning system optimization.
The remainder of this paper is organized as follows.
In Section~\ref{sec:Related}, we review the literature to put our work in comparative perspectives.
Section~\ref{sec:model_formulation} introduces the system model and problem formulation.
Section~\ref{sec:alg} presents our algorithms and their performance analysis.
Section~\ref{sec:numerical} presents numerical results and Section~\ref{sec:conclusion} concludes this paper.

\section{Related Work} \label{sec:Related}

As mentioned in Section~\ref{sec:intro}, due to the high computational workload of ML applications, many distributed ML frameworks (e.g., TensorFlow~\cite{TensorFlow}, MXNet~\cite{MXNet}, PyTorch~\cite{Paszke19:Pytorch}, Caffe~\cite{Caffe}) have been proposed to leverage modern large-scale computing clusters.
A common distributed training architecture implemented in these distributed ML frameworks is the PS architecture~\cite{Li14:ML_Static_OSDI,Chilim14:ML_Static_OSDI}, which employs multiple workers and PSs (implemented as virtual machines or containers) to collectively train a global ML model.
Coupled with the {\em iterative} ML training based on stochastic gradient descent (SGD), the interactions between machines in the distributed ML cluster are significantly different from those in traditional cloud computing platforms (e.g., MapReduce~\cite{Dean08:MapReduce} and Dryad~\cite{Isard07:Dryad} and references therein).
For example, a MapReduce job usually partitions the input data into independent chunks, which are then processed by the {\it map} step in a parallel fashion. 
The output of the maps are then fed to  the {\it reduce} step to be aggregated to yield the final result.
Clearly, the data flows in MapReduce are a ``one-way traffic", which is unlike those iterative data flows in ML training jobs whose completions highly depend on the ML job's convergence property.
As a result, existing job scheduling algorithms for cloud systems are not suitable for distributed ML frameworks.


Among distributed ML system studies, most of the early attempts (see, e.g., ~\cite{Li14:ML_Static_OSDI,Chilim14:ML_Static_OSDI} and references therein) only considered static allocation of workers and parameter servers.
To our knowledge, the first work on understanding the performance of distributed ML frameworks is~\cite{Yan15:MLFramework_KDD}, where Yan {\em et al.} developed analytical models to quantify the impacts of model-data partitioning and system provisioning for DNN.
Subsequently, Chun {\em et al.} ~\cite{Chun16:Dolphin} developed heuristic dynamic system reconfiguration algorithms to allocate workers and parameter servers to minimize the runtime, but without providing optimality guarantee.
The first dynamic distributed scheduling algorithm with optimality guarantee was reported in ~\cite{Sun17:ML_MILP}, where Sun {\em et al.} used standard mixed integer linear program (MILP) solver to dynamically compute the worker-parameter server partition solutions.
Due to NP-hardness of the MILP, the scalability of this approach is limited. 
The most recent work~\cite{Zhang20:onlinescheduling} proposed an online scheduling algorithm to schedule synchronous training jobs in ML clusters with the goal to minimize the weighted completion time.
However, the consecutive time slots were allocated for each training job, and the numbers of workers and parameter servers could not be adjusted.

Another line of the research is to leverage the learning-based approach to do the resource scheduling.
There are a number of recent works using deep reinforcement learning (DRL) for resource allocation, device placement, and video streaming.
For example, Mao {\it et al.}~\cite{Mao16:DRL} and Chen {\it et al.}~\cite{Chen17:DRL} designed a multi-resource cluster scheduler using DRL with the goal to minimize average job slowdown.
The proposed scheduler picks one or more of the waiting jobs in the queue and allocate to machines at each time slot, and the resource demand of each job is unknown until after its arrival.
Later, Mao {\it et al.}~\cite{Mao18:graph-based,Mao19:learningscheduling} used DRL to heuristically train scheduling policies for graph-based parallel jobs by setting both parallelism level and execution order.
Meanwhile, Mirhoseini {\it et al.}~\cite{Mirhoseini18:hierarchicalmodel, Mirhoseini17:deviceplacement} utilized DRL to design a model for efficient placement of computational graphs onto hardware devices, aiming at minimize the running time of each individual TensorFlow job. 
Although various performance gains have been empirically reported, these DRL-based studies do not offer optimality performance guarantee due to the lack of theoretical foundation of DRL as of today.

\begin{figure*}[t!]
\centering
\begin{minipage}[t]{0.3\textwidth}
\centering
\vspace{-1.26in}
\includegraphics[height=0.14\textheight]{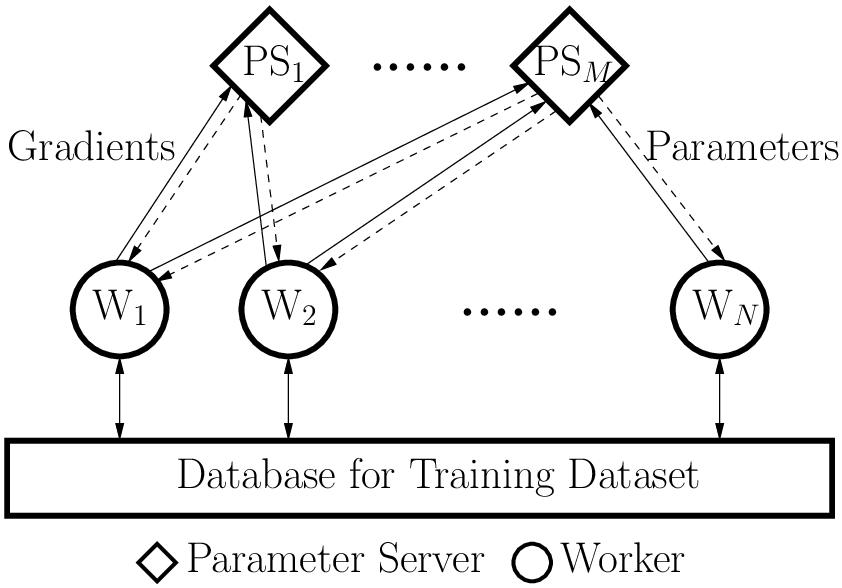}
\vspace{-.03in}
\caption{Illustration of distributed training with the PS architecture.}
\label{fig:DistrML}
\end{minipage}
\hspace{0.01\columnwidth}%
\begin{minipage}[t]{0.3\textwidth}
\centering
\vspace{-1.3in}
\includegraphics[width=0.23\textheight]{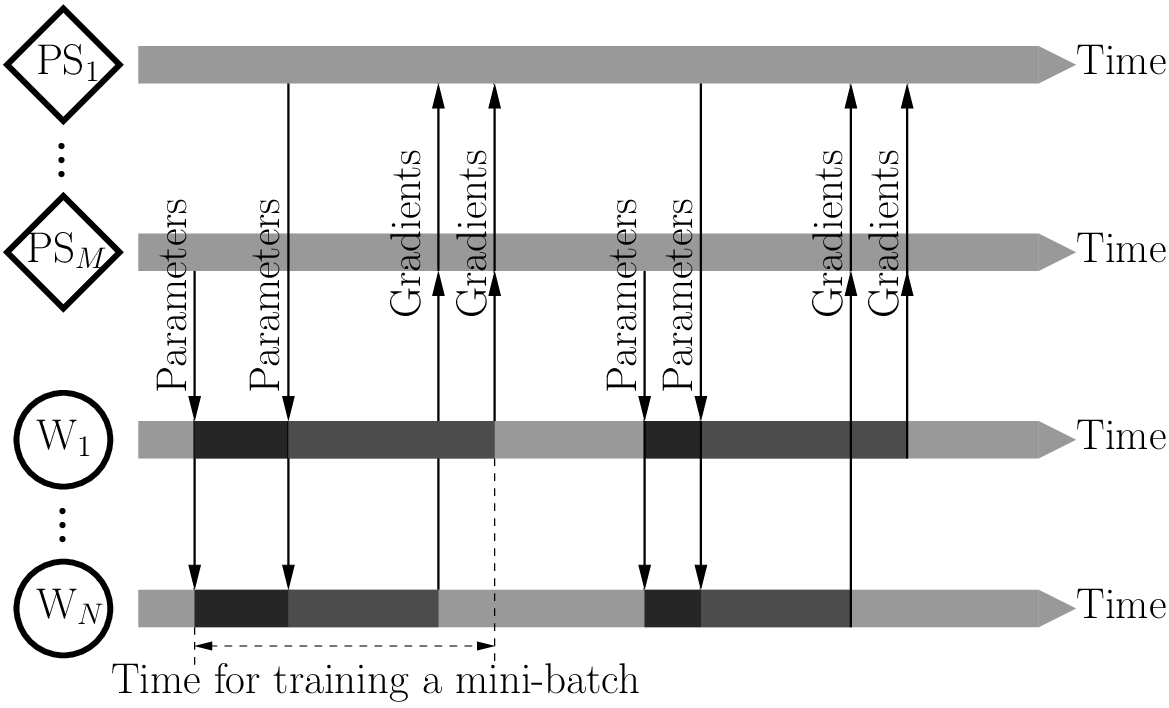}
\vspace{-.15in}
\caption{The workflow of iterative training.}
\label{fig:MLWorkflow}
\end{minipage}
\hspace{0.01\columnwidth}%
\begin{minipage}[t]{0.3\textwidth}
\centering
\includegraphics[width=0.16\textheight]{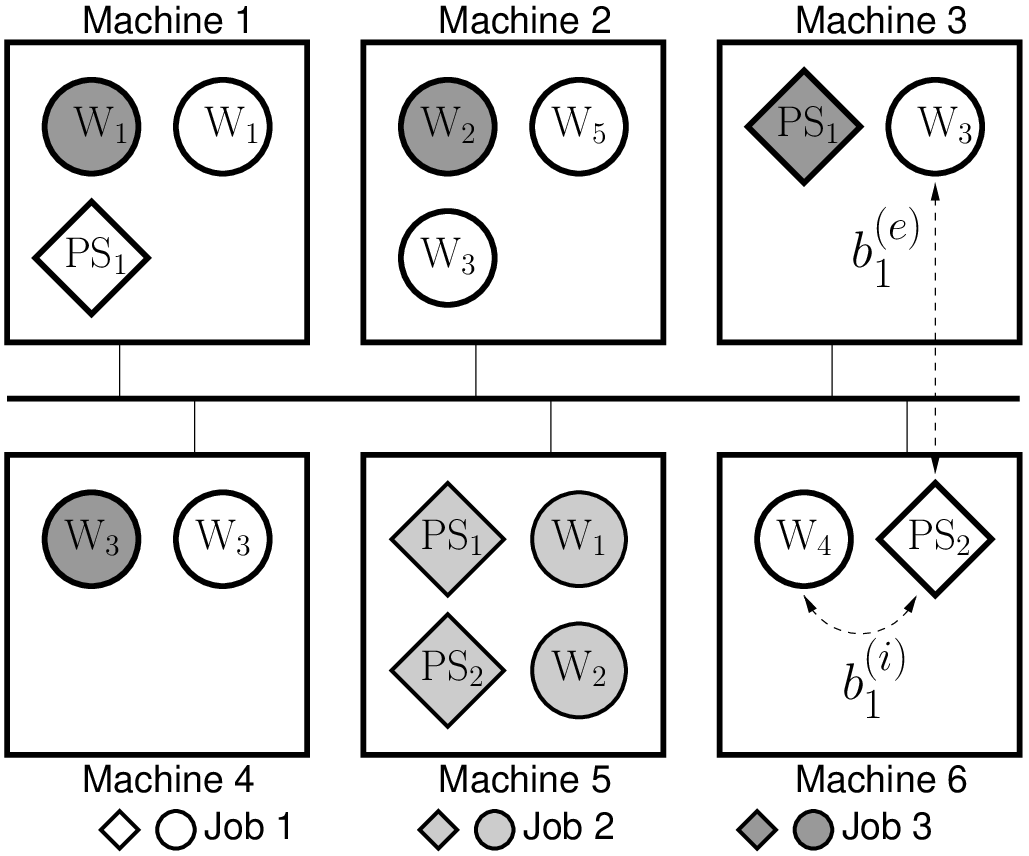}
\vspace{.06in}
\caption{Colocated parameter servers and workers on physical machines.}
\label{fig:Colocate}
\end{minipage}
\vspace{-.1in}
\end{figure*}

The most relevant work to ours is ~\cite{Bao18:ML_INFOCOM}, where Bao {\em et al.} developed an online primal-dual approximation algorithm, OASiS, to solve the scheduling problem for distributed ML systems.
Our work differs from~\cite{Bao18:ML_INFOCOM} in the following key aspects: 1) In \cite{Bao18:ML_INFOCOM}, the workers and parameter servers are allocated on two {\em strictly separated} sets of physical machines, i.e., {\em no} worker and parameter server can share the same physical machine, which significantly simplified the underlying optimization problem. In this work, we consider the cases that workers and parameter servers can be {\em co-located} on the same physical machine, which is the common practice in existing ML systems (see, e.g., ~\cite{MXNet1,MXNet2}). Such co-location can significantly reduce inter-server communication, expedite training, and improve resource utilization efficiency between workers and parameter servers.
However, as will be shown later, the co-location setting leads to an integer {\em non-convex} optimization problem with {\em non-deterministic} constraints, which is much harder and necessitates new algorithm design.
2) Ref.~\cite{Bao18:ML_INFOCOM} advocates dynamic worker number adjustment, but does {\em not} guarantee the same global batch size across the training iterations. 
According to recent literature~\cite{Goyal17:SGD}, maintaining a consistent global batch size is important for ensuring convergence of DNN training, when the worker number varies. We ensure a consistent global batch size in our model.
We note that the co-location setting was considered in~\cite{Peng18:ML_EuroSys}. However, the scheduling algorithm therein is a heuristic and does not provide performance guarantee. This motivates us to develop new algorithms with provable performance to fill this gap.

\section{System Model and Problem Formulation} \label{sec:model_formulation}

In this section, we first provide a quick overview on the architecture of distributed ML frameworks to familiarize readers with the necessary background.
Then, we will introduce our analytical models for ML jobs and resource constraints, as well as the overall problem formulation.

\smallskip
{\bf 1) Distributed Machine Learning: A Primer.} 
As illustrated in Fig.~\ref{fig:DistrML}, 
the key components of a PS-based distributed ML system include parameter servers, workers, and the training dataset, which are usually implemented over a connected computing cluster that contains multiple physical machines.
The training dataset of an ML job is stored in distributed data storage (e.g., HDFS~\cite{HDFS}) and usually divided into equal-sized data chunks.
Each data chunk further contains multiple equal-sized mini-batches.

To date, one of the most widely adopted training algorithms in distributed ML frameworks is the {\em stochastic gradient descent method} (SGD)~\cite{Dean:2012:LSD}. 
 With SGD, the interactions between workers and parameter servers are illustrated in Fig~\ref{fig:MLWorkflow}.
A worker is loaded with the DNN model (we focus on data parallel training) with current values of the model parameters (e.g., the weights of a DNN) and retrieves a new data chunk from the data storage.
In each training iteration, a worker processes one mini-batch from its data chunk to compute gradients (i.e., directions and magnitudes of parameter changes).\footnote{
As an example, in a DNN model, gradients can be computed by the well-known ``back-propagation'' approach.
}
Upon finishing a mini-batch, the worker sends the gradients to the parameter servers, receives updated global parameters, and then continues with the next training iteration to work on the next mini-batch.
On the parameter server side, parameters are updated as: $\w[k] = \w[k-1] + \alpha_{k} \g[k]$, where $\w[k]$, $\alpha_{k}$, and $\g[k]$ denote the parameter values, step-size, and stochastic gradient in the $k$-th update, respectively.

\smallskip
{\bf 2) Modeling of Learning Jobs:} 
We consider a time-slotted system.
The scheduling time-horizon is denoted as $\mathcal{T}$ with $|\mathcal{T}|=T$.
We use $\mathcal{I}$ to represent the set of training jobs and let $a_{i}$ denote the arrival time-slot of job $i \in \mathcal{I}$. 
As shown in Fig.~\ref{fig:Colocate}, parameter servers and workers 
 could spread over multiple physical machines.
We let $\mathcal{H}$ represent the set of physical machines.
For each job $i$, we use $w_{ih}[t], s_{ih}[t] \geq 0$ to represent the allocated numbers of workers and parameter servers on machine $h \in \mathcal{H}$ in each time-slot $t \geq a_{i}$, respectively.
Further, we let $ \mathcal{P}_{i}[t] \!\triangleq\! \{h \!\in\! \mathcal{H}|s_{ih}[t] \!>\! 0\}$ and $ \mathcal{W}_{i}[t] \!\triangleq\! \{h \!\in\! \mathcal{H}| w_{ih}[t] \!>\! 0\}$ denote the sets of physical machines that contain parameter servers and workers for job $i$ in time-slot $t$, respectively. 

We use a binary variable $x_{i}\in \{0,1\}$ to indicate whether job $i$ is admitted ($x_{i} =1$) or not ($x_{i}=0$).
We use $\tau_{i}$ to denote the training for each sample of job $i$.
We let $b_{i}(h,p)$ denote the data rate of the link between a worker for job $i$ (on machine $h$) and a parameter server (on machine $p$).
Each worker or parameter server is exclusively assigned to some job $i$, and $b_i(h,p)$ is reserved and decided by the user upon job submission, which is common to ensure the data transfer performance~\cite{Bao18:ML_INFOCOM}. 
Note that the value of $b_{i}(h,p)$ is {\em locality-dependent} where the slowest worker will become the bottleneck since we focus on Bulk-Synchronous-Parallel (BSP) scheme ~\cite{Cheatham05:BSP}.
Specifically, we have:
\begin{align*}
b_{i}(h,p) =
\begin{cases}
b_{i}^{(i)}, & \text{if $h=p$}, \\
b_{i}^{(e)}, & \text{otherwise},
\end{cases}
\end{align*}
where $b_{i}^{(i)}$ and $b_{i}^{(e)}$ denote the internal and external communication link rates, respectively.
For example, as shown in Fig.~\ref{fig:Colocate}, since Job 1's worker $\mathrm{W}_{4}$ and parameter server $\mathrm{PS}_{2}$ are both on the same machine, they  communicate at the internal link rate $b_{1}^{(i)}$.
On the other hand, since Job 1's worker $\mathrm{W}_{3}$ and parameter server $\mathrm{PS}_{2}$ are on different physical machines, they communicate at the external link rate $b_{1}^{(e)}$.
In practice, it usually holds that $b_{i}^{(e)} \ll b_{i}^{(i)}$.

Next, we calculate the amount of time for a worker on machine $h$ to process a sample.
We use $F_{i}$ to denote the global batch size of job $i$, which is a fixed constant across all time-slots.\footnote{We note that this fixed global batch size requirement is {\em compliant} with the standard SGD implementation~\cite{Bazaraa_Sherali_Shetty_93:NLP} and important for ensuring convergence~\cite{Goyal17:SGD}.
In contrast, the global batch size in some existing works on dynamic ML resource allocation (e.g.,~\cite{Bao18:ML_INFOCOM}) could be time-varying, which necessitates time-varying dynamic learning rate adjustments to offset correspondingly and further complicates the SGD implementation.
} 
We assume $F_i$ is equally divided among workers, i.e., the local batch size at each worker is: $F_i/\sum_{h'\in\mathcal{H}}w_{ih'}[t]$.\footnote{ Most distributed ML frameworks (e.g., Tensorflow~\cite{Keras}) set the same local batch size to each worker for the distributed training.}

We assume symmetric link speed in both directions between a worker and a PS. 
Let $g_{i}$ denote the size of gradients/parameters of job $i$.
Then, from a worker's perspective, to push gradients to and pull updated parameters from the PSs for job $i$,
the combined uplink/downlink communication time can be computed as:
$(2{g_i/\sum_{h'\in\mathcal{H}}s_{ih'}[t])/(\min_{p\in \mathcal{P}_i[t]}b_i(h,p)})$,
where the numerator term $g_i/\sum_{h'\in\mathcal{H}}s_{ih'}[t]$ follows from the assumption of even parameter distribution among the PSs, and the denominator is due to the fact that push/pull time is decided by the slowest link among all connections from the worker to all PSs (i.e., $\min_{p\in P_i[t]}b_i(h,p)$). 
Hence, the average computation and communication time to process a sample on machine $h\in\mathcal{W}_i[t]$ for job $i$ in time slot $t$ can be computed as:
\begin{align*}
\underbrace{\tau_{i}}_{\substack{\mathrm{Training \,\, time} \\ {\mathrm{per \,\, sample}}}} \!\!\! + \underbrace{ \bigg(\frac{2g_i/\sum_{h'\in\mathcal{H}}s_{ih'}[t]}{\min_{p\in \mathcal{P}_i[t]}b_i(h,p)} \bigg) \bigg/ \bigg(\frac{F_i}{\sum_{h'\in\mathcal{H}}w_{ih'}[t]} \bigg) }_{\mathrm{Communication \,\, time \,\, per \,\, sample}}.
\end{align*}
Recall that we focus on the BSP scheme, where all workers are synchronized before they proceed to the next iteration.
In other words, the total number of samples trained on machine $h\in\mathcal{W}_i[t]$ for job $i$ in time slot $t$ is determined by the slowest link among all connections from all workers to all PS (i.e., $\min_{p\in \mathcal{P}_i[t],h'\in\mathcal{W}_i[t]}b_i(h',p)$).
It then follows that the {\em number of samples trained} on machine $h\in\mathcal{W}_i[t]$ for job $i$ in time-slot $t$ can be computed as: 
\begin{align} \label{eqn_NumTrainedSamples}
\frac{w_{ih}[t]}{\tau_{i} + \Big( \frac{2g_i/\sum_{h'\in\mathcal{H}}s_{ih'}[t]}{\min_{p\in \mathcal{P}_i[t],h'\in\mathcal{W}_i[t]}b_i(h',p)} \Big) \Big/ \Big( \frac{F_i}{\sum_{h'\in\mathcal{H}}w_{ih'}[t]}\Big)}.
\end{align} 
Note that in practice, ML users usually specify a fixed ratio between worker number and PS number (e.g., often 1:1 \footnote{In practice, the ratio between numbers and parameter servers are specified by the user upon the job's submission (e.g., 1:1 in Kubernetes~\cite{Kubernetes}).}) when launching their training jobs to ensure appropriate coordinations between workers and PSs in terms of channel bandwidth, memory allocation, etc.
To model this practice, we define the ratio of worker number to PS number for each job $i$ as: 
\begin{align}\label{eqn_WorkerPSRatio}
\gamma_i \triangleq \frac{\sum_{h'\in\mathcal{H}}w_{ih'}[t]}{\sum_{h'\in\mathcal{H}}s_{ih'}[t]}, \quad \forall i, t.
\end{align}
With $\gamma_{i}$, we can rewrite Eq.~(1) as: 
\begin{align*}
\frac{w_{ih}[t]}{\tau_{i} +\frac{\gamma_i}{F_i} \frac{2g_i}{\min_{p\in \mathcal{P}_i[t],h'\in\mathcal{W}_i[t]}b_i(h',p)}}.
\end{align*}

Suppose that, for job $i$, there are $K_{i}$ data samples in its training dataset.
In practice, $K_i \gg F_i$.
In ML systems, an epoch is defined as a round of training that exhausts all data samples.
We let $E_{i}$ denote the number of epochs needed by job $i$.
In this paper, we assume that 
the epoch of each job is predetermined.
This is because it is often difficult to estimate the required number of epochs for SGD-type methods' convergence.
Therefore, most SGD-type algorithms in practice stop after a fixed number of iterations (i.e., fixed number of epochs, see, e.g., ~\cite{Wang18:SpiderBoost} and references therein) to avoid excessive training delay.

Then, the total number of samples to be processed for job $i$ over the entire training process is $E_iK_i$.
To make sure that there are sufficient workers allocated for job $i$ over the entire training horizon, we have:
\begin{align} \label{eqn_min_workers_minibatch}
\!\!\!\!\!\! \sum_{t \in \mathcal{T}}\sum_{h \in \mathcal{H}} \frac{w_{ih}[t]}{\tau_{i} +\frac{\gamma_i}{F_i} \frac{2g_i}{\min_{p\in \mathcal{P}_i[t],h'\in\mathcal{W}_i[t]}b_i(h',p)}} \!\geq\! x_{i} E_{i}K_{i},\forall i \in \mathcal{I}. \!\!\!\!
\end{align}
We note that, with co-located workers and parameter servers on each machine, Eq.~(\ref{eqn_min_workers_minibatch}) is {\em non-deterministic} due to the existence of the $\min\{\cdot\}$ operator.
As will be shown later, this non-determistic constraint makes the scheduling design far more complicated than related works~\cite{Li14:ML_Static_OSDI,Chilim14:ML_Static_OSDI,Chun16:Dolphin,Bao18:ML_INFOCOM}.

To model the fact that the largest number of assigned concurrent workers is no more than the global batch size 
 (otherwise, some workers will be idle), we have: 
\begin{align} \label{eqn_max_workers_datachunk}
\sum_{h \in \mathcal{H}}w_{ih}[t] \leq x_{i} F_i, \quad \forall i \in \mathcal{I}, a_{i} \leq t \leq T.
\end{align}

{\bf 3) Resource Constraint Modeling:}
We let $\mathcal{R}$ denote the set of resources (e.g., CPU/GPU, memory, storage, etc.).
Let $\alpha_{i}^{r}$ and $\beta_{i}^{r}$ be the amount of type-$r$ resource required by a worker and a parameter server for job $i$, respectively. 
Let $C_{h}^{r}$ be the capacity of type-$r$ resource on machine $h$. 
To ensure the resources do not exceed type-$r$'s limit, we have:
\begin{align}
\label{capacity}\sum_{i \in \mathcal{I}}(\alpha_{i}^{r}w_{ih}[t]+\beta_{i}^{r}s_{ih}[t])\leq C_{h}^{r}, \forall t \in \mathcal{T}, r \in \mathcal{R}, h \in \mathcal{H}.
\end{align}
Note that for job $i$, its completion time $\tilde{t}_{i}$ corresponds to the latest time-slot where there remain some active workers allocated for it.
Therefore, we have:
\begin{align} \label{eqn_completion_time}
\tilde{t}_{i}=\arg \max_{t \in \mathcal{T}}\bigg\{\sum_{h \in \mathcal{H}} w_{ih}[t]>0 \bigg\}, \quad \forall i \in \mathcal{I}.
\end{align}
To ensure that no workers and parameter servers are allocated before job $i$'s arrival, we have:
\begin{align} 
\label{yz}&w_{ih}[t]=s_{ih}[t]=0, \quad \forall i \in \mathcal{I}, h \in \mathcal{H}, t<a_{i}.
\end{align}

{\bf 4) Objective Function and Problem Statement:} 
Let $u_{i}(\tilde{t}_{i} - a_{i})$  be the utility function for job $i$,  which is non-increasing with respect to the training time $\tilde{t}_{i} - a_{i}$.
The utility functions could play the role of various performance metrics based on job completion times (e.g., fairness).
In this paper, our goal is to maximize the overall utility for all jobs. 
Putting all constraints and the objective function together, the offline (with knowledge of $a_{i}$, $\forall i$) distributed ML resource scheduling problem (DMLRS) can be formulated as:
\begin{align} 
\textbf{DMLRS: } & \underset{\x,\w,\s}{\text{Maximize }} \sum_{i \in \mathcal{I}}x_{i} u_{i}(\tilde{t}_{i} - a_{i}) \nonumber\\
& \text{subject to  }  \text{Constraints (\ref{eqn_min_workers_minibatch}) -- (\ref{yz})}. \nonumber
\end{align}
Problem DMLRS is an integer nonlinear program, which is NP-hard in general~\cite{Hochbaum97}.
Also, Problem DMLRS involves two {\em non-deterministic} constraints in (\ref{eqn_min_workers_minibatch}) and (\ref{eqn_completion_time}), which are not amenable for conventional optimization techniques.
Moreover, the arrivals $\{a_{i}, \forall i\}$ are often unknown in practice, which necessitates {\em online optimization}.
Overcoming these challenges constitutes the rest of this paper.
To conclude this section, we summarize the key notation used in this paper in Table~\ref{table:notation} for easy reference.

\begin{table}[t!]
\centering
\caption{Notation.}
\label{table:notation}
\vspace{-.1in}
\begin{tabular}{| c | l |}
\hline
$\mathcal{I/T}$ & The set of jobs$/$System timespan \\ \hline
$\tilde{t}_{i}/a_{i}$ & Completion time of job $i$ $/$Arrival time of job $i$ \\ \hline
$u_{i}(\cdot)/K_{i}$ & Job $i$'$s$ utility function$/$Number of samples in $i$ \\ \hline
$\mathcal{R/H}$ & The set of resource types$/$The set of machines \\ \hline
$E_{i}/F_{i}$ &\# of training epochs/Global batch size for job $i$ \\ \hline
$x_{i}$ & Admission decision variable to accept job $i$ or not \\ \hline
$C_{h}^{r}$ & Capacity of type-$r$ resource on server $h$ \\ \hline
$\alpha_{i}^{r}$ & Type-$r$ resource required by a worker in job $i$ \\ \hline
$\beta_{i}^{r}$ & \begin{tabular}[c]{@{}l@{}l@{}}Type-$r$ resource required by a PS in job $i$\end{tabular} \\ \hline
$w_{ih}[t]$ & Number of workers of job $i$ on server $h$ in $t$ \\ \hline
$s_{ih}[t]$ & Number of PSs of job $i$ on server $h$ in $t$ \\ \hline
$b_{i}(h,p)$ &\begin{tabular}[c]{@{}l@{}l@{}}Bandwidth consumed by a worker of job $i$, where \\ $b_{i}(h,p) =b_{i}^{(e)}$, if $h\neq p$ or $b_{i}^{(i)}$, otherwise. \end{tabular} \\ \hline
$\tau_{i}$ & Time to train a sample  for job $i$ \\ \hline
$g_{i}$ &\begin{tabular}[c]{@{}l@{}}Size of gradients and parameters for job $i$\end{tabular} \\ \hline
$\mathcal{W}_{i}[t]$ & \begin{tabular}[c]{@{}l@{}}Set of physical machines containing workers\\ for job $i$ in $t$\end{tabular} \\ \hline
$\mathcal{P}_{i}[t]$ & \begin{tabular}[c]{@{}l@{}}Machines containing parameter servers\\ for job $i$ in $t$\end{tabular} \\ \hline
$x_{\pi_i}$ &  \begin{tabular}[c]{@{}l@{}} Binary decision variable to select schedule $\pi$ \\for job $i$ or not\end{tabular} \\ \hline
$\tilde t_{\pi_i}$ & The completion time slot of job $i$ with schedule $\pi$\\ \hline
$w_{ht}^{\pi_i}$ & \# of workers on server $h$ in $t$ for job $i$ in schedule $\pi$\\ \hline
$s_{ht}^{\pi_i}$ & \# of PSs on server $h$ for schedule $\pi$ in $t$ \\ \hline
$\Pi_{i}$ & Set of all feasible schedules for job $i$\\ \hline
$\gamma_i$ & The ratio of worker number to PS number for job $i$\\ \hline
$\rho_h^r[t]$ &  Allocated type-$r$ resource on machine $h$ in time $t$ \\ \hline
$Q_h^r(\cdot)$ & Price function for type-$r$ resource on machine $h$\\ \hline
\end{tabular}
\vspace{-.1in}
\end{table}

\section{Online Scheduling Algorithm Design} \label{sec:alg}
In this section, we structure the key components of our online scheduling algorithm design for solving Problem~DMLRS into three steps from Sections~\ref{subsec:reformulation} to \ref{subsec:best_schedule}.
We state our main theoretical performance results in Section~\ref{subsec:performance_analysis}.

\subsection{Reformulation for Non-Deterministic Constraint (\ref{eqn_completion_time})} \label{subsec:reformulation}
The first challenge in solving Problem~DMLRS stems from the non-deterministic ``argmax'' structure in constraint (\ref{eqn_completion_time}).
To address this challenge, we let $\Pi_{i}$ be the set of all feasible schedules for job $i \in \mathcal{I}$ that satisfy constraints (\ref{eqn_min_workers_minibatch}), (\ref{eqn_max_workers_datachunk}).
Each schedule $\pi_{i} \in \Pi_{i}$ is defined by the numbers of workers $w_{ht}^{\pi_i}$ and parameter servers $s_{ht}^{\pi_i}$ allocated for job $i$ on machine $h$ in each time-slot $t$, i.e., $\pi_i \!\triangleq\! \{ w_{ht}^{\pi_i}, s_{ht}^{\pi_i}, \forall t \!\in\! \mathcal{T}, h \!\in\! \mathcal{H} \}$. 
Note that $w_{ht}^{\pi_i}$ and $s_{ht}^{\pi_i}$ are constants, not to be confused with decision variables $w_{ih}[t]$ and $s_{ih}[t]$. 
We define a binary variable $x_{\pi_i} \!\in\! \{0,1\}$ that is equal to 1 if job $i$ is admitted and scheduled under schedule $\pi_{i}$, or 0, otherwise. 
Clearly, due to the combinatorial nature, $|\Pi_{i}|$ is exponential.
We let $\tilde{t}_{\pi_i}$ denote job $i$'s completion time  under schedule $\pi_i$. 
Then, one can equivalently reformulate Problem DMLRS as:
\begin{align} 
& \textbf{R-DMLRS:} \nonumber \\
& \underset{\x}{\text{Maximize }}\sum_{i \in \mathcal{I}} \sum_{\pi_{i} \in \Pi_{i}} x_{\pi_i} u_{i}(\tilde{t}_{\pi_i} - a_{i}) \nonumber \\
\label{eqn_rsc_limit}& \text{subject to  } \sum_{i\in \mathcal{I}} \sum_{\pi_i \in \Gamma(t,h)} (\alpha_{i}^{r} w_{ht}^{\pi_i} + \beta_{i}^{r} s_{ht}^{\pi_i}) x_{\pi_i} \leq C_{h}^{r}, \nonumber \\
&\hspace{1in} \forall t \in \mathcal{T}, r\in \mathcal{R}, h \in \mathcal{H}, \\
\label{eqn_sch_select} &\hspace{.6in}  \sum_{\pi_i \in \Pi_{i}} x_{\pi_i} \leq 1, \quad \forall i \in \mathcal{I}, \\
&\hspace{.6in} x_{\pi_i} \in \{0,1\}, \quad \forall i \in \mathcal{I}, \pi_i \in \Pi_{i}, \nonumber
\end{align}
\noindent where we use $\Gamma(t,h)$ to represent the set of feasible schedules that use machine $h$ to deploy workers or parameter servers in time-slot $t$.
Constraint (\ref{eqn_rsc_limit}) guarantees that, in any time-slot $t$ and on any machine $h$, the total amount of consumed type-$r$ resources will not exceed the capacity limit $C_{h}^{r}$.
Constraint (\ref{eqn_sch_select}) ensures that, for each job $i$, at most one feasible schedule from $\Pi_{i}$ will be selected.
Note that Problem R-DMLRS is an integer linear program (ILP) and a feasible solution to Problem R-DMLRS has a corresponding feasible solution to the original Problem DMLRS, 
and vice versa.
Notice that the non-deterministic constraint (\ref{eqn_completion_time}) {\em no longer exists} in Problem~R-DMLRS.
Further, if relaxing binary $x_{\pi_i}$-variables to real-valued, Problem R-DMLRS is a linear program (LP).
However, it remains difficult to solve Problem~R-DMLRS since it has an {\em exponential} number of $x_{\pi_i}$-variables due to the combinatorial nature of feasible schedules.
We will address this challenge in the next subsection.

\subsection{An Online Primal-Dual Framework for R-DMLRS}
In what follows, we adopt a primal-dual online algorithmic framework to reduce the number of binary variables, which is an  effective approach to address this kind of challenge in the literature (see, e.g., \cite{Buchbinder09,Bao18:ML_INFOCOM}).
Note that, in the dual of Problem~R-DMLRS, the number of dual variables is polynomial.
Meanwhile, although there are an exponential number of constraints in the dual problem, one only needs to be concerned with the set of active (binding) constraints, which are easier to deal with.
To see this, we associate dual variables $p_{h}^{r}[t] \geq 0$, $\forall t\in \mathcal{T}$, $h \in \mathcal{H}$, $r \in \mathcal{R}$ and $\lambda_{i} > 0$, $i \in \mathcal{I}$, with (\ref{eqn_rsc_limit}) and (\ref{eqn_sch_select}), respectively.
Then, following the standard procedure of dualization (relaxing the integrality constraints), we obtain the following dual problem:
\begin{align} \label{eqn_reformulation_dual}
& \textbf{D-R-DMLRS:} \nonumber \\
& \underset{\vlambda,\p}{\text{Minimize }}  && \hspace{-.3in} \sum_{i \in \mathcal{I}} \lambda_{i} + \sum_{t\in\mathcal{T}}\sum_{h\in\mathcal{H}}\sum_{r\in\mathcal{R}}p_{h}^{r}[t] C_{h}^{r}  & \\
& \text{subject to} && \hspace{-.3in} \lambda_{i} \geq u_{i}(\tilde{t}_{\pi_i} - a_{i}) -  \sum_{t \in \mathcal{T}(\pi_i)} \sum_{h \in \mathcal{H}(\pi_i[t])} \sum_{r \in \mathcal{R}} (\alpha_{i}^{r} w_{ht}^{\pi_i} \nonumber\\
\label{eqn_dualconstr} & && \hspace{-.05in} + \beta_{i}^{r} s_{ht}^{\pi_i}) p_{h}^{r}[t], \quad \forall i \in \mathcal{I}, \pi_i \in \Pi_{i}, \\
 & && \hspace{-.3in} p_{h}^{r}[t] \geq 0, \quad \forall t \in \mathcal{T},  h \in \mathcal{H}, r \in \mathcal{R}, \nonumber\\
 & && \hspace{-.3in} \lambda_{i} \geq 0, \quad \forall i \in \mathcal{I}, \nonumber
\end{align}

\noindent where $\mathcal{T}(\pi_i)$ denotes the time-slots utilized by schedule $\pi_i$ and $\mathcal{H}(\pi_i[t])$ denotes the set of machines containing workers and/or parameter servers under $\pi_i$ in time-slot $t$.
Here, $p_{h}^{r}[t]$ can be viewed as the price for type-$r$ resource in time $t$.
Then, the right-hand-side (RHS) of (\ref{eqn_dualconstr}) can be interpreted as job utility minus overall resource cost of job $i$ using schedule $\pi_i$.
Thus $\lambda_{i}\geq 0$ can be viewed as the payoff of admitting job $i$ with $\pi_i$. 
Next, we examine the properties of Problem D-R-DMLRS.
To minimize (\ref{eqn_reformulation_dual}), we tend to reduce $\lambda_{i}$ and $p_{h}^{r}[t]$  until they hit zero.
However, as $\lambda_{i}$ and $p_{h}^{r}[t]$ decrease, the left-hand-side (LHS) and RHS of (\ref{eqn_dualconstr}) decreases and increases, respectively (the term
$u_{i}(\tilde{t}_{\pi_i}-a_{i})$
in the RHS of  (\ref{eqn_dualconstr}) is a constant given $\pi_i$).
Therefore, $\lambda_{i}$ will drop to a value $\lambda_{i}^{*}$, which is equal to maximum of the RHS of (\ref{eqn_dualconstr}) achieved by some schedule $\pi_{i}^{*}$ and dual price $p_{h}^{r*}[t]$, i.e.,
\begin{align*}
&\lambda_{i}^{*} \!=\! u_{i}(\tilde{t}_{\pi_i^*}\!-\!a_{i})
%
%
\!-\!\!\sum_{t \in \mathcal{T}(\pi_i^*)} \! \sum_{h \in \mathcal{H}(\pi_i^{*}[t])} \! \sum_{r \in \mathcal{R}} (\alpha_{i}^{r} w_{ht}^{\pi_i^{*}} \!+\! \beta_{i}^{r} s_{ht}^{\pi_i^{*}}) p_{h}^{r*}[t].
\end{align*}
This optimality structural insight implies that Problem~D-R-DMLRS is equivalent to finding an optimal schedule $\pi_i^{*}$ and dual price $p_{h}^{r*}[t]$ to maximize the RHS of (\ref{eqn_dualconstr}).
The above insights motivate the following primal-dual-based algorithm as shown in Algorithm~1. 

\medskip
\hrule 
\vspace{.03in}
\noindent {\bf Algorithm~1:} {Primal-Dual Online Resource Scheduling (PD-ORS).}
\vspace{.03in}
\hrule
\vspace{0.1in}
\noindent {\bf Initialization:}
\begin{enumerate} [topsep=1pt, itemsep=-.1ex, leftmargin=.2in]
\item[1.] Let $w_{ih}[t]=0$, $s_{ih}[t]=0$, $\forall i,t,h$. Let $\rho_{h}^{r}[t] = 0$, $\forall h,r,t$. Choose some appropriate initial values for $p_{h}^{r}[0]$.
\end{enumerate}

\noindent {\bf Main Loop:}
\begin{enumerate} [topsep=1pt, itemsep=-.1ex, leftmargin=.2in]
\item[2.] Upon the arrival of job $i$, determine a schedule $\pi_{i}^{*}$ to maximize the RHS of (\ref{eqn_dualconstr}) and its corresponding payoff $\lambda_{i}$ using {\bf Algorithm~2} ({\em to be specified}).
\item[3.] If $\lambda_{i} > 0$, set $x_{i} = 1$. 
Set $w_{ih}[t]$ and $s_{ih}[t]$ according to schedule $\pi_{i}^{*}$, $\forall t \in \mathcal{T}(\pi_{i}^{*})$, $h\in\mathcal{H}(\pi_{i}^{*}[t])$. 
\begin{itemize}
\item Update $\rho_{h}^{r}[t] \leftarrow \rho_{h}^{r}[t] + \alpha_{i}^{r} w_{ih}[t] + \beta_{i}^{r}s_{ih}[t]$, $\forall t \in \mathcal{T}(\pi_{i}^{*})$, $h\in\mathcal{H}(\pi_{i}^{*}[t])$, $r \in \mathcal{R}$.
\item Update $p_{h}^{r}[t] = Q_{h}^{r}(\rho_{h}^{r}[t])$, $\forall t \in \mathcal{T}(\pi_{i}^{*})$, $h\in\mathcal{H}(\pi_{i}^{*}[t])$, $r \in \mathcal{R}$.
Schedule job $i$ based on $\pi_{i}^{*}$ and go to Step 2.
\end{itemize}
\item[4.] If $\lambda_{i} \leq 0$, set $x_{i} = 0$ and reject job $i$ and go to Step 2.
\end{enumerate}
\smallskip
\hrule
\medskip

The intuition of Algorithm~1 is as follows:
By the complementary slackness condition of the Karush-Kuhn-Tucker (KKT) conditions~\cite{Bazaraa_Sherali_Shetty_93:NLP}, the primal constraint (\ref{eqn_sch_select}) must be tight when dual variable $\lambda_{i} > 0$, which implies that $x_{i} =1$ (Step 3) in Problem DMLRS.
Otherwise, if $\lambda_{i} = 0$, then the RHS of (\ref{eqn_dualconstr}) is non-positive, meaning the utility is low compared to the cost of resource consumption under schedule $\pi_{i}^{*}$.
Therefore, we should reject job $i$ ($x_{i} = 0$ in Step 4).
However, in order for the PD-ORS algorithm to work, two challenging components need to be specified: the schedule $\pi_{i}^{*}$ and how to update the cost function $Q_{h}^{r}(\cdot)$\footnote{
Here, we note that the ``cost function'' $Q_{h}^{r}(\cdot)$ is interpreted from servers' perspective rather than jobs' perspective.
Specifically, higher cost means servers allocated more resources to jobs, which implies higher utility for the jobs since jobs receive more resources from servers.}
In what follows, we will first focus on designing $Q_{h}^{r}(\cdot)$ and defer the challenging problem of finding $\pi_{i}^{*}$ to Section~\ref{subsec:best_schedule}.
For the design of $Q_{h}^{r}(\cdot)$, 
consider the following choice of $Q_{h}^{r}(\cdot)$:
\begin{align} \label{eqn_cost_update_fn}
Q_{h}^{r}(\rho_{h}^{r}[t]) = L(U^{r}/L)^{\frac{\rho_{h}^{r}[t]}{C_{h}^{r}}},
\end{align}
where constants $U^{r}$, $\forall r$, and $L$ are defined as follows:
\begin{align}
\label{eqn_Ur} & \!\!\!\!\!\! U^{r} \!\triangleq\!\max_{i\in\mathcal{I}}\frac{u_{i}(\ceil{\frac{E_iK_i}{F_i}(\tau_{i}+2g_{i}\gamma_i/(b_i^{(i)}F_i))}-a_{i})}{\alpha_{i}^{r}+\beta_{i}^{r}},\forall r\in\mathcal{R},\!\!\!\!\!\! \\
\label{eqn_L} & \!\!\!\!\!\! L \triangleq \min_{i\in\mathcal{I}}\frac{1/(2\mu) u_{i}(T-a_{i})}{\sum_{r\in\mathcal{R}}\ceil{E_{i}K_i(\tau_{i}+2g_{i}\gamma_i/(b_i^{(e)}F_i)}(\alpha_{i}^{r}+\beta_{i}^{r})}. \!\!\!\!\!\!
\end{align}
The scaling factor $\mu$ in the definition of $L$ satisfies as follows:
\begin{align*}
    \frac{1}{\mu}\leq \frac{\ceil{E_{i}K_{i}(\tau_{i}+2g_{i}\gamma_i/(b_{i}^{(e)}F_i))}\sum_{r\in\mathcal{R}}(\alpha_{i}^{r}+\beta_{i}^{r})}{T\sum_{h\in\mathcal{H}}\sum_{r\in\mathcal{R}}C_{h}^{r}},\forall i\in\mathcal{I}.
\end{align*}
Here, $U^{r}$ is the maximum unit-resource job utility to deploy workers and parameter servers with type-$r$ resource .
Here, $u_{i}(\ceil{\frac{E_iK_i}{F_i}(\tau_{i}+2g_{i}\gamma_i/(b_{i}^{(i)}F_i))}-a_{i})$ is the largest utility job $i$ can achieve by using the maximum number of {\em co-located} workers and parameter servers (hence communicating rate is $b_{i}^{(i)}$) at all times during all $E_{i}$ epochs, so that $\ceil{\frac{E_iK_i}{F_i}(\tau_{i}\!+\!2g_{i}\gamma_i/(b_{i}^{(i)}F_i))}\!-\!a_{i}$ is the earliest possible job completion time.
Similarly, $L$ represents the minimum unit-time unit-resource job utility among all jobs, with $u_{i}(T-a_{i})$ being the smallest utility for job $i$, and workers and parameter servers communicate at small external rate $b_{i}^{(e)}$.
We use $\rho_h^r[t]$ to denote the allocated amount of type-$r$ resource to machine $h$ for (future) time slot $t$.
The intuition behind the $Q_{h}^{r}(\cdot)$ function is as follows:
i) At $t\!=\!0$, $\rho_{h}^{r}[0]\!=\!0, \forall h\in\mathcal{H},r\in\mathcal{R}$. 
Hence, the price $p_{h}^{r}[0]=L$ is the lowest, $\forall h,r$, and any job can be admitted;
ii) As allocated resources increases, the price increases exponentially fast to reject early coming jobs with low utility and to reserve resources for later arrived jobs with higher utility; 
iii) When some type-$r$ resource is exhausted, i.e., $\rho_{h}^{r}[t] = C_{h}^{r}, \exists r\in\mathcal{R}$, $Q_{h}^{r}[C_{h}^{r}]=U^{r}$ and no job that requires type-$r$ resources will be admitted since the $U^{r}$ is the highest price.
As will be shown later, this price function leads to a logarithmically scaling competitive ratio in online algorithm design.
Note that computing the price function in Algorithm~1 requires the information of constants $U^r$, $L$, which can usually be estimated empirically based historical data.

Here, we point out a few interesting insights on the design choices of the cost function in Eq.~\eqref{eqn_cost_update_fn}.
Note that $U^r$ and $L$ are defined in Eqns~(\ref{eqn_Ur}) and (\ref{eqn_L}), respectively.
Here, we intentionally choose $U^r$ to be dependent on $r$ and $L$ to be independent on $L$ due to the following reasons:

First, the rationality of choosing an upper bound $U^r$ that varies with different resource types is to ensure that when some type-$r$ resource is exhausted, no more jobs that require type-$r$ resource should be allocated.
In other words, when the allocated amount of type-$r$ resource reaches the capacity of physical machine $h$, i.e., $\rho_h^r[t]=C_n^r, \exists r\in\mathcal{R}$, the price $p_h^r=Q_h^r(C_h^r)$ should reach the upper bound $U^r$, indicating that any job that requires type-$r$ resource will not be allocated to $h$.
However, jobs that do not require type-$r$ resource should still be able to be scheduled on machine $h$.
For example,  jobs that do not require GPU can still be placed on a machine with no available GPUs.

Second, the reason that we choose the lower bound $L$ to be independent of any resource type $r$ is to yield a larger ratio of $\frac{U^r}{L}$.
The larger the ratio of $\frac{U^r}{L}$ is, the greater the price will be.
Intuitively, the ratio $\frac{U^r}{L}$ can be interpreted as the scheduling ``uncertainty,'' which increases as the ratio gets larger, implying the price function reacts to the accumulative resource consumption more aggressively.
Thus, choosing $L$ to be independent of resource type $r$ allows the price function $Q_h^r(\cdot)$ to react more aggressively to the accumulative allocated resource amount.
We note that one can also choose the lower bound to be dependent on resource type $r$ by replacing $L$ with $L^r$.
By doing so, the log-scaling theoretical competitive ratio in Theorem~\ref{thm_Competitive_1} still holds and the proof in Appendix~\ref{appdx:thm_Competitive} only needs to be slightly updated  with the new notation $L^r$.
However, the empirical performance of using $L^r$ as lower bound is worse since the price function reacts less aggressively to the accumulative allocated resources.

\begin{figure*}[t!]
\centering
\subfigure[{$|\mathcal{P}_{i}[t]|\ne 1$.}]{%
\label{fig:IntExt1}%
\includegraphics[width=0.24\textwidth]{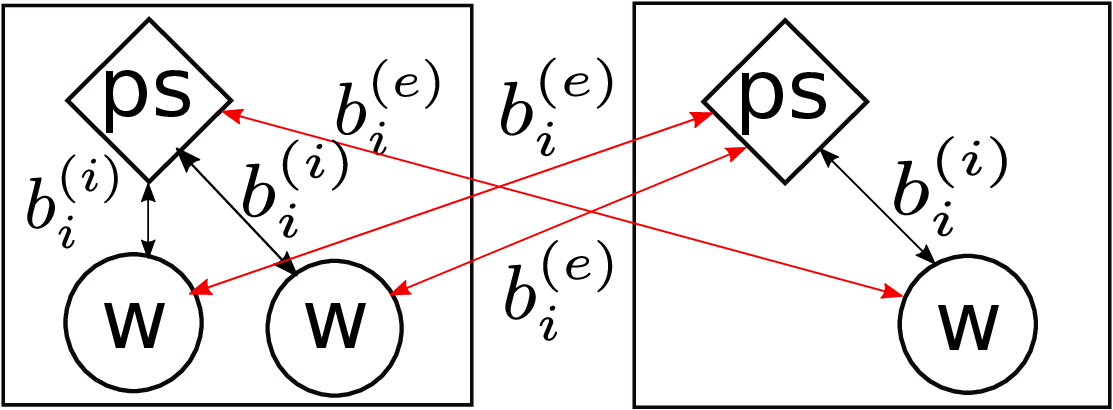}}%
\hspace{6pt}%
\subfigure[{$|\mathcal{P}_{i}[t]|=1,|\mathcal{W}_{i}[t]|\ne 1$.}]{%
\label{fig:IntExt2}%
\includegraphics[width=0.24\textwidth]{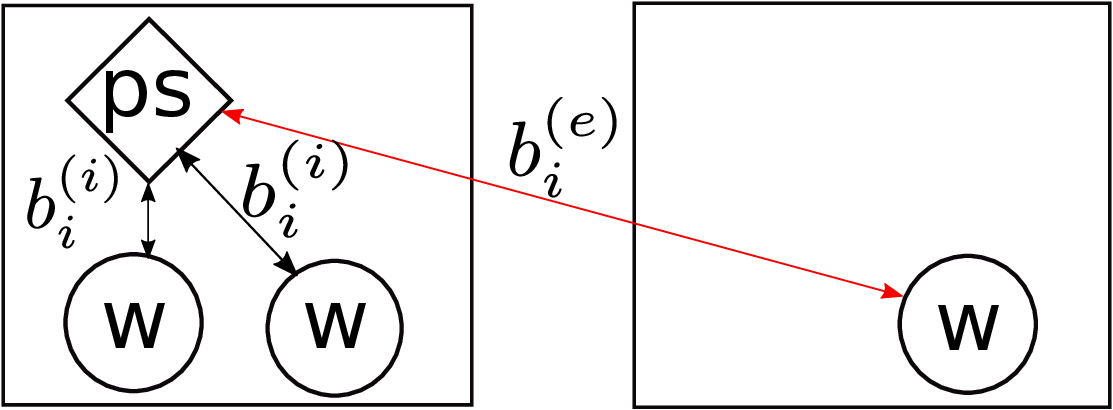}}%
\hspace{6pt}%
\subfigure[{$|\mathcal{P}_{i}[t]|\!\!=\!\!|\mathcal{W}_{i}[t]|\!\!=\!\!1,\mathcal{P}_{i}[t]\!\!\ne\!\!\mathcal{W}_{i}[t]$.}]{%
\label{fig:IntExt3}%
\includegraphics[width=0.24\textwidth]{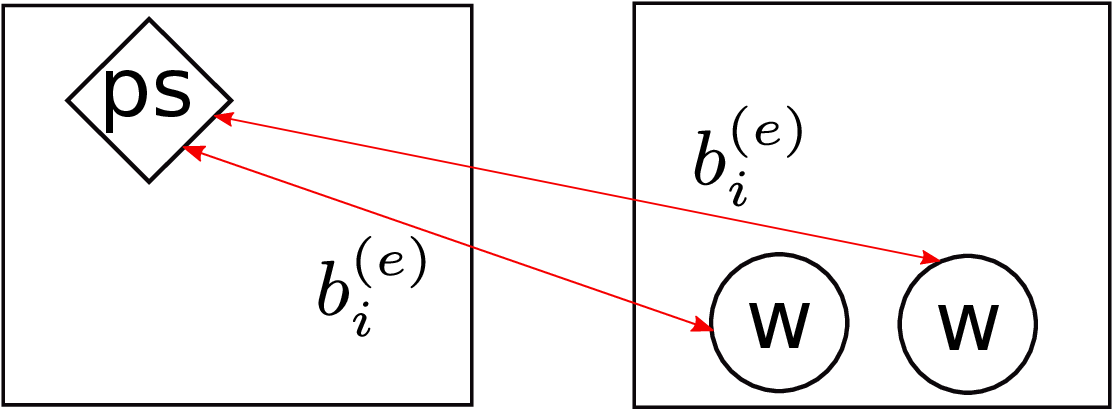}}%
\hspace{6pt}%
\subfigure[{$|\mathcal{P}_{i}[t]|\!\!=\!\!|\mathcal{W}_{i}[t]|\!\!=\!\!1,\mathcal{P}_{i}[t]\!\!=\!\!\mathcal{W}_{i}[t]$.}]{%
\label{fig:IntExt4}%
\includegraphics[width=0.24\textwidth]{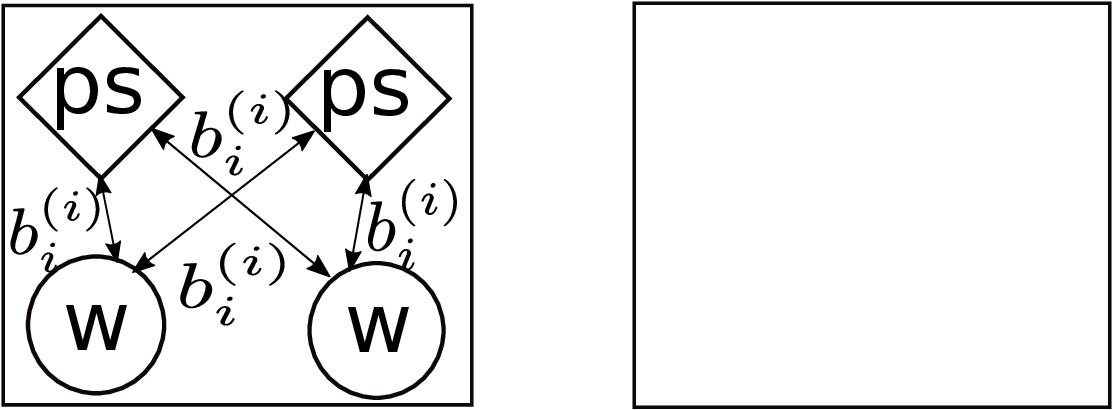}}%
\vspace{-.08in}
\caption{Values of $\frac{2g_{i}}{\min_{p\in\mathcal{P}_{i}[t],h\in\mathcal{W}_i[t]}b_{i}(h,p)}$ under various settings of $\mathcal{P}_{i}[t]$ and $\mathcal{W}_{i}[t]$. 
Here,
$\frac{2g_{i}}{\min_{p\in\mathcal{P}_{i}[t],h\in\mathcal{W}_i[t]}b_{i}(h,p)} = 2\frac{g_{i}}{b_{i}^{(i)}}$ if and only if in (d).}%
\label{fig:IntExt}%
\vspace{-.15in}
\end{figure*}
\vspace{-.1in}
\subsection{Determining Schedule $\pi_{i}^{*}$ in Step 2 of Algorithm 1} \label{subsec:best_schedule}

Now, consider the problem of finding a schedule $\pi_{i}^{*}$ in Step 2 of Algorithm~1 to maximize the RHS of (\ref{eqn_dualconstr}).
First, we note that {\em any} schedule for job $i$ has a unique completion time $\tilde{t}_{i}$, including the maximizer $\pi_{i}^{*}$ for (\ref{eqn_dualconstr}). 
Hence, the problem of finding the maximum RHS of (\ref{eqn_dualconstr}) can be decomposed as:
\begin{align} \label{eqn_ProbFindSch}
\!\!\!\underset{{\tilde{t}_{i}}}{\text{Max}} \left\{
\begin{array}{l}
\underset{{ \w,\s }}{\text{Max}} \,\, u_{i}(\tilde{t}_{i} \!-\! a_{i}) \!-\! \sum\limits_{t\in\mathcal{T}}\sum\limits_{h\in \mathcal{H}}\sum\limits_{r\in\mathcal{R}} p_{h}^{r}[t] \\
\hspace{1.1in} \times (\alpha_{i}^{r} w_{ih}[t] + \beta_{i}^{r} s_{ih}[t]) \\
\text{s.t.  }\  \alpha_{i}^{r} w_{ih}[t] + \beta_{i}^{r} s_{ih}[t] \leq \hat{C}_{h}^{r}[t], \\
\hspace{1.1in} \forall t\in \mathcal{T},r\in\mathcal{R},h\in\mathcal{H}, \\
\hspace{.3in} \text{Constraints (\ref{eqn_min_workers_minibatch})(\ref{eqn_max_workers_datachunk})(\ref{yz}) for } x_{i}=1,
\end{array}
\right\}, \!\!\!\!\!
\end{align}

\noindent where $\hat{C}_{h}^{r}[t] \triangleq C_{h}^{r} - \rho_{h}^{r}[t]$.
Note that in the inner problem, $u_{i}(\tilde{t}_{i} \!-\! a_{i})$ is a constant for any given $\tilde{t}_{i}$.
Thus, the inner problem can be simplified as:
\begin{align}
\label{eqn_obj_given_t} \underset{\w,\s}{\text{Minimize }} & \sum_{t \in [a_{i}, \tilde{t}_{i}]} \sum_{h\in\mathcal{H}}\sum_{r\in\mathcal{R}} p_{h}^{r}[t](\alpha_{i}^{r} w_{ih}[t] + \beta_{i}^{r} s_{ih}[t]) \\
\label{eqn_workload_coupling} \text{subject to } & \sum_{t\in[a_{i},\tilde{t}_{i}]} \sum_{h\in\mathcal{H}} \frac{w_{ih}[t]}{\tau_{i}+\frac{\gamma_i}{F_i} \frac{2g_{i}}{\min_{p\in \mathcal{P}_{i}[t],h'\in\mathcal{W}_i[t]}b_{i}(h',p)}}\geq V_{i}, \!\!\! \\
\label{eqn_rsrc_limit} & \alpha_{i}^{r} w_{ih}[t] \!+\! \beta_{i}^{r} s_{ih}[t] \!\leq\! \hat{C}_{h}^{r}[t], \forall r,h, \forall t \!\in\! [a_{i},\tilde{t}_{i}], \!\!\!\\
&\text{Constraint (\ref{eqn_max_workers_datachunk}) for all } t\in[a_{i},\tilde{t}_{i}],\nonumber
\end{align}
where $V_{i} \triangleq E_{i} K_{i}$ represents the total training workload
(total count of samples trained, where a sample is counted $E_i$ times if trained for $E_i$ times).
Note that  in Problem~(\ref{eqn_obj_given_t}), the only coupling constraint is (\ref{eqn_workload_coupling}).
This observation inspires a dynamic programming approach to solve Problem (\ref{eqn_obj_given_t}).
Consider the problem below if training workload at time $t$ is known (denoted as $V_{i}[t]$):
\begin{align}
\label{eqn_obj_given_t2} \underset{w_{ih}[t],s_{ih}[t],\forall h}{\text{Minimize }} & \hspace{.1in} \sum_{h\in\mathcal{H}}\sum_{r\in\mathcal{R}} p_{h}^{r}[t](\alpha_{i}^{r} w_{ih}[t] + \beta_{i}^{r} s_{ih}[t]) \\
\label{eqn_nondeterministic} \text{subject to } & \hspace{.1in} \sum_{h\in\mathcal{H}} \frac{w_{ih}[t]}{\tau_{i}+\frac{\gamma_i}{F_i} \frac{2g_{i}}{\min_{p\in \mathcal{P}_{i}[t],h'\in\mathcal{W}_i[t]}b_{i}(h',p)}}\geq V_{i}[t], \\
& \hspace{.1in} \text{Constraints (\ref{eqn_max_workers_datachunk})(\ref{eqn_rsrc_limit}) for the given } t.\nonumber
\end{align}
Let $\Theta(\tilde{t}_{i},V_{i})$ and $\theta(t,V_{i}[t])$ denote the optimal values of Problems~(\ref{eqn_obj_given_t}) and (\ref{eqn_obj_given_t2}), respectively.
Then, Problem~(\ref{eqn_obj_given_t}) is equivalent to the following dynamic program:
\begin{align} \label{eqn_2d_dp}
\Theta(\tilde{t}_{i},V_{i}) = \min_{v\in[0,V_{i}]} \left\{ \theta(\tilde{t}_{i},v) + \Theta(\tilde{t}_{i}-1,V_{i}-v) \right\}.
\end{align}
We find the optimal workload $v$ to be completed in time slot $\tilde{t}_{i}$ by enumerating it from 0 to $E_iK_i$, and the remaining workload $E_iK_i-v$ will be carried out in time span $[a_i,\tilde{t}_{i}-1]$.
The optimal workload would be the schedule with minimum costs, i.e., the objective function of Problem~(\ref{eqn_obj_given_t2}) is minimum.
Then we proceed to the next time slot $\tilde{t}_{i}-1$ with the workload $E_iK_i-v$, which is the same as finding the optimal schedule and cost as the last time slot $\tilde{t}_{i}$ except in a smaller scale.
Then, by enumerating all $\tilde{t}_{i} \in [a_{i}, T]$ and solving the dynamic program $\Theta(\tilde{t}_{i},V_{i})$ in (\ref{eqn_2d_dp}) for every choice of $\tilde{t}_{i}$, we can solve Problem (\ref{eqn_ProbFindSch}) and determine the optimal schedule $\pi_{i}^{*}$.
We summarize this procedure in Algorithms~2 and 3:
\medskip
\hrule 
\vspace{.03in}
\noindent {\bf Algorithm~2:} Determine $\pi_{i}^{*}$ in Step 2 of Algorithm 1.
\vspace{.03in}
\hrule
\vspace{0.1in}
\noindent {\bf Initialization:}
\begin{enumerate} [topsep=1pt, itemsep=-.1ex, leftmargin=.2in]
\item[1.] Let $\tilde{t}_{i}\!=\!a_{i}$. Let $\lambda_{i} \!=\! 0$, $\pi_{i}^{*} \!=\! \varnothing$, $w_{ih}[t] \!=\! s_{ih}[t] \!=\! 0$, $\forall i,t,h$.
\end{enumerate}

\noindent {\bf Main Loop:}
\begin{enumerate} [topsep=1pt, itemsep=-.1ex, leftmargin=.2in]
\item[2.] Compute $\Theta(\tilde{t}_{i},V_{i})$ in (\ref{eqn_2d_dp}) using {\bf Algorithm 3}. 
Denote the resulted schedule as $\pi_{i}$.
Let $\lambda'_{i} = u_{i}(\tilde{t}_{i}-a_{i}) - \Theta(\tilde{t}_{i},V_{i})$.
If $\lambda'_{i} > \lambda_{i}$, let $\lambda_{i} \leftarrow \lambda'_{i}$ and $\pi_{i}^{*} \leftarrow \pi_{i}$.
\item[3.] Let $\tilde{t}_{i} \leftarrow \tilde{t}_{i} + 1$. If $\tilde{t}_{i} > T$, stop; otherwise, go to Step 2.
\end{enumerate}
\smallskip
\hrule
\medskip

\medskip
\hrule 
\vspace{.03in}
\noindent {\bf Algorithm~3:} Solving $\Theta(\tilde{t}_{i},V_{i})$ by Dynamic Programming.
\vspace{.03in}
\hrule
\vspace{0.1in}
\noindent {\bf Initialization:}
\begin{enumerate} [topsep=1pt, itemsep=-.1ex, leftmargin=.2in]
\item[1.] Let {\em cost-min} = $\infty$, $\pi_{i} = \varnothing$, and $v = 0$.
\end{enumerate}

\noindent {\bf Main Loop:}
\begin{enumerate} [topsep=1pt, itemsep=-.1ex, leftmargin=.2in]
\item[2.] Compute $\theta(\tilde{t}_{i}, v)$ using {\bf Algorithm~4} ({\em to be specified}). 
Denote the resulted cost and schedule as {\em cost-v} and $\hat{\pi}_{i}$.

\item[3.] Compute $\Theta(\tilde{t}_{i}-1,V_{i}-v)$ by calling {\bf Algorithm~3} itself.
Denote the resulted cost and schedule as {\em cost-rest} and $\tilde{\pi}_{i}$.
\item[4.] If $\text{{\em cost-min}} > \text{{\em cost-v}} + \text{{\em cost-rest}}$ then $\text{{\em cost-min}} = \text{{\em cost-v}} + \text{{\em cost-rest}}$ and let $\pi_{i} \leftarrow \hat{\pi}_{i} \cup \tilde{\pi}_{i}$.
\item[5.] Let $v \leftarrow v+1$. If $v>V_{i}$ stop; otherwise go to Step 2.
\end{enumerate}
\smallskip
\hrule
\medskip
In Algorithm~3, however, how to compute $\theta(t,v)$ in Step 2 (i.e., Problem~(\ref{eqn_obj_given_t2})) is yet to be specified.
A challenge in solving (\ref{eqn_obj_given_t2}) is the {\em non-deterministic} constraint in (\ref{eqn_nondeterministic}), where $b_{i}(h,p)$ can be either $b_{i}^{(i)}$ or $b_{i}^{(e)}$.
Therefore, we need to handle both cases.
To this end, we observe the following basic fact about $\frac{2g_{i}}{\min_{p \in \mathcal{P}_{i}[t],h\in\mathcal{W}_i[t]}b_{i}(h,p)}$ (also illustrated in Fig.~\ref{fig:IntExt}) as stated in Fact~\ref{fact1}.
We omit the proof of this fact due to its simplicity, which is illustrated in Fig.~\ref{fig:IntExt}.

\begin{fact} \label{fact1}
The function $(2g_{i}/\min_{p\in\mathcal{P}_{i}[t],h\in\mathcal{W}_i[t]}b_{i}(h,p)) = 2g_{i}/b_{i}^{(i)}$ if and only if $|\mathcal{P}_{i}[t]| = |\mathcal{W}_{i}[t]| =1$ and $\mathcal{P}_{i}[t] = \mathcal{W}_{i}[t]$; otherwise, $(2g_{i}/\min_{p\in\mathcal{P}_{i}[t],h\in\mathcal{W}_i[t]}b_{i}(h,p)) = 2g_{i}/b_{i}^{(e)}$.
\end{fact}

With Fact~\ref{fact1}, we now consider the following two cases:

\smallskip
{\em Case 1): $b_{i}^{(i)}$ (Internal Communication):}
In Case 1), Fact~\ref{fact1} implies that Problem~(\ref{eqn_obj_given_t2}) reduces to a {\em single-machine} problem (i.e., discarding $\sum_{h\in\mathcal{H}}\{\cdot\}$ in (\ref{eqn_obj_given_t2}) and (\ref{eqn_nondeterministic})).
Note further that if we temporarily ignore the workload-coupling constraint (\ref{eqn_nondeterministic}) and use the worker-PS ratio in (\ref{eqn_WorkerPSRatio}), Problem~(\ref{eqn_obj_given_t2}) can be {\em decoupled} across resources and simplified as:
\begin{equation}\label{linear}
\sum_{r\in\mathcal{R}}
\begin{Bmatrix*}[l]
&\hspace{-.1in} \text{Min} &\hspace{-.1in}p_{h}^{r}[t]s_{ih}[t](\alpha_{i}^{r} \gamma_i + \beta_{i}^{r} )\\
&\hspace{-.1in} \text{s.t.} &\hspace{-.1in} s_{ih}[t] (\alpha_{i}^{r} \gamma_i+ \beta_{i}^{r})\leq \hat{C}_{h}^{r}[t], \\
&&\hspace{-.1in} \text{Constraint }(\ref{eqn_max_workers_datachunk}) \text{ for given } r,h
\end{Bmatrix*}, \!\!\!\!
\end{equation}
where each summand in (\ref{linear}) is an integer linear program (ILP) having a trivial solution $w_{ih}[t]=s_{ih}[t]=0$, $\forall h \in \mathcal{H}$.
However, $w_{ih}[t]=0$, $\forall h\in\mathcal{H}$, clearly violates the workload constraint ($\ref{eqn_nondeterministic}$). 
Thus, when (\ref{linear}) is optimal, there should be {\em exactly one} machine $h' \in\mathcal{H}$ with $w_{ih'}[t]\geq 1$ and {\em exactly one} machine $h'' \in\mathcal{H}$ with $s_{ih''}[t] \geq 1$. 
This observation shows that the optimal solution of (\ref{eqn_obj_given_t2}) {\em tends to favor} $|\mathcal{P}_{i}[t]| = |\mathcal{W}_{i}[t]|=1$ if the workload constraint (\ref{eqn_nondeterministic}) is not binding.

Notice that the workload constraint (\ref{eqn_nondeterministic}) in Case 1) becomes $\gamma_is_{ih}[t] \geq V_{i}[t](\tau_{i} + \frac{2g_{i}\gamma_i}{b_{i}^{(i)}F_i})$. 
This implies the following simple solution: 
We can first sort each physical machine $h$ according to $\sum_{r\in\mathcal{R}}p_h^r[t](\alpha_i^r\gamma_i+\beta_i^r)$ and calculate the minimum number of $s_{ih}[t]=V_{i}[t](\tau_{i} + \frac{2g_{i}\gamma_i}{b_{i}^{(i)}F_i})/\gamma_i$ from the workload constraint.
The last step is to check if the machine satisfy the resource capacity constraint (\ref{eqn_rsrc_limit}) and constraint (\ref{eqn_max_workers_datachunk}).
If so, we return the schedule $(w_{ih}[t],s_{ih}[t])$ and the corresponding cost value.

\smallskip
{\em Case 2): $b_{i}^{(e)}$ (External Communication):} 
For those settings that do not satisfy $|\mathcal{P}_{i}[t]| = |\mathcal{W}_{i}[t]| =1$ and $\mathcal{P}_{i}[t] = \mathcal{W}_{i}[t]$, Fact~\ref{fact1} indicates that parameter servers and workers are communicating at external rate $b_{i}^{(e)}$.
In this case, the workload constraint (\ref{eqn_nondeterministic}) simply becomes: $\sum_{h \in \mathcal{H}} w_{ih}[t] \geq V_{i}[t](\tau_{i} + \frac{2g_{i}\gamma_i}{b_{i}^{(e)}F_i})$.
Then, we can rewrite Problem~(\ref{eqn_obj_given_t2}) as:
\begin{align}
\label{eqn_obj_given_t3} \underset{w_{ih}[t],s_{ih}[t],\forall h}{\text{Minimize }} & \hspace{.1in} \sum_{h\in\mathcal{H}} p_{h}^{w}[t] w_{ih}[t] + p_{h}^{s}[t] s_{ih}[t] \\
\label{eqn_packing} \text{subject to } & \hspace{.1in} \alpha_{i}^{r} w_{ih}[t] + \beta_{i}^{r} s_{ih}[t] \leq \hat{C}_{h}^{r}[t], \,\, \forall h,r,\\
\label{eqn_packing2}  &\hspace{.1in} \sum\nolimits_{h\in\mathcal{H}} w_{ih}[t] \leq F_i,\\
\label{eqn_cover1} &\hspace{.1in} \sum\nolimits_{h\in\mathcal{H}} w_{ih}[t] \geq V_{i}[t] \bigg( \tau_{i} + \frac{2g_{i}\gamma_i}{b_{i}^{(e)}F_i} \bigg), 
\end{align}
where $p_{h}^{w}[t] \triangleq \sum_{r\in\mathcal{R}} p_{h}^{r}[t] \alpha_{i}^{r}$ and $p_{h}^{s}[t] \triangleq \sum_{r\in\mathcal{R}} p_{h}^{r}[t] \beta_{i}^{r}$ denote the aggregated prices of all resources of allocating workers and PSs on machine $h$ in time $t$, respectively.

Unfortunately, Problem~(\ref{eqn_obj_given_t3}) is a highly challenging integer programming problem with generalized packing and cover type constraints (i.e., integer variables rather than 0-1 variables) in (\ref{eqn_packing})--(\ref{eqn_cover1}), respectively,
which is clearly NP-Hard.
Also, it is well-known that there are {\em no} polynomial time approximation schemes (PTAS) even for the basic set-cover and bin-packing problems unless P = NP~\cite{Hochbaum97}.
In what follows, we will pursue an instance-dependent constant ratio approximation scheme to solve Problem~(\ref{eqn_obj_given_t3}) in this paper. 
To this end, we propose a randomized rounding scheme:
First, we solve the linear programming relaxation of Problem~(\ref{eqn_obj_given_t3}).
Let $\{ \bar{w}_{ih}[t], \bar{s}_{ih}[t], \forall h,t \}$ be the fractional optimal solution.
We let $\delta\in(0,1]$ be a parameter.
Let $G_\delta$ be a constant {(the notation $G_\delta$ signifies that $G_{\delta}$ is dependent on $\delta$) to be defined later, and let $w'_{ih}[t] \!=\! G_\delta\bar{w}_{ih}[t], s'_{ih}[t] \!=\! G_\delta\bar{s}_{ih}[t]$, $\forall h,t$.
Then, we randomly round $\{ w'_{ih}[t], s'_{ih}[t], \forall h,t \}$ to obtain an integer solution as follows: 
\begin{align}
\label{eqn_rounding_w} w_{ih}[t] &= \begin{cases}
\lceil w'_{ih}[t] \rceil, & \hspace{-.1in} \text{with probability } w'_{ih}[t] - \lfloor w'_{ih}[t] \rfloor,\\
\lfloor w'_{ih}[t] \rfloor, & \hspace{-.1in} \text{with probability } \lceil w'_{ih}[t] \rceil - w'_{ih}[t],
\end{cases} \!\!\!\\
\label{eqn_rounding_s} s_{ih}[t] &= \begin{cases}
\lceil s'_{ih}[t] \rceil, & \hspace{-.1in} \text{with probability } s'_{ih}[t] - \lfloor s'_{ih}[t] \rfloor,\\
\lfloor s'_{ih}[t] \rfloor, & \hspace{-.1in} \text{with probability } \lceil s'_{ih}[t] \rceil - s'_{ih}[t].
\end{cases}
\end{align}
We will later prove in Theorem~\ref{thm_Alg4_1} (when $0<G_\delta\leq 1)$ and Theorem~\ref{thm_Alg4_2} (when $G_\delta>1$) that the approximation ratio of this randomized rounding scheme in (\ref{eqn_rounding_w})-(\ref{eqn_rounding_s}) enjoys a ratio that is independent on the problem size. 

\smallskip
Lastly, summarizing results in Cases 1) -- 2) yields the following approximation algorithm for solving Problem~(\ref{eqn_obj_given_t2}):

\medskip
\hrule 
\vspace{.03in}
\noindent {\bf Algorithm~4:} Solving $\theta(t,v)$ (i.e., Problem (\ref{eqn_obj_given_t2})).
\vspace{.03in}
\hrule
\vspace{0.1in}
\noindent {\bf Initialization:}
\begin{enumerate} [topsep=1pt, itemsep=-.1ex, leftmargin=.2in]
\item[1.] Let $w_{ih}[t] \!=\! s_{ih}[t] \!=\! 0$, $\forall h$. Let $h\!=\!1$. 
Pick some $\delta\in(0,1]$.
Let $G_\delta$ be defined as in Eq.~(\ref{eqn_Thm3_G}) or Eqn~(\ref{eqn_Thm4_G}).
Let $D \!=\! \lceil v(\tau_{i} + 2g_{i}\gamma_i/(b_{i}^{(i)}F_i)) \rceil$. 
Let $h^* = \varnothing$. Let {\em cost-min}$=\infty$.
Choose some integer $S \geq 1$.
Let $iter\leftarrow1$.
\end{enumerate}

\noindent {\bf Handling Internal Communication:}
\begin{enumerate} [topsep=1pt, itemsep=-.1ex, leftmargin=.2in]
\item[2.] Sort machines in $\mathcal{H}$ according to $\sum_{r\in\mathcal{R}} p_h^r[t](\alpha_i^r\gamma_i+\beta_i^r)$ in non-decreasing order into $h_1,h_2,...,h_H$.

\item[3.] Calculate the minimum number of $s_{ih}[t]=V_i[t] \Big(\tau_i+\frac{2g_i\gamma_i}{b_i^{(i)}F_i} \Big)/\gamma_i$.

\item[4.] If Constraint (\ref{eqn_max_workers_datachunk}) is not satisfied, go to Step 7.

\item[5.] If Constraint (\ref{eqn_packing}) is not satisfied, go to Step 7.

\item[6.] Return {\em cost-min}$\sum_{r\in\mathcal{R}} p_h^r[t]s_{ih}[t](\alpha_i^r\gamma_i+\beta_i^r)$ and $h^* = h$.

\item[7.] Let $h\leftarrow h+1$. If $h>H$, stop; otherwise, go to Step 2.

\end{enumerate}

\noindent {\bf Handling External Communication:}
\begin{enumerate} [topsep=1pt, itemsep=-.1ex, leftmargin=.2in]
\item[8.] Solve the linear programming relaxation of Problem~(\ref{eqn_obj_given_t3}).
Let $\{ \bar{w}_{ih}[t], \bar{s}_{ih}[t], \forall h,t \}$ be the fractional optimal solution.

\item[9.] Let $w'_{ih}[t] = G_\delta\bar{w}_{ih}[t], s'_{ih}[t] = G_\delta\bar{s}_{ih}[t]$, $\forall h,t$.

\item[10.] Generate an integer solution $\{ w_{ih}[t], s_{ih}[t], \forall h,t \}$ following the randomized rounding scheme in (\ref{eqn_rounding_w})--(\ref{eqn_rounding_s}).


\item[11.] If $\{ w_{ih}[t], s_{ih}[t], \forall h,t \}$ is infeasible or $iter< S$, then $iter\leftarrow iter+1$, 
go to Step 10.
\end{enumerate}

\noindent {\bf Final Step:}
\begin{enumerate} [topsep=1pt, itemsep=-.1ex, leftmargin=.2in]
\item[12.] Compare the solutions between internal and external cases. Pick the one with the lowest cost among them and return the cost and the corresponding schedule $\{ w_{ih}[t], s_{ih}[t], \forall h,t \}$.
\end{enumerate}
\smallskip
\hrule
\medskip

In the internal communication part of Algorithm~4, we first sort the machines and then check each machine one by one (Step 2).
We calculate the minimum number of $s_{ih}[t]$ needed to satisfy the learning workload demand $D$ (Step 3).
If Constraint (\ref{eqn_max_workers_datachunk}) is satisfied (Step 4), we further check the resource capacity constraint (\ref{eqn_rsrc_limit}) (Step 5).
If we detect a machine with all above constraints satisfied, we return the cost and schedule accordingly (Step 6).
After exploring one machine, we move on to the next one as long as it is not the last machine  (Step 7).
The external communication part is based on LP relaxation (Step 8) and randomized rounding (Step 9-12).
Note that the randomized rounding will find at most $S$ integer feasible solutions (Step 12).
Finally, we choose the lowest cost among the solutions from the internal and external communication parts (Step 13).

\subsection{Theoretical Performance Analysis} \label{subsec:performance_analysis}

We now examine the competitive ratio of our PD-ORS algorithm.
Note that the key component in PD-ORS is our proposed randomized rounding scheme in (\ref{eqn_rounding_w})--(\ref{eqn_rounding_s}), which is in turn the foundation of Algorithm~1.
Thus, we first prove the following results regarding the randomized rounding algorithm.
Consider an integer program with generalized cover/packing constraints: $\min \{ \c^{\top} \x : \A\x \geq \a, \B\x \leq \b, \x \in \mathbb{Z}_{+}^{n} \}$, where $\A \in \mathbb{R}_{+}^{m\times n}$, $\B \in \mathbb{R}_{+}^{r\times n}$, $\a \in \mathbb{R}_{+}^{m}$, $\b \in \mathbb{R}_{+}^{r}$, and $\c \in \mathbb{R}_{+}^{n}$.
Let $\bar{\x}$ be a fractional optimal solution.
Consider the randomized rounding scheme: Let $\x' = G_\delta\bar{\x}$ for some $G_\delta$ (to be specified).
Randomly round $\x'$ to $\hat{\x} \in \mathbb{Z}_{+}^{n}$ as:
$\hat{x}_{j} = \lceil x'_{j} \rceil$ w.p. $x'_{j} - \lfloor x'_{j} \rfloor$ and $\hat{x}_{j} = \lfloor x'_{j} \rfloor$ o.w.
Note that in the rounding process, if $G_\delta > 1$ ($G_\delta \in (0,1]$), the packing (cover) constraint is prone to be violated and the cover (packing) constraint is easier to be satisfied.
Hence, depending on which constraint is more preferred to be feasible, we consider two cases with respect to $G_\delta$.

1) $0<G_\delta\leq1$: We have the following approximation result (see proof in Appendix~\ref{appdx:lem_ILP1}):

\begin{lem}[Rounding] \label{lem_ILP}
Let $W_{a} \!\!\triangleq\!\! \min \{ a_{i}/[\A]_{ij}\!: [\A]_{ij} \!\!>\!\! 0 \}$ and $W_{b} \!\!\triangleq\!\! \min \{ b_{i}/[\B]_{ij}\!: [\B]_{ij} \!\!>\!\! 0 \}$.
Let $\delta \in (0,1]$ be a given constant and define $G_\delta$ as:
\begin{align*}
G_\delta\triangleq1+ \frac{3\ln(3r/\delta)}{2W_b} - \sqrt{ \bigg(\frac{3\ln(3r/\delta)}{2W_b} \bigg)^{2} + \frac{3\ln(3r/\delta)}{W_b} }.
\end{align*}
Then, with probability greater than $1-\delta$, $\hat{\x}$ achieves a cost at most $\frac{3G_\delta}{\delta}$ times the cost of $\bar{\x}$.
Meanwhile, $\hat{\x}$ satisfies
$\mathrm{Pr}\Big\{ (\A\hat{\x})_{i} \leq a_{i} (1- (\frac{2}{G_\delta W_a}\ln(\frac{3m}{\delta}))^{\frac{1}{2}})G_\delta, \exists i \Big\} \leq \frac{\delta}{3m}$.
\end{lem}

\begin{rem}[Discussions on Feasibility]
{\em
An insightful remark of Lemma~\ref{lem_ILP} is in order.
Note that, theoretically, the expression $1- (\frac{2}{G_\delta W_a}\ln(\frac{3m}{\delta}))^{\frac{1}{2}}$ in Lemma~\ref{lem_ILP} could become negative.
In this case, the last probabilistic inequality in Lemma~\ref{lem_ILP} trivially holds and is not meaningful in characterizing the feasibility, even though the inequality remains valid.
In order for the probabilistic statement to be meaningful in characterizing the feasibility of the integer linear program, we solve for $\delta$ by enforcing $1- (\frac{2}{G_\delta W_a}\ln(\frac{3m}{\delta}))^{\frac{1}{2}} > 0$, which yields $\delta\geq3m /e^\frac{G_\delta W_a}{2}$.
On the other hand, we prefer $\delta$ to be small since it bounds the feasibility violation probability and approximation ratio achievability in Lemma~\ref{lem_ILP}.
Hence, the smaller the value of $3m /e^\frac{G_\delta W_a}{2}$, the less restrictive the condition $\delta\geq3m /e^\frac{G_\delta W_a}{2}$ is.

\begin{figure}[t]
    \centering
    \includegraphics[width=0.9\columnwidth]{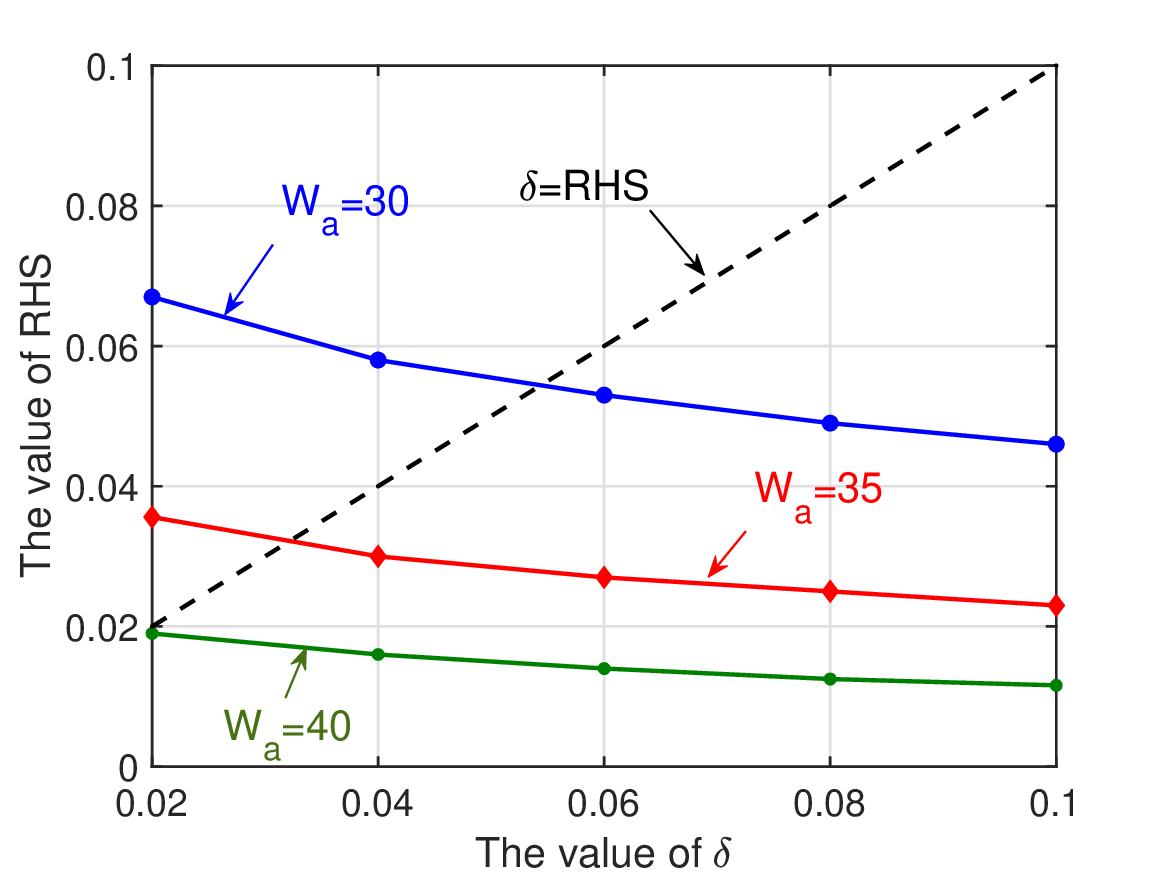}
    \caption{The feasibility study.}
    \label{fig:Gdelta}
\end{figure}

To gain a deeper understanding on how restrictive the condition $\delta\geq 3m /e^\frac{G_\delta W_a}{2}$ is, we conduct a case study and the results are illustrated in Fig.~\ref{fig:Gdelta}.
Here, we let RHS $\triangleq3m /e^\frac{G_\delta W_a}{2}$ for convenience.
We vary $\delta$ from $0.02$ to $0.1$.
Clearly, the left-hand-side (LHS) of the condition is the $45^{\circ}$ straight line.
In order for the condition to hold, the curve of RHS should fall under this straight line.
Based on typical computing cluster parameters, we set $W_b$ to $15$, and set $r\triangleq RH+1$ to $401$ ($R=4, H=100$).
We can see from Fig.~\ref{fig:Gdelta} that as $W_a$ increases, the curve of RHS crosses the dashed line of LHS at a smaller $\delta$-value.
This means that, the larger the value of $W_a$, the easier for $1- (\frac{2}{G_\delta W_a}\ln(\frac{3m}{\delta}))^{\frac{1}{2}}$ to become positive.
Hence, the probabilistic feasibility characterization in Lemma~\ref{lem_ILP} is useful for typical system parameters in practice.
}
\end{rem}

2) $G_\delta>1$: We have the following approximation result (see proof in Appendix~\ref{appdx:lem_ILP2}):

\begin{lem}[Rounding] \label{lem_ILP2}
Let $W_{a}\triangleq \min \{ a_{i}/[\A]_{ij}: [\A]_{ij}> 0 \}$ and $W_{b}\triangleq \min \{ b_{i}/[\B]_{ij}: [\B]_{ij}> 0 \}$.
Let $\delta \in (0,1]$ be a given constant and define $G_\delta$ as:
\begin{align*}
G_\delta \triangleq 1 + \frac{\ln(3m/\delta)}{W_a} + \sqrt{ \bigg(\frac{\ln(3m/\delta)}{W_a} \bigg)^{2} + \frac{2\ln(3m/\delta)}{W_a} }.
\end{align*}
Then, with probability greater than $1-\delta$, $\hat{\x}$ achieves a cost at most $\frac{3G_\delta}{\delta}$ times the cost of $\bar{\x}$.
Meanwhile, $\hat{\x}$ satisfies  
$\mathrm{Pr} \Big\{ (\B\hat{\x})_{i} > b_{i} (1+ (\frac{3}{G_\delta W_b}\ln(\frac{3r}{\delta}))^{\frac{1}{2}})G_\delta, \exists i \Big\} \leq \frac{\delta}{3r}$.
\end{lem}

Several important remarks for Lemmas~\ref{lem_ILP} and \ref{lem_ILP2} are summarized as follows:
\begin{list}{\labelitemi}{\leftmargin=1em \itemindent=-0.0em \itemsep=.2em}

\item[i)] Note that  Alg.~4 is a randomized algorithm.
Therefore, its performance is also characterized probabilistically, and $\delta$ is used for such probabilistic characterization.
Here, it means that with probability $1 - \delta$, one achieves an approximation ratio at most $\frac{3G_\delta}{\delta}$ and a probabilistic feasibility guarantee as stated in Lemma~\ref{lem_ILP} when $0<G_\delta\leq 1$ and Lemma~\ref{lem_ILP2} when $G_\delta>1$.
In other words, the statements mean that the probability of getting a better approximation ratio is smaller under randomized rounding (i.e., better approximation ratio $\Rightarrow$ smaller $\frac{3G_\delta}{\delta}$ $\Rightarrow$ larger $\delta$ $\Rightarrow$ smaller probability $1-\delta$).
That is, the trade-off between the approximation ratio value and its achieving probability is quantified by $\delta$.
A larger $\delta$ implies a smaller approximation ratio, but the probability of obtaining a feasible solution of this ratio is also smaller (i.e., more rounds of rounding needed).
Interestingly, for $\delta=1$, Lemmas~\ref{lem_ILP} and \ref{lem_ILP2} indicate that there is still {\em non-zero} probability to achieve an approximation ratio not exceeding $3G_\delta$.

\item[ii)] Note that if we pick $G_\delta\in(0,1]$,  the approximation ratio $\frac{3G_\delta}{\delta}$ decreases (the smaller the approximation ratio, the better) as $\delta$ increases based on Eqn.~(\ref{eqn_L_def_1}) in Appendix~\ref{appdx:lem_ILP1}.
However, the growth rate of $G_\delta$ is slower compared to that of $\delta$ due to the $\log$ operator.
On the other hand, if we pick $G_\delta>1$, $\frac{3G_\delta}{\delta}$ decreases as $\delta$ increases according to Eqn.~(\ref{eqn_L_def}) in Appendix~\ref{appdx:lem_ILP2}.
Therefore, the approximation ratio is ultimately controlled by parameter $\delta$.
Also, the theoretical approximation ratio $\frac{3G_\delta}{\delta}$ is conservative. 
Our numerical studies show that the approximation ratio performance in reality is much smaller than $\frac{3G_\delta}{\delta}$.

\item[iii)] The probabilistic guarantee of the cover constraint ($\A\x \geq \a$) when $0<G_\delta\leq 1$ and packing constraint ($\B\x \leq \b$) when $G_\delta>1$, is unavoidable and due to the fundamental hardness of satisfying both cover and packing constraints, which are of {\em conflicting} nature: Any strategy trying to better satisfy the packing constraints (multiplying a $G_\delta$-factor with $G_\delta\in (0,1]$) may increase the probability of violating the cover constraints, and the probability of violating the packing constraints may be increased otherwise.
However, the probabilistic bound here is for worst case and may be pessimistic.

\item[iv)] The results in Lemmas~\ref{lem_ILP} and \ref{lem_ILP2} are in fact applicable for general ILP with mixed cover/packing constraints.
Hence, the results and their insights in Lemmas~\ref{lem_ILP} and \ref{lem_ILP2}  could be of independent theoretical interest.
\end{list}

By specializing Lemma~\ref{lem_ILP} and Lemma~\ref{lem_ILP2} with parameters in Problem~(\ref{eqn_obj_given_t3}), we have the following approximation results for Algorithm 4.
The first result corresponds to the case where the feasibility of the resource constraint (packing) is more favored, i.e., $0<G_\delta\leq1$:

\begin{thm}[Approximation Performance of Alg.~4 When Resource Constraint Feasibility is Favored]\label{thm_Alg4_1}
Let $W_{1}\!\triangleq\! V_{i}[t] \big( \tau_{i} + \frac{2g_{i}\gamma_i}{b_{i}^{(e)}F_i}\big)$, $W_{2}\!\triangleq\!\min\{ F_i, \hat{C}_{h}^{r}[t]/\alpha_{i}^{r}, \hat{C}_{h}^{r}[t]/\beta_{i}^{r}, \forall r,h \}$, and  
$\delta \in (0,1]$. Define $G_\delta$ as:
\begin{align} \label{eqn_Thm3_G}
&\nonumber G_\delta\!\triangleq\! 1 \!+\! \frac{3\ln(3(RH\!+\!1)/\delta)}{2W_2}\\
 &\!-\! \sqrt{ \bigg(\frac{3\ln(3(RH+1)/\delta)}{2W_2} \bigg)^{2} \!+\! \frac{3\ln(3(RH\!+\!1)/\delta)}{W_2} }.
\end{align}
Then, with probability greater than $1\!-\!\delta$, Algorithm~4 obtains a schedule $\{ w_{ih}[t], s_{ih}[t], \forall t,h \}$ that has an approximation ratio at most $\frac{3G_\delta}{\delta}$ with $\mathrm{Pr}\{ \text{LHS(\ref{eqn_cover1})} \!\leq\! W_1(1- (\frac{2}{G_\delta W_1}\ln(\frac{3}{\delta}))^{\frac{1}{2}})G_\delta, \exists i\} \leq \frac{\delta}{3}$.
\end{thm}

The next result corresponds to the case where the feasibility of the workload constraint (cover) is more favored, i.e., $G_\delta > 1$:
\begin{thm}[Approximation Performance of Alg.~4  When Workload Constraint Feasibility is Favored]\label{thm_Alg4_2}
Let $W_{1}\triangleq V_{i}[t] \big( \tau_{i} + \frac{2g_{i}\gamma_i}{b_{i}^{(e)}F_i}\big)$, $W_{2} \triangleq \min\{ F_i, \hat{C}_{h}^{r}[t]/\alpha_{i}^{r}, \hat{C}_{h}^{r}[t]/\beta_{i}^{r}, \forall r,h \}$,  
and $\delta \in (0,1]$. Define $G_\delta$ as:
\begin{align} \label{eqn_Thm4_G}
&G_\delta \triangleq 1 + \frac{\ln(3/\delta)}{W_1} + \sqrt{ \bigg(\frac{\ln(3/\delta)}{W_1} \bigg)^{2} + \frac{2\ln(3/\delta)}{W_1} }.
\end{align}
Then, with probability greater than $1-\delta$, Algorithm~4 obtains a schedule $\{ w_{ih}[t], s_{ih}[t], \forall t,h \}$ that has an approximation ratio at most $\frac{3G_\delta}{\delta}$ with $\mathrm{Pr}\{ \text{LHS(\ref{eqn_packing})} > \hat{C}_h^r[t] G_\delta(1+(\frac{3}{G_\delta W_2}\ln(\frac{3(HR+1)}{\delta}))^{\frac{1}{2}})\} \leq \frac{\delta}{3(HR+1)}$.
\end{thm}

Note that Eqn.~(\ref{eqn_packing2}) is guaranteed in practice since the number of samples is typically far more than the number of workers.
The competitive ratio of our online algorithm is the worst-case upper bound of the ratio of the overall utility of admitted jobs devided by the offline optimal solution of Problem~DMLRS to the total utility
of admitted jobs achieved by Algorithm~1 in the overall time horizon.
Theorems~\ref{thm_Alg4_1} and \ref{thm_Alg4_2} follow directly from Lemmas~\ref{lem_ILP} and \ref{lem_ILP2}, respectively, and we omit the proof here for brevity.
Based on these results, we can establish the overall competitive ratio for Algorithm~1 as follows.
\begin{thm}[Competitive Ratio of Alg~1 when $0\!<\!G_\delta\!\leq\! 1$]\label{thm_Competitive_1} 
Let $\delta$, $G_\delta$ and $W_1$ be as defined in Theorem~\ref{thm_Alg4_1}.
Let $U^{r}$ and $L$ be as defined in (\ref{eqn_Ur}) and (\ref{eqn_L}), respectively.
Then, with probability greater than $(1-(\delta/3)^S)^{TK_iE_i}$, PD-ORS in Algorithm~1 returns a feasible solution that is $\frac{6G_\delta}{\delta} \max_{r\in\mathcal{R}} (1,\ln\frac{U^{r}}{L})$--competitive.
\end{thm}

\begin{thm}[Competitive Ratio of Alg~1 when $G_\delta>1$]\label{thm_Competitive_2} 
Let $\delta$, $G_\delta$ and $W_2$ be as defined in Theorem~\ref{thm_Alg4_2}.
Let $U^{r}$ and $L$ be as defined in (\ref{eqn_Ur}) and (\ref{eqn_L}), respectively.
Then, with probability greater than $(1-(\delta/3(HR+1))^S)^{TK_iE_i}$, PD-ORS in Algorithm~1 returns a feasible solution that is $\frac{6G_\delta}{\delta} \max_{r\in\mathcal{R}} (1,\ln\frac{U^{r}}{L})$--competitive.
\end{thm}

It is worth pointing out that in Theorems~\ref{thm_Competitive_1} and \ref{thm_Competitive_2}, the feasibility achieving probability values, i.e., $(1-(\delta/3)^S)^{TK_iE_i}$ and $(1-(\delta/3(HR+1))^S)^{TK_iE_i}$, can controlled by choosing appropriate values of $\delta$ and $S$ (i.e., rounds of rounding) to offset the impact of total number of DP iterations $TK_i E_i$.
The smaller $\delta$ and the larger $S$ are, the higher the feasibility achieving probability.
Theorems~\ref{thm_Competitive_1} and \ref{thm_Competitive_2} can be proved by weak duality and the approximation results in Theorems~\ref{thm_Alg4_1} and \ref{thm_Alg4_2}.
We provide a proof in Appendix~\ref{appdx:thm_Competitive}.

\begin{thm}[Polynomial Running Time]\label{thm_Time}
By combining Algorithms~1--4, 
the time complexity of PD-ORS is
$O(\sum_{i=1}^{|\mathcal{I}|} T K_i^2 E_{i}^{2} (H^{3} + S) )$,
which is polynomial.
\end{thm}

\begin{proof}
When solving $\theta(t,v)$ using Algorithm~4, it takes $O(H\log H)$ iterations to sort machines in internal communication case under each time slot $t$ and looping all machines to calculate the minimum number $s_{ih}[t]$ takes $O(H)$.
Thus, it takes $(H\log H)$ time for the internal communication part in Algorithm~4.
For the external communication part in Algorithm~4, solving the LP relaxation of Problem~\eqref{eqn_obj_given_t3} can be approximately bounded $O(H^3)$ if we use a polynomial time LP solver (e.g., Vaidya's algorithm~\cite{Vaidya87:LP}).
According to Algorithm~4, the rounding time is proportional to $S$.
Hence, the running time for the external communications part is upper bounded by $O(|\mathcal{H}|^3 + S)$.
Combining the discussions above, the running time complexity of Algorithm~4 is $O(H\log H + H^3 + S) = O(H^3 + S)$.
Moreover, the number of states $(t,v)$ is $O(TK_iE_i)$ in the dynamic programming (DP) for each job $i$, and the time complexity of executing DP is $O(K_iE_i)$.
Thus, the time complexity is $O(TK_i^2E_i^2)$ in DP.
In Algorithm~1, the number of steps in the main loop is equal to the number of jobs.
Therefore, the overall running time complexity can be computed as $O(\sum_{i=1}^{|\mathcal{I}|} T K_i^2 E_{i}^{2} (H^{3} + S) )$.
\end{proof}

\section{Numerical Results}
\label{sec:numerical}

In this section, we conduct simulation studies to evaluate the efficacy of our proposed PD-ORS algorithm.
We test an ML system with jobs parameters generated uniformly at random from the following intervals: $E_{i} \in [50,200]$, $K_{i} \in [20000,500000]$,
$g_{i} \in [30,575]$ MB, $\tau_{i} \in [10^{-5},10^{-4}]$ 
time slots, $\gamma_i\in[1,10]$,
$F_i\in[1,200]$. 
We consider four types of resources: GPU, CPU, memory, and storage.
For fair comparisons, following similar settings in~\cite{Li14:OSDI}~\cite{Chilimbi14:OSDI}~\cite{Sun17:Dorm}, we set resource demand of each worker as follows: 0--4 GPUs, 1--10 vCPUs, 2--32 GB memory, and 5--10GB storage.
We set resource demand of each parameter server as follows: 1--10 vCPUs, 2--32GB memory and 5-10GB storage.
The resource capacity of each physical machine is set roughly 18 times of the resource demands of a worker/PS following EC2 C5n instances~\cite{AmazonEC2Instances}. 
We set the job arrival pattern according to the Google Cluster data~\cite{Reiss12:googlecluster}, but with normalized job arrival rates in alternating time-slots as follows: the arrival rates are $1/3$ and $2/3$ in odd and even time-slots, respectively.
For fair comparisons, we adopt the Sigmoid utility function ~\cite{Huang15:CORA,Bao18:ML_INFOCOM}: $u_i(t-a_i)=\frac{\theta_1}{1+e^{\theta_2(t-a_i-\theta_3)}}$, where $\theta_1\in[1,100]$ indicates the priority of job $i$, $\theta_2$ indicates how critical the time is for the job $i$, and $\theta_3\in[1,15]$ is the estimated target completion time.
We set $\theta_2=0$ for time-insensitive jobs (10\% of jobs), $\theta_2\in[0.01,1]$ for time-sensitive jobs (55\% of jobs) and $\theta_2\in[4,6]$ for time-critical jobs (35\% of jobs).

We first compare our PD-ORS algorithm with three baseline job scheduling policies: (1) {\em FIFO} in Hadoop and Spark~\cite{Zaharia10:USENIX}, where the jobs are processed in the order of their arrival times. In our setting, the fixed number of workers (parameter servers) is between 1 to 30, 
(2) {\em Dominant Resource Fairness Scheduling} (DRF) in Yarn~\cite{Vavilapalli13:Yarn} and Mesos~\cite{Hindman11:Mesos}, where the jobs are scheduled based on their dominant resource share in the cloud to achieve its max-min fairness. 
The number of workers and parameter servers are computed and allocated dynamically, 
and (3) {\em Dorm}~\cite{Sun17:Dorm}, where the numbers of workers (parameter servers) are computed and placed by an MILP resource utilization maximization problem with fairness and adjustment overhead constraints.
Workers and parameter servers are placed in a round-robin fashion on available machines in Baselines (1) and (2).
The comparison results are shown in Figs.~\ref{fig:jobFixed} and~\ref{fig:machineFixed}.
In Fig.~\ref{fig:jobFixed}, we set $T=20$ and $I=50$, while in Fig.~\ref{fig:machineFixed}, we set $T=20$ and $H=100$.
We can see that PD-ORS significantly outperforms other policies and the gains in total utility becomes more pronounced as the numbers of jobs and machines increase. 

\begin{figure*}[t!]
\begin{minipage}{0.3\textwidth}
\includegraphics[width=1\textwidth]{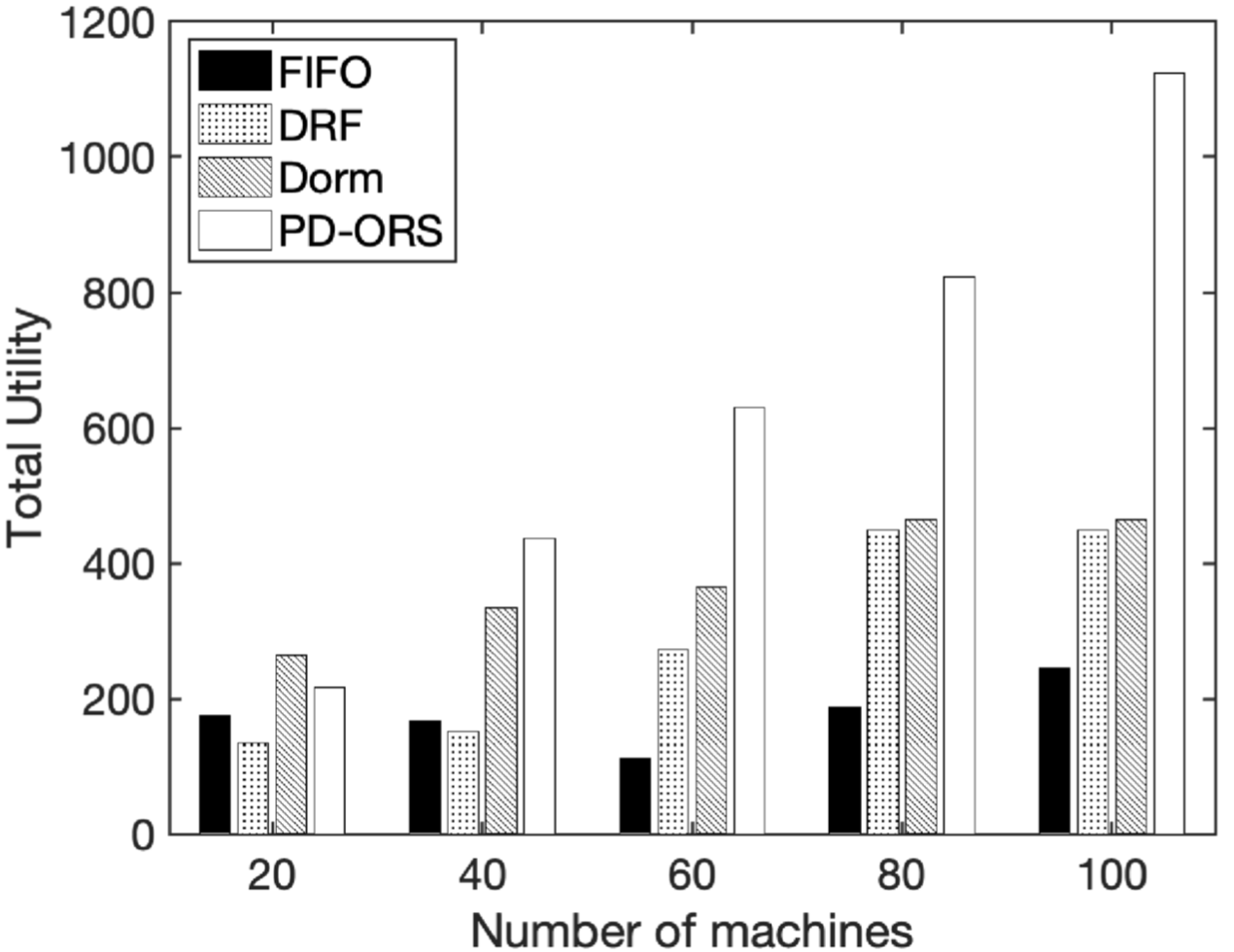}
\caption{Total utility with  increasing number of machines (synthetic data).}\label{fig:jobFixed}
\end{minipage}
\hspace{0.06\columnwidth}%
\begin{minipage}{0.3\textwidth}
\includegraphics[width=1\textwidth]{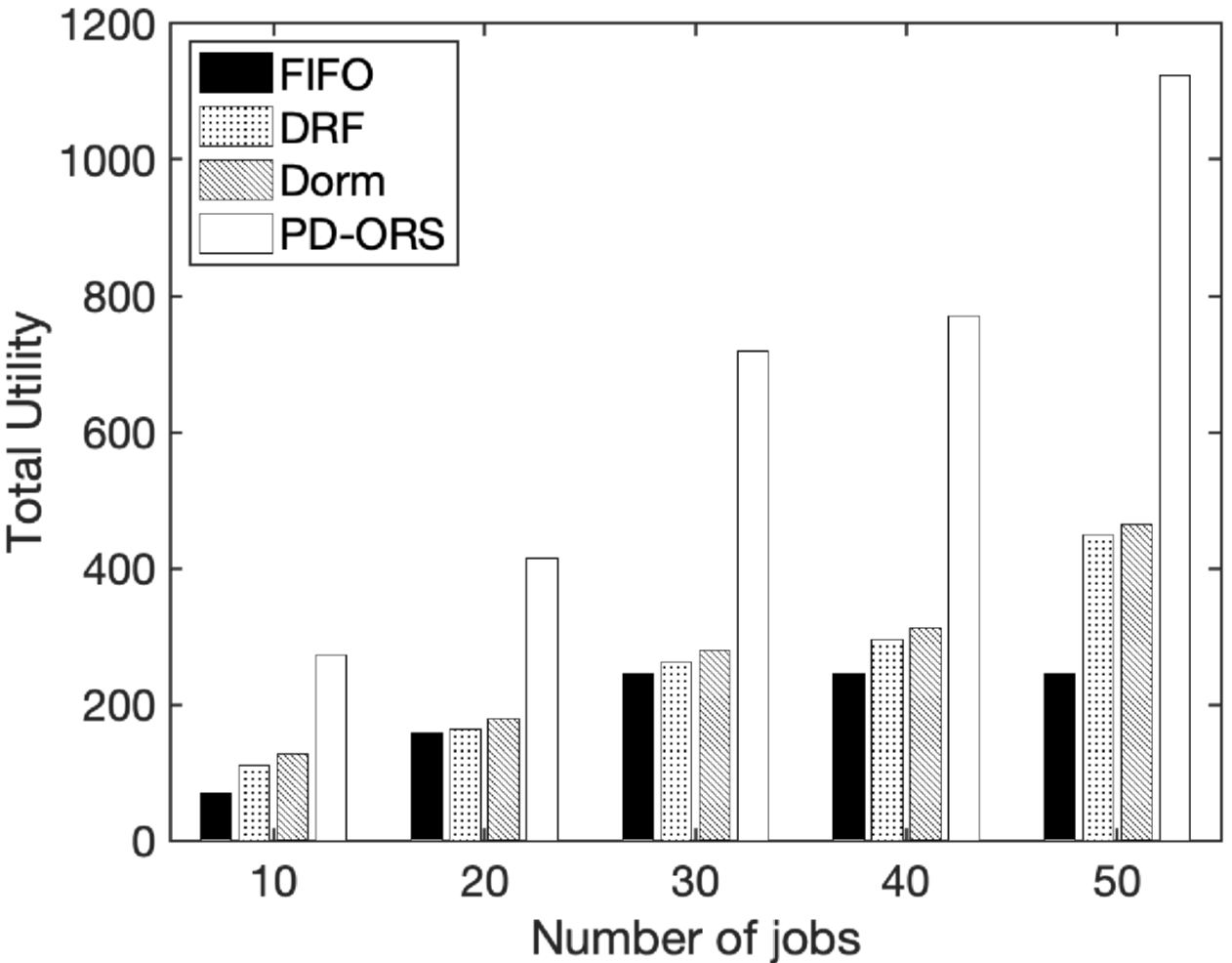}
\caption{Total utility with increasing number of jobs (synthetic data).} \label{fig:machineFixed}
\end{minipage}
\hspace{0.06\columnwidth}%
\begin{minipage}{0.3\textwidth}
\vspace{.05in}
\includegraphics[width=1.08\textwidth]{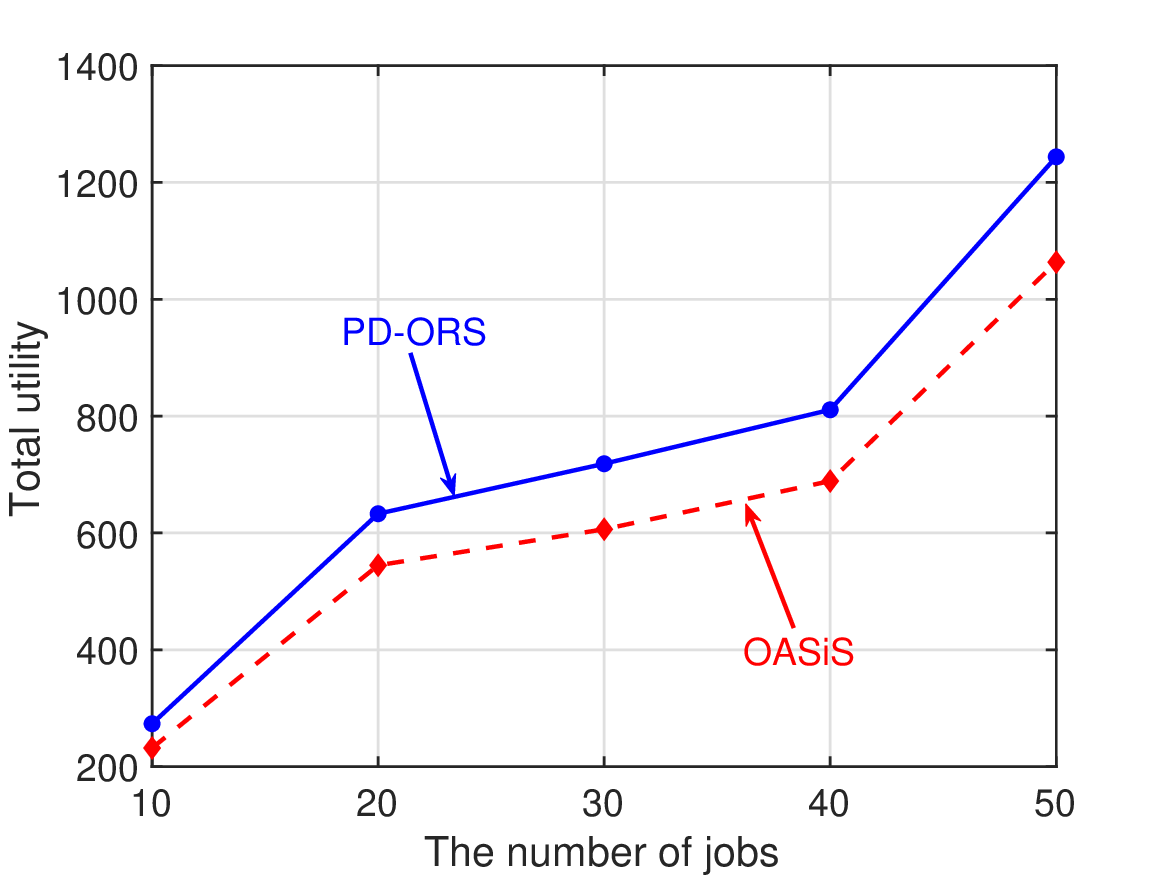}
\caption{Utility comparison between PD-ORS and OASiS with increasing number of jobs.} \label{fig:UtilityComp_1}
\end{minipage}
\end{figure*}

\begin{figure*}[t!]
\begin{minipage}{0.3\textwidth}
\vspace{-.2in}
   \includegraphics[width=1.1\textwidth]{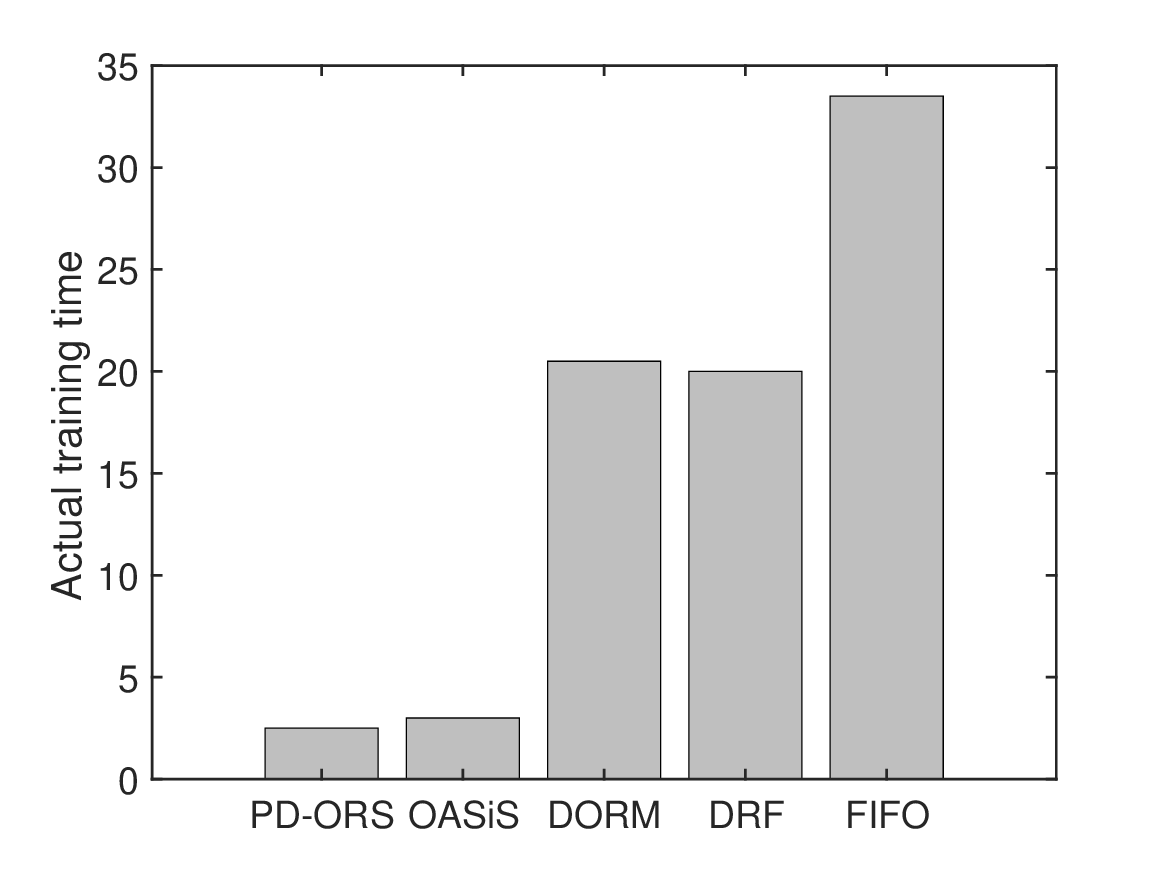}
    \caption{Median of actual training time comparison.} \label{fig:actual_time}
\end{minipage}
\hspace{0.06\columnwidth}%
\begin{minipage}{0.3\textwidth}
\vspace{-.3in}
	\includegraphics[width=1.11\textwidth]{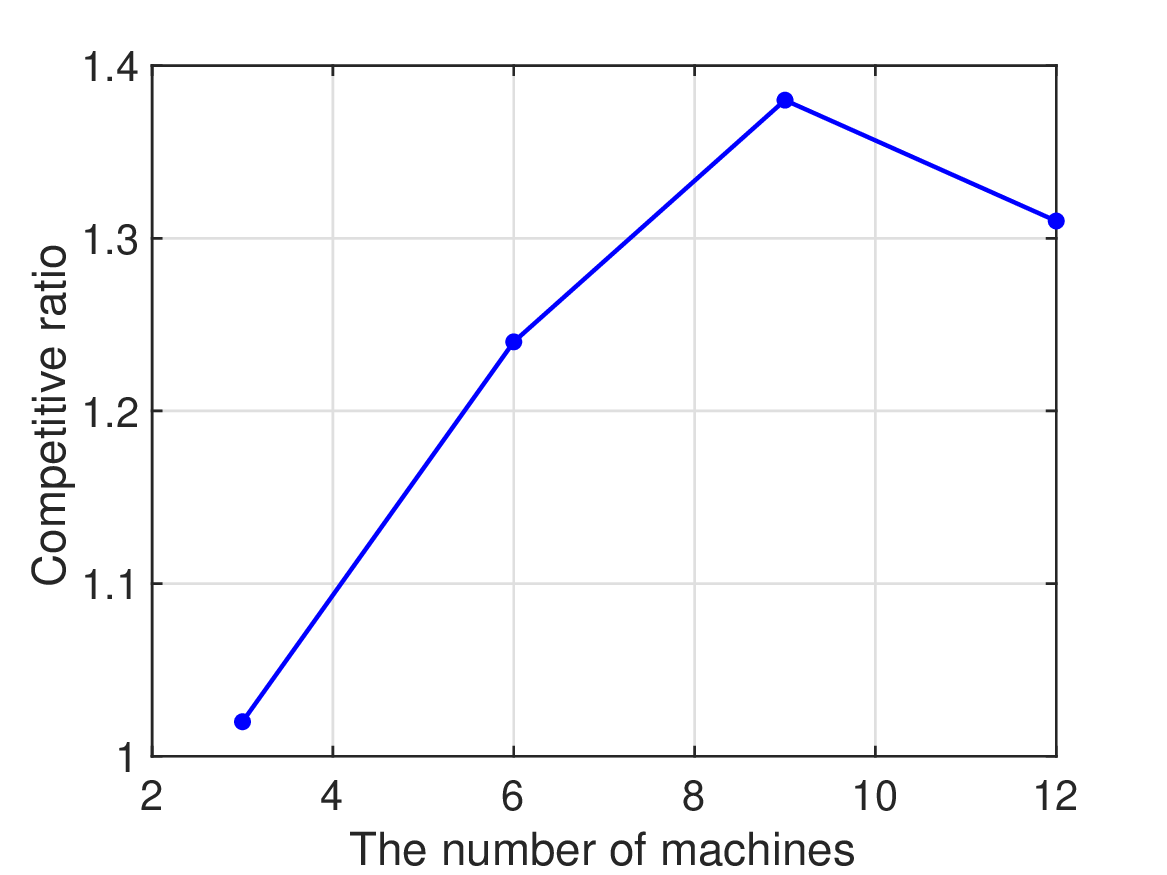}
        \caption{Competitive ratio.} \label{fig:ratio}
\end{minipage}
\hspace{0.06\columnwidth}%
\begin{minipage}{0.3\textwidth}
\vspace{-.1in}
   \includegraphics[width=1.11\textwidth]{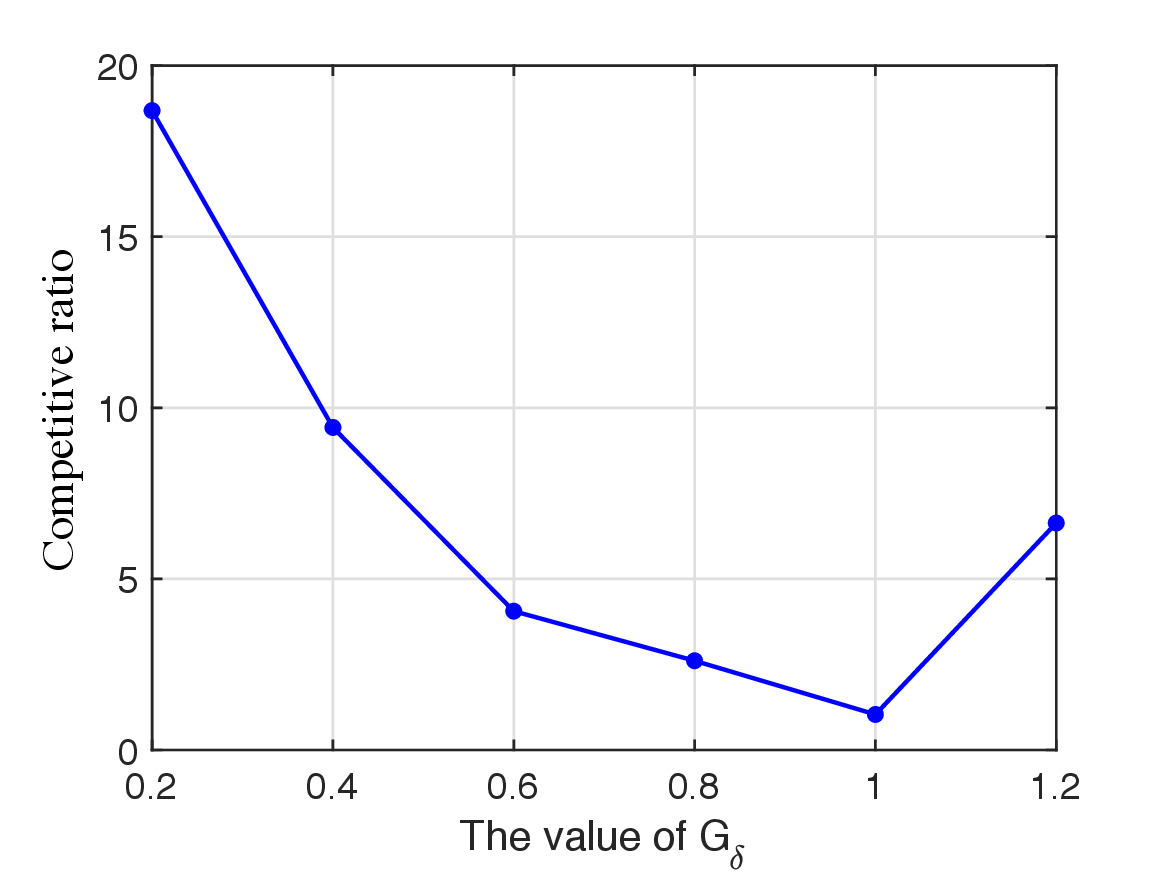}
    \caption{Impact of pre-rounding gain factor $G_\delta$ on competitive ratio.}\label{fig:Factor_ApprRatio}
\end{minipage}
\end{figure*}

\begin{figure*}[t!]
\begin{minipage}{0.3\textwidth}
\includegraphics[width=1.01\textwidth]{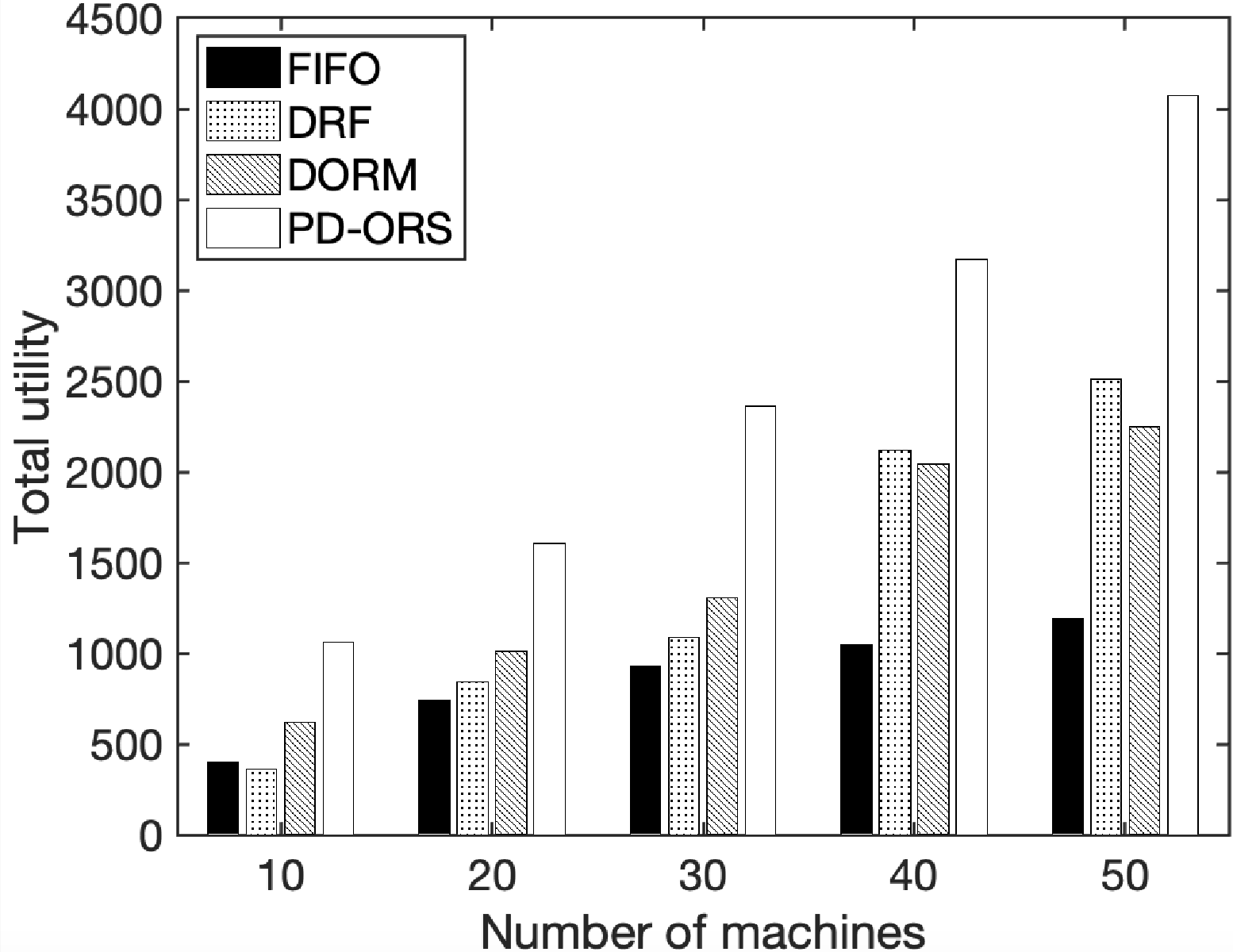}
\caption{Total utility with  increasing number of machines (Google cluster data trace).}\label{fig:jobFixed_response}
\end{minipage}
\hspace{0.06\columnwidth}%
\begin{minipage}{0.3\textwidth}
\includegraphics[width=1.02\textwidth]{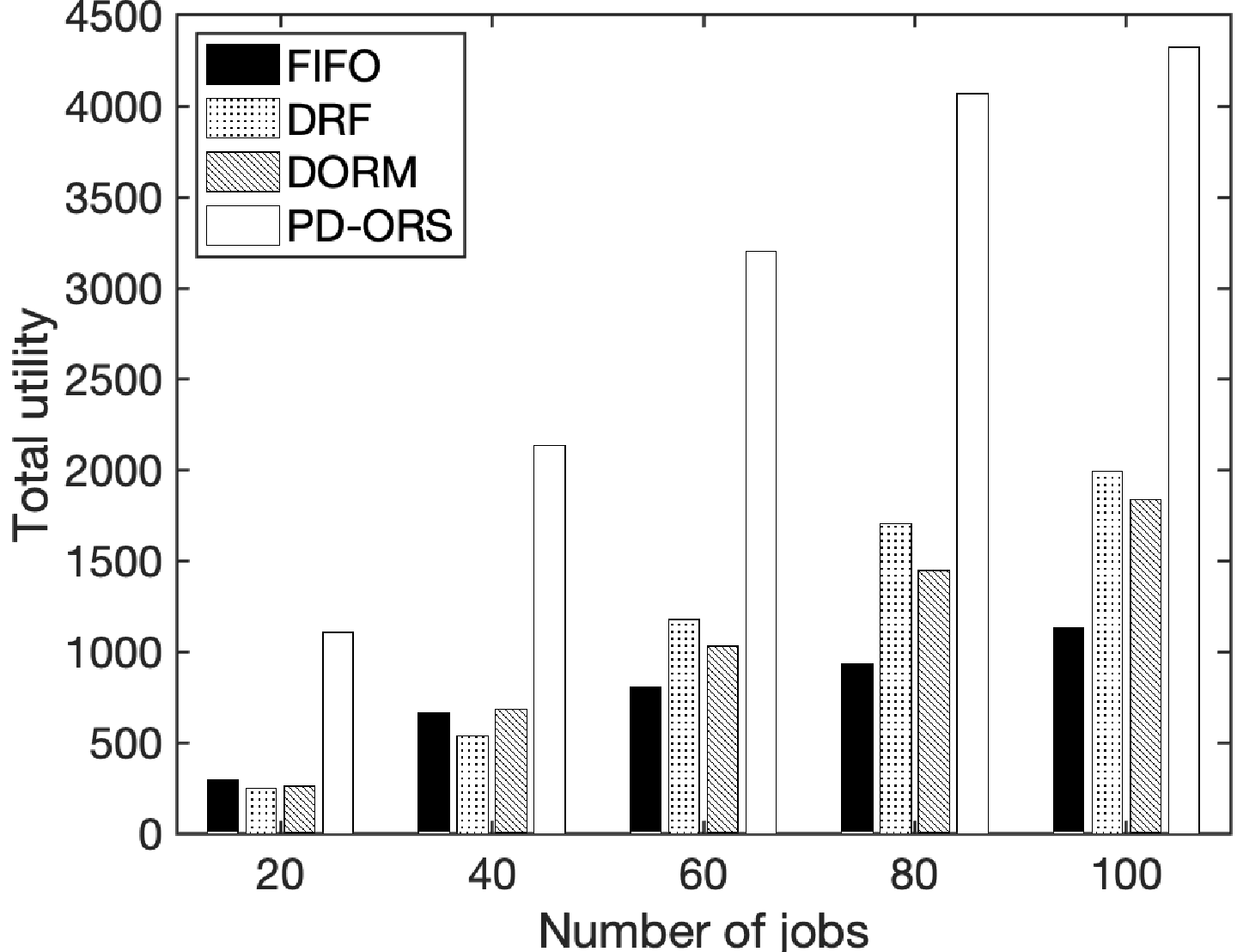}
\caption{Total utility with  increasing number of jobs (Google cluster data trace).} \label{fig:serverFixed_response}
\end{minipage}
\hspace{0.06\columnwidth}%
\begin{minipage}{0.3\textwidth}
\vspace{-.1in}
\includegraphics[width=1.125\textwidth]{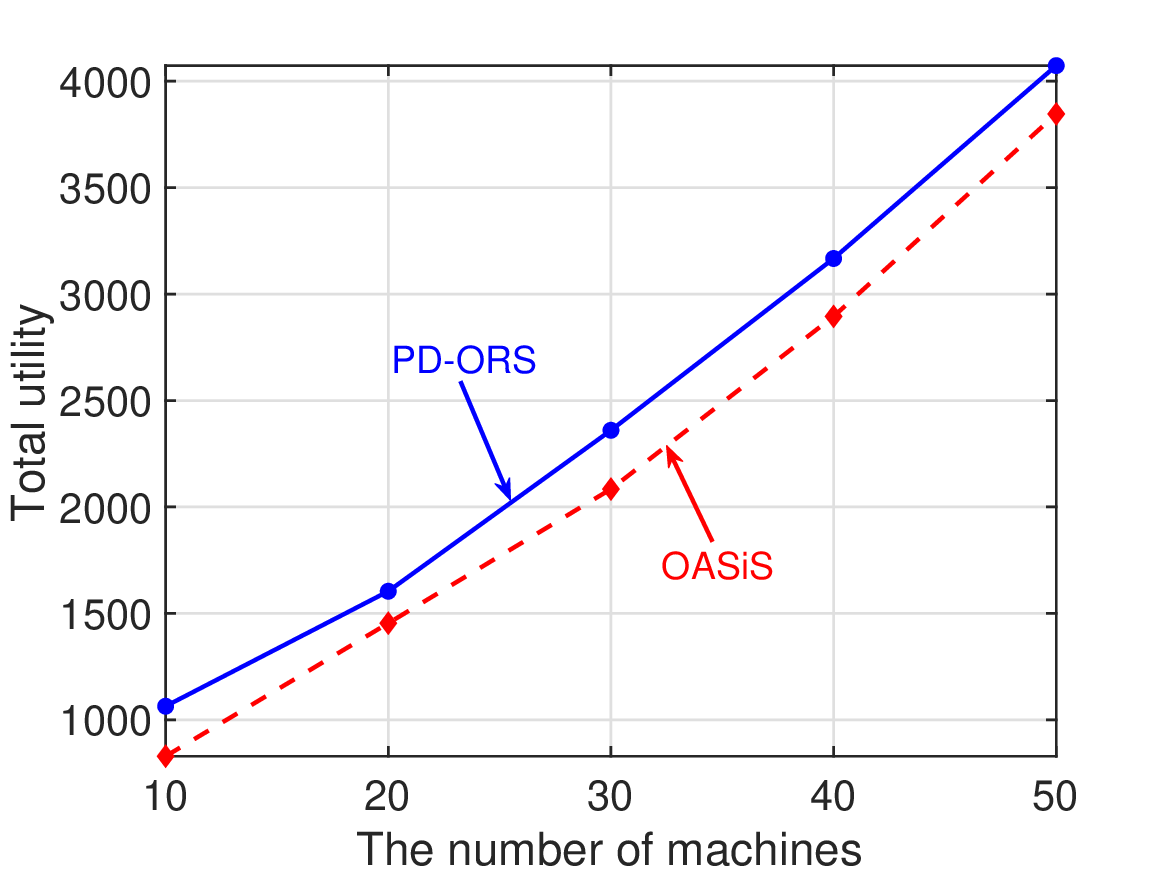}
\caption{Total utility with  increasing number of machines [T=80, I=100,  (10\%, 55\% , 35\%)].}\label{fig:jobFixed1}
\end{minipage}
\end{figure*}

\begin{figure*}[t!]
\begin{minipage}{0.3\textwidth}
\includegraphics[width=1.11\textwidth]{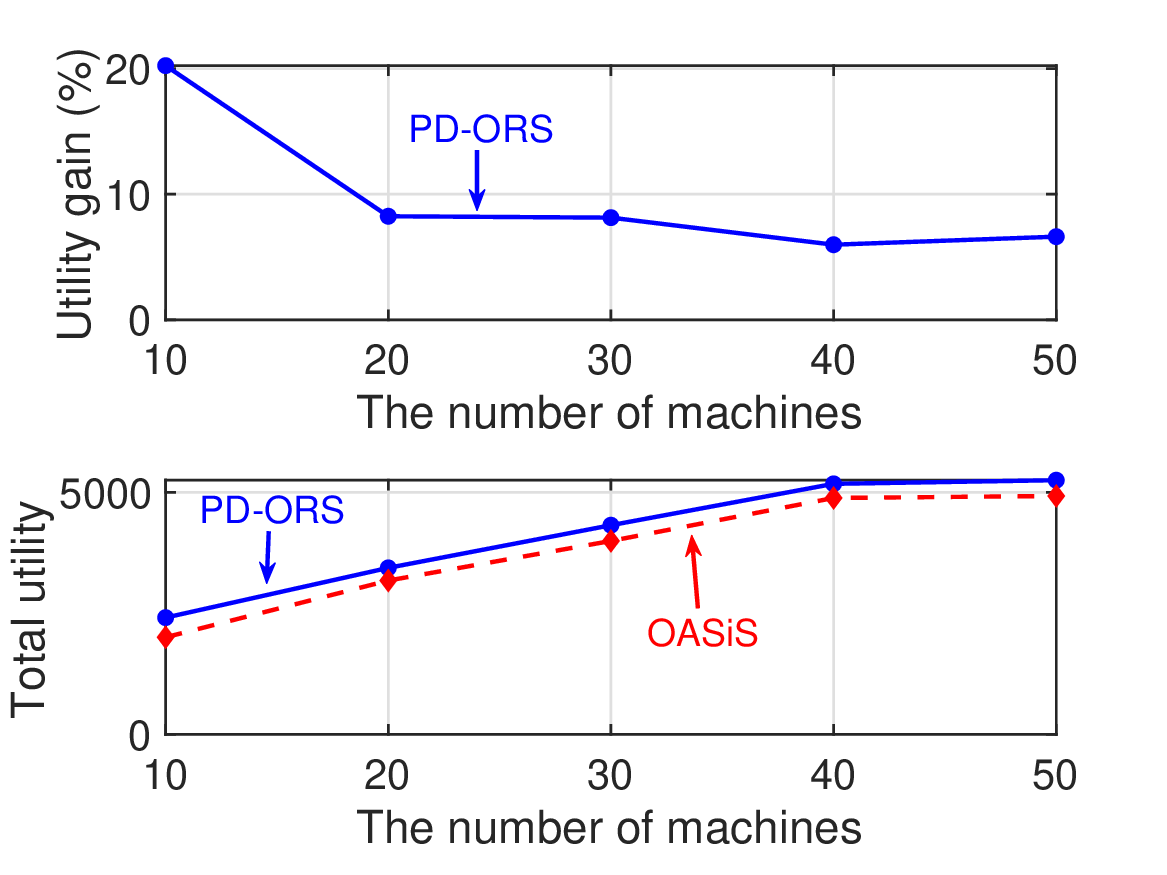}
\caption{Total utility with  increasing number of machines [T=80, I=100,  (30\%, 69\% , 1\%)].} \label{fig:jobFixed2}
\end{minipage}
\hspace{0.03\columnwidth}%
\begin{minipage}{0.3\textwidth}
        \includegraphics[width=1.11\textwidth]{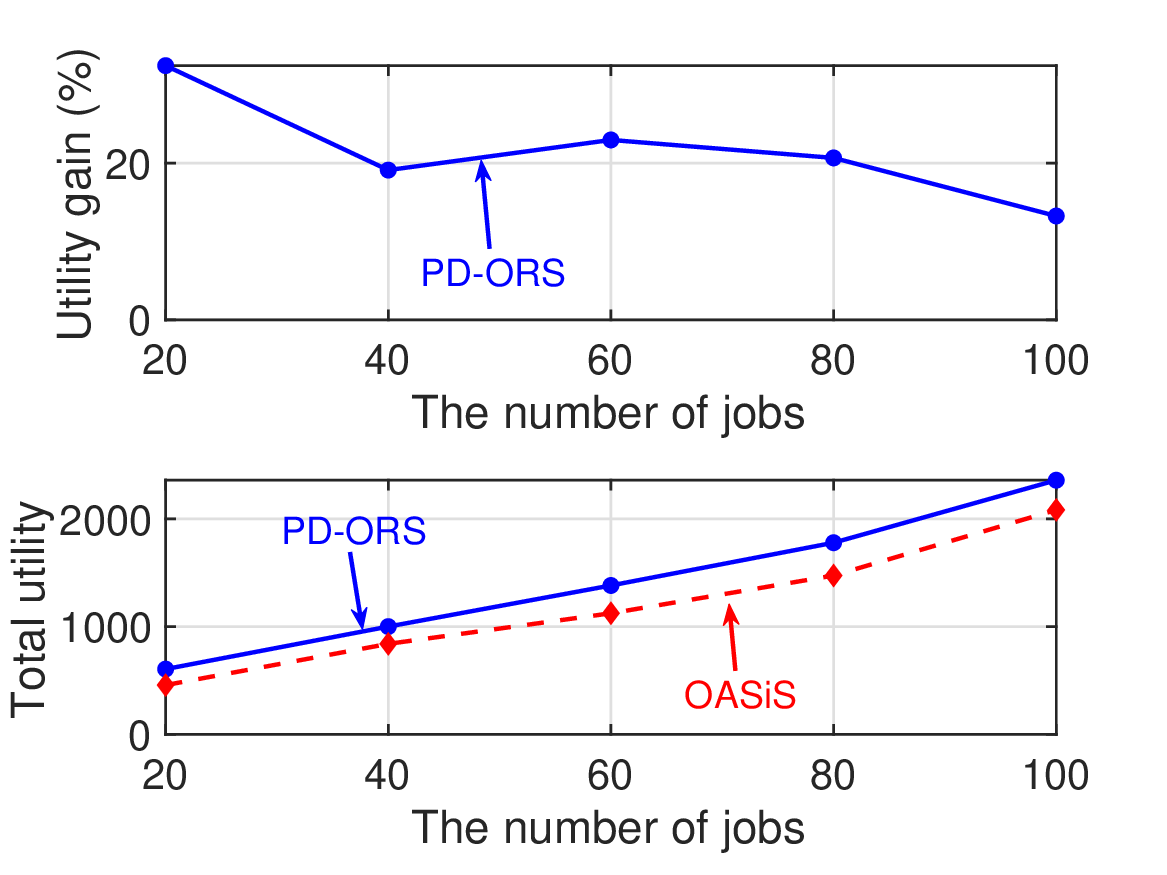}
        \caption{Total utility with  increasing number of jobs [T=80, H=30,  (10\%, 55\% , 35\%)].} \label{fig:serverFixed1}
\end{minipage}
\hspace{0.03\columnwidth}%
\begin{minipage}{0.3\textwidth}
   \includegraphics[width=1.25\textwidth]{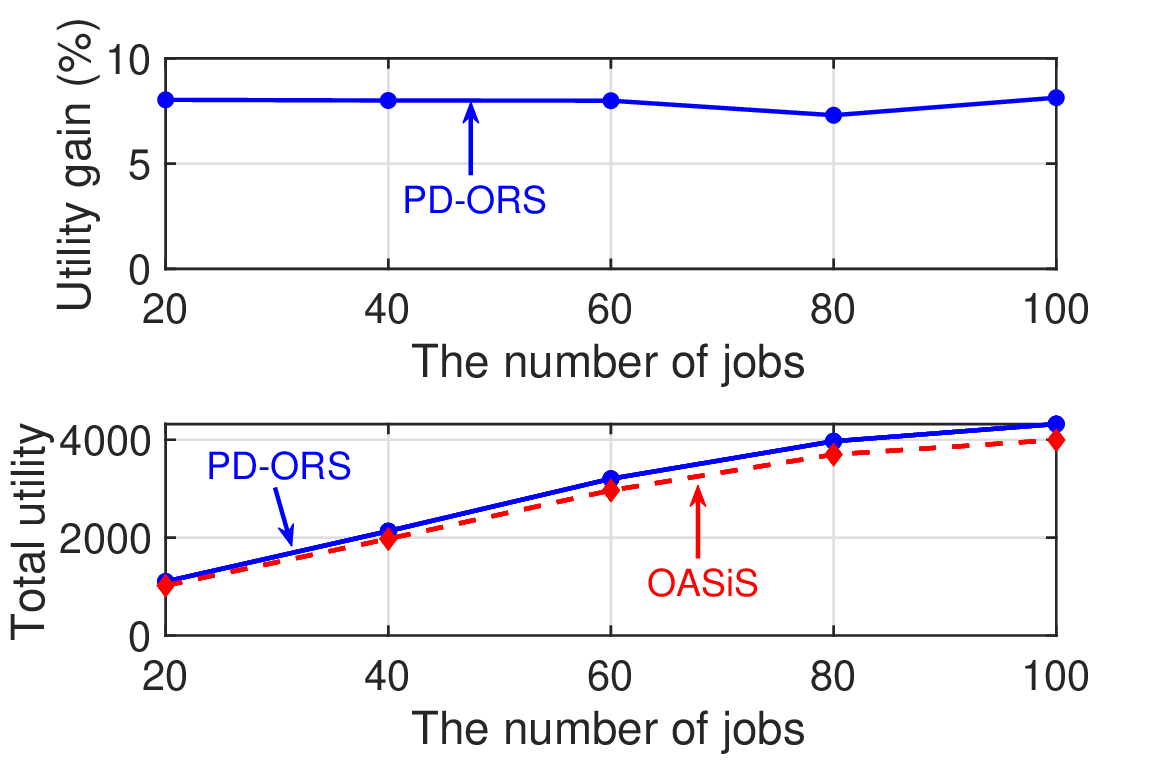}
    \caption{Total utility with  increasing number of jobs [T=80, H=30,  (30\%, 69\% , 1\%)].}\label{fig:serverFixed2}
\end{minipage}
\end{figure*}
 
Next, we compare our PD-ORS algorithm with the OASiS algorithm in~\cite{Bao18:ML_INFOCOM}, which is also a dynamic scheduling scheme. 
As mentioned earlier, the key difference in OASiS is that parameter servers and workers are located on two {\em strictly} separated sets of machines (i.e., no co-located workers and PSs).
Here, we let $H=100$ and $T=20$.
For OASiS, half of the machines host parameter servers and the other half host workers.
For fair comparisons, both algorithms adopt the same Sigmoid utility function.
The comparison results are shown in Fig.~\ref{fig:UtilityComp_1}.
We can see that PD-ORS outperforms OASiS by allowing co-located parameter servers and workers. 
We can see from Fig.~\ref{fig:UtilityComp_1} that the performance gap between PD-ORS and OASiS widens as the number of jobs increases, which implies that PD-ORS is more scalable.
This is due to the advantage afforded by colocation of workers and parameter servers, which allows each physical machine to be fully utilized.
On the other hand, the strict separation of workers and parameter servers in OASiS may lead to the inability of placing workers on server-side machines, should there be available resources or vice verse.

Next, we investigate  the actual training time (completion time - arrival time) under different methods, where $T=80$, $H=30$ and $I=100$.
The median of the actual training time is shown in Fig.~\ref{fig:actual_time}.
Here, we simply set its training time to $T$ (i.e., 80) if the job cannot be finished within the scheduling time span $\mathcal{T}$.
As we can see from Fig.~\ref{fig:actual_time}, PD-ORS outperforms other scheduling policies, i.e., it has the smallest median time.
Also, due to the co-location advantage of PD-ORS, its median time is smaller compared to OASiS, where workers and parameter servers are placed in strictly separated sets of machines.
We expect that the difference between PD-ORS and OASiS will become more noticeable as the number of machines or the capacity of each machine increases since it will allow more co-location placements.

Next, we demonstrate the competitive ratio of our algorithm PD-ORS, which is the ratio between the total job utility of the offline optimal solution and the total job utility achieved by PD-ORS.
Recall that Problem~DMLRS is a non-convex problem with constraints (e.g., Eq.~(\ref{eqn_NumTrainedSamples})) that are not amenable to be directly solved by conventional optimization techniques. To obtain its offline optimum,  all possible combinations of $w_{ih}[t], s_{ih}[t], \forall i, h, t$ need to be considered, which is time prohibitive.
Thus, we limit the number of jobs $I$ to 10 and time span $T$ to 10, and the result is shown in Fig.~\ref{fig:ratio}.
As we can see from the figure, the performance ratio is between 1.0 to 1.4, indicating that our proposed algorithm PD-ORS has a good competitive ratio performance.

Lastly, we examine the performance of the randomized rounding scheme in Algorithm~4, which is the key of PD-ORS.
We evaluate the rounding performance in terms of the ratio between the optimal total utility and the total utility obtained by our algorithm.
The optimal utility is computed using the Gurobi optimization solver.
We let $H=100$, $I=50$, $T=20$.
We vary the pre-rounding gain factor $G_{\delta}$ (Theorems~\ref{thm_Alg4_1} and \ref{thm_Alg4_2}) from $0.2$ to $1.2$.
The results are shown in Fig.~ \ref{fig:Factor_ApprRatio}. 
The packing constraints are easier to satisfy with a smaller $G_\delta$, while the cover constraints are prone to be violated as $G_\delta$ gets smaller.
In our experiments, if the total rounds of randomized rounding before we find an integer feasible solution exceeds a preset threshold (e.g, $5000$), we will discard the corresponding job.
Theorem~\ref{thm_Alg4_1} suggests that there is a trade-off: if we set $G_\delta$ to be close to one to pursue a better total utility result, the rounding time could be large to obtain a feasible solution.
As $G_\delta$ increases, the probability of violating the packing constraints increases, meaning we need to have more rounding attempts to obtain an integer feasible solution.
However, according to our numerical experiences, if the machine's resource capacity is relatively large compared to the jobs' resource demands per worker/PS, the number of rounding attempts is small and not sensitive to the variation of $G_\delta$.
On the other hand, as $G_\delta$ decreases, the probability of violating the cover constraint increases.
However, in practice, the model usually converges with fewer number of iterations than the pre-defined training epochs since the required number of epochs is usually overestimated~\cite{Jeon18:multi-tenant}.
In other words, the violation of the cover constraint in one iteration may be acceptable.

As we can see from Fig.~\ref{fig:Factor_ApprRatio}, the best approximation ratio value is achieved when $G_\delta=1$.
This is because if $G_\delta$ approaches $0$, it implies that $\delta$ decreases at a larger rate (cf. Eq.~(\ref{eqn_Thm3_G})), resulting the increment of the performance ratio.
On the other hand, if $G_\delta$ goes to infinity, it implies $\delta$ decreases (cf. Eq.~\eqref{eqn_Thm4_G}), resulting in a much faster increment of the performance ratio.
Also, we can see from Fig.~\ref{fig:Factor_ApprRatio} that the performance ratios for all choices of $G_\delta$ are much better than the theoretical bounds in Theorems~\ref{thm_Alg4_1} and \ref{thm_Alg4_2}, which shows that the approximation ratio is much tighter than the worse-case bound suggested in Theorems~\ref{thm_Competitive_1} and \ref{thm_Competitive_2}.

Next, we show further experimental results with real-world data traces.
We first compare our PD-ORS algorithm with baseline scheduling algorithms, where we follow job arrivals exactly based on timestamps recorded in the Google Cluster data~\cite{Reiss12:googlecluster} by scaling down the original job trace (i.e., a ``snippet'' of the trace).
Here, we set $T=80$, $I=100$ and $H=30$.
The comparison results are shown in Figures~\ref{fig:jobFixed_response} - \ref{fig:serverFixed_response}.
Similarly, as we can see from the figures, our algorithm PD-ORS outperforms other scheduling policies.
In addition, due to the co-location advantage of PD-ORS, it achieves more total job utility compared to OASiS.

In the previous experiments, we have set the portions for time-insensitive jobs, time-sensitive jobs and time-critical jobs to 10\%, 55\% and 35\%, respectively, which follows the default setting in~\cite{Bao18:ML_INFOCOM} for fair comparison. 
Theoretically, the larger the portion of time-sensitive and time-critical jobs is, the better the performance is in terms of job utility compared to other scheduling policies.
Based on the Google trace analysis~\cite{Minet18:google_trace}, there are four categories of scheduling class of a job to indicate the latency sensitivity of the job, where we label class 0 as time-insensitive, Classes 1 and 2 as time-sensitive, and Class 3 as time-critical.
In order to follow the practical setting in the trace, we roughly set the ratio to 30\%, 69\% and 1\%.
We set $T=80$.
The number of machines increases from $10$ to $50$ with the number of job fixed to $100$, and the number of jobs increases from $20$ to $100$ with the number of machines fixed at $30$.
We let Figs.~\ref{fig:jobFixed1} and \ref{fig:serverFixed1} follow the previous ratio setting (i.e., 10\%, 55\% and 35\%), and Figs.~\ref{fig:jobFixed2} and \ref{fig:serverFixed2} follow the revised ratio setting (i.e., 30\%, 69\% and 1\%).
We examine the utility gain compared to OASiS, where it is normalized.
We present our experimental results in Figs.~\ref{fig:jobFixed1} -- \ref{fig:serverFixed2}.
We can see from the figures that as the portion of critical jobs decreases by 34\%, the utility gain becomes smaller.
That is, the advantage of our proposed algorithm PD-ORS becomes less prominent.

We note that although in theory we can compare our PD-ORS algorithm (a special case with utility function $u(x)=x$) with Optimus in~\cite{Peng18:ML_EuroSys} (which also takes co-location into consideration), it is not straightforward to do so in practice.
Optimus requires an offline stage to estimate the $\theta$-parameters of the speed function $f(p_j, w_j)$ (cf. $\theta_{0}$--$\theta_{3}$ in [\!\!\cite{Peng18:ML_EuroSys}, Eq.~(3)] for asynchronous training and $\theta_{0}$--$\theta_{4}$ for synchronous training in [\!\!\cite{Peng18:ML_EuroSys}, Eq.~(4)]).
Estimating these parameters requires specific hardware and software packages that are not available in our current experimental environment.
Due to the above computing resource limitations and time constraints, we are unable to conduct experiments to directly compare PD-ORS and Optimus in this work.
Also, our focus in this work is on scheduling algorithmic designs for deep learning training with main contributions being on the theoretical aspects (proving rigorous scheduling performance guarantees).
Implementing our PD-ORS algorithm in a similar testbed environment and having a comparison with Optimus is very interesting and will be our next step in future studies.
\section{Conclusion}
\label{sec:conclusion}

In this paper, we investigated online resource scheduling for distributed machine learning jobs in a shared computing cluster.
We considered the most general setting that workers and parameter servers can be co-located on the same set of physical machines.
We showed that this problem can be formulated as a challenging integer nonlinear programming problem with {\em non-deterministic} constraints.
We developed an efficient online scheduling algorithm with competitive ratio guarantee. 
Our main contributions are three-fold: 
i) We developed a new analytical model that jointly considers resource locality and allocation; 
ii) Through careful examinations of worker-server configuration relationships, we resolved the locality ambiguity in the model and reduce the problem to a mixed cover/packing integer program that entails low-complexity approximation algorithm design;
iii) We proposed a meticulously designed randomized rounding algorithm to solve the mixed cover/packing integer program and rigorously established its approximation ratio guarantee. 
Collectively, our results expand the theoretical frontier of online optimization algorithm design for resource optimization in distributed machine learning systems.

\begin{IEEEbiography}[{\includegraphics[width=1in, height=1.25in]{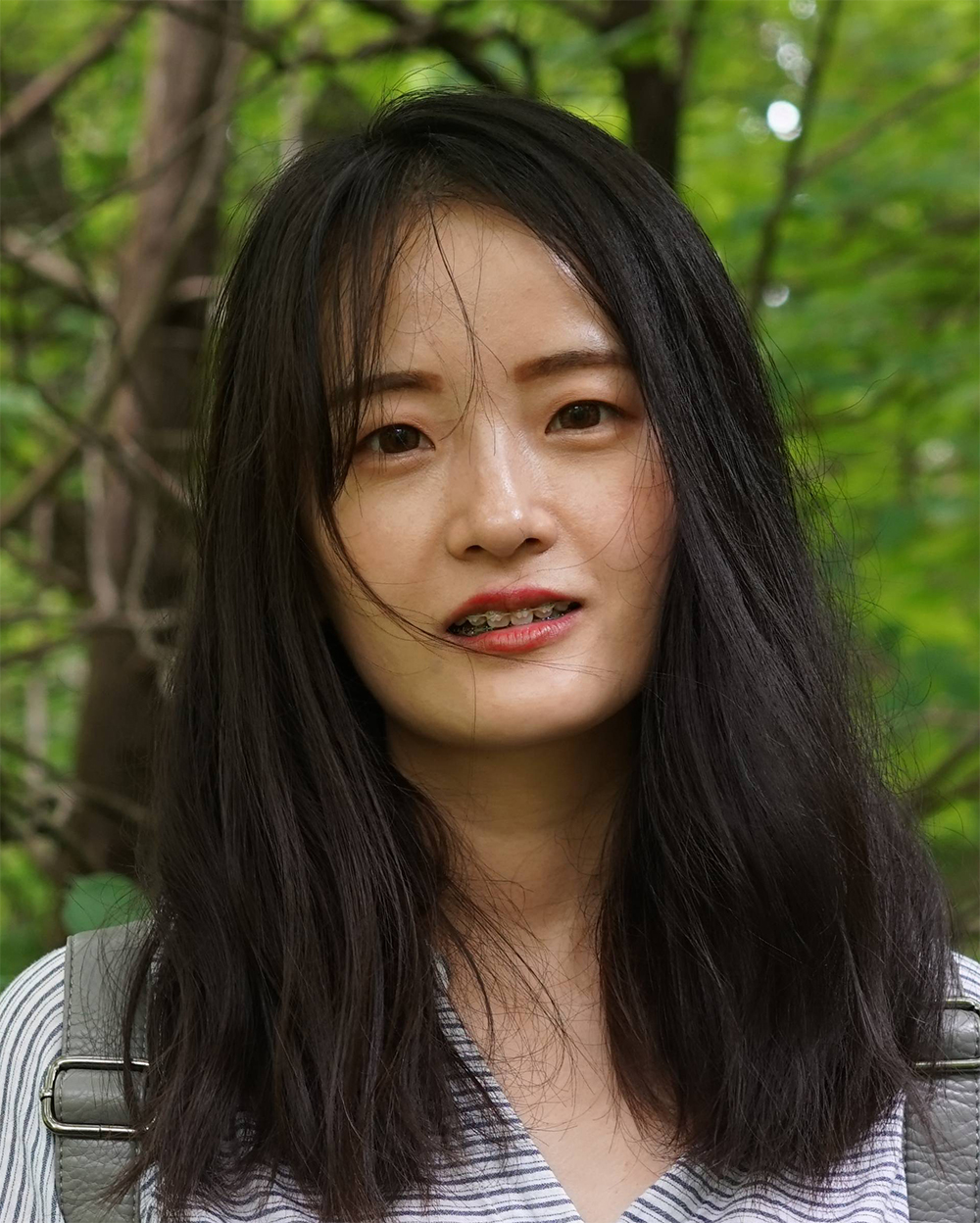}}]
{Menglu Yu} received her B.S. degree from the Department of Electrical and Information Engineering at Hunan University, China in 2014.
She is currently pursuing her Ph.D. degree in the Department of Computer Science at Iowa State University.
Her primary research interests include optimization for distributed machine learning systems and data centers, as well as network optimization.
\end{IEEEbiography}

\begin{IEEEbiography}[{\includegraphics[width=1in,height=1.25in,clip,keepaspectratio]{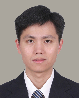}}]
{Jia Liu} (S'03--M'10--SM'16) is an Assistant Professor in the Department of Electrical and Computer Engineering at The Ohio State University, where he joined in Aug. 2020. 
He received his Ph.D. degree from the Department of Electrical and Computer Engineering at Virginia Tech in 2010. 
From Aug. 2017 to Aug. 2020, he was an Assistant Professor in the Department of Computer Science at Iowa State University. 
His research areas include theoretical machine learning, control and optimization for stochastic networks, and optimization for data analytics infrastructure and cyber-physical systems. 
Dr. Liu is a senior member of IEEE and a member of ACM. He has received numerous awards at top venues, including IEEE INFOCOM'19 Best Paper Award, IEEE INFOCOM'16 Best Paper Award, IEEE INFOCOM'13 Best Paper Runner-up Award, IEEE INFOCOM'11 Best Paper Runner-up Award, and IEEE ICC'08 Best Paper Award. 
Dr. Liu is a recipient of the NSF CAREER Award in 2020. He is a recipient of the Google Faculty Research Award in 2020. 
He is also a winner of the LAS Award for Early Achievement in Research from the College of Liberal Arts and Sciences at Iowa State University in 2020, and the Bell Labs President Gold Award in 2001. 
His research is supported by NSF, AFOSR, AFRL, and ONR.
\end{IEEEbiography}

\begin{IEEEbiography}[{\includegraphics[width=1in,height=1.25in,clip,keepaspectratio]{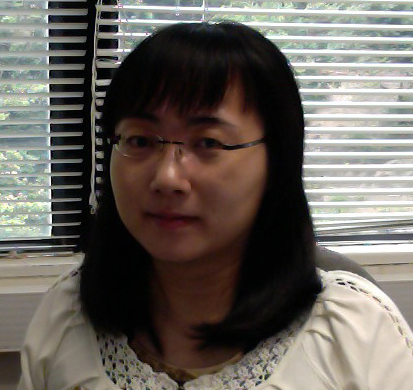}}] 
{Chuan Wu} received her B.Engr. and M.Engr. degrees in 2000 and 2002 from the Department of Computer Science and Technology, Tsinghua University, China, and her Ph.D. degree in 2008 from the Department of Electrical and Computer Engineering, University of Toronto, Canada. Since September 2008, Chuan Wu has been with the Department of Computer Science at the University of Hong Kong, where she is currently a Professor. Her current research is in the areas of cloud computing, distributed machine learning systems and algorithms, and intelligent elderly care technologies. She is a senior member of IEEE, a member of ACM, and served as the Chair of the Interest Group on Multimedia services and applications over Emerging Networks (MEN) of the IEEE Multimedia Communication Technical Committee (MMTC) from 2012 to 2014. She is an associate editor of IEEE Transactions on Cloud Computing, IEEE Transactions on Multimedia, ACM Transactions on Modeling and Performance Evaluation of Computing Systems, and IEEE Transactions on Circuits and Systems for Video Technology. She was the co-recipient of the best paper awards of HotPOST 2012 and ACM e-Energy 2016. 
\end{IEEEbiography}
%

\begin{IEEEbiography}[{\includegraphics[width=1in,height=1.25in,clip,keepaspectratio]{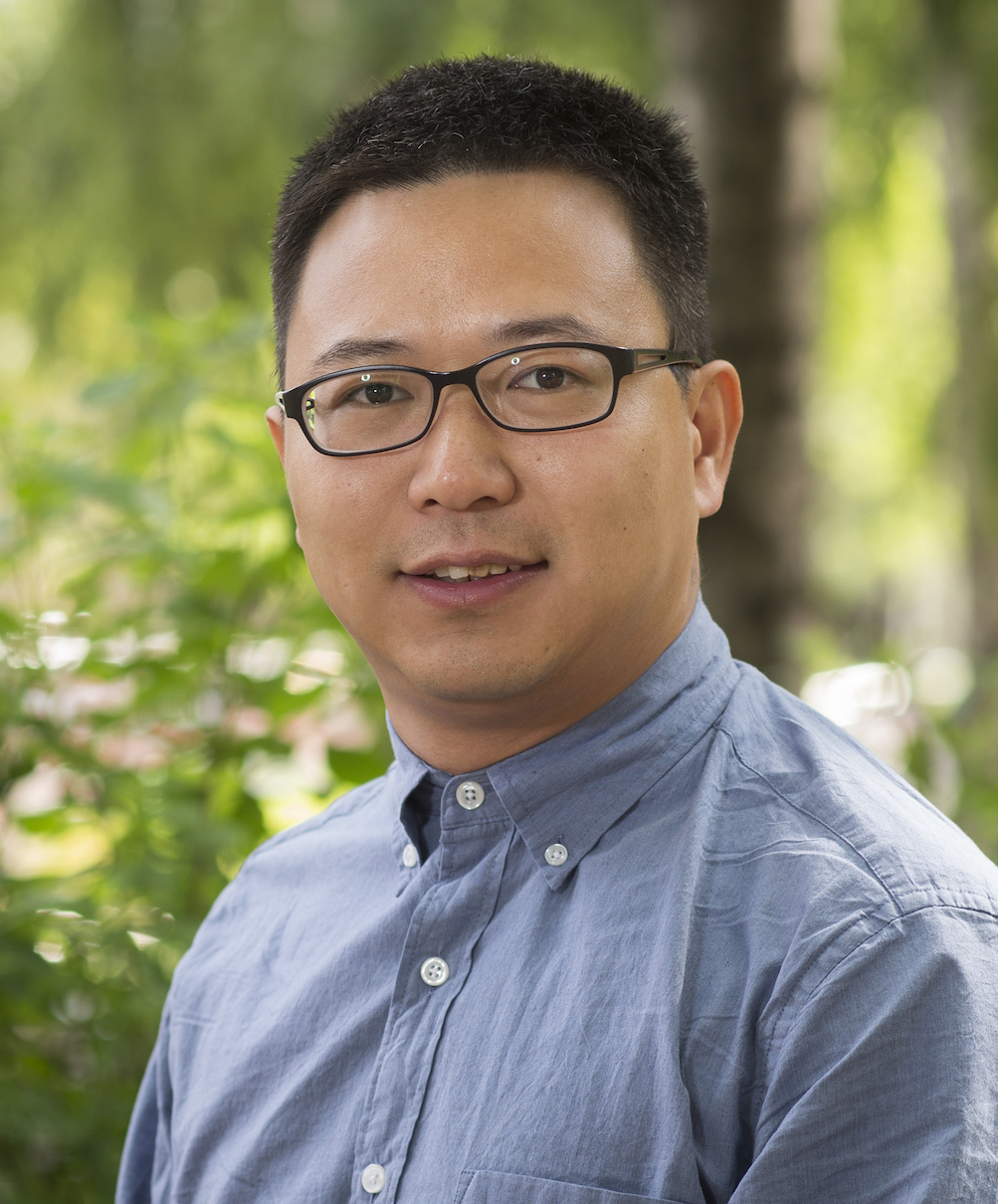}}]{Bo Ji}(S'11-M'12-SM'18)
received his B.E. and M.E. degrees in Information Science and Electronic Engineering from Zhejiang University, Hangzhou, China, in 2004 and 2006, respectively, and his Ph.D. degree in Electrical and Computer Engineering from The Ohio State University, Columbus, OH, USA, in 2012. Dr. Ji is an Associate Professor in the Department of Computer Science at Virginia Tech, Blacksburg, VA, USA. Prior to joining Virginia Tech, he was an Associate Professor in the Department of Computer and Information Sciences and a faculty member of the Center for Networked Computing at Temple University, where he was an Assistant Professor from July 2014 to June 2020. He was also a Senior Member of the Technical Staff with AT\&T Labs, San Ramon, CA, from January 2013 to June 2014. His research interests are in the modeling, analysis, control, optimization, and learning of computer and network systems, such as wired and wireless networks, large-scale IoT systems, high performance computing systems and data centers, and cyber-physical systems. He currently serves on the editorial boards of the IEEE/ACM Transactions on Networking, IEEE Internet of Things Journal, and IEEE Open Journal of the Communications Society. Dr. Ji is a senior member of the IEEE and a member of the ACM. He is a National Science Foundation (NSF) CAREER awardee (2017) and an NSF CISE Research Initiation Initiative (CRII) awardee (2017). He is also a recipient of the IEEE INFOCOM 2019 Best Paper Award.
\end{IEEEbiography}
%
%
\begin{IEEEbiography}
[{\includegraphics[width=1in,height=1.25in,clip,keepaspectratio]{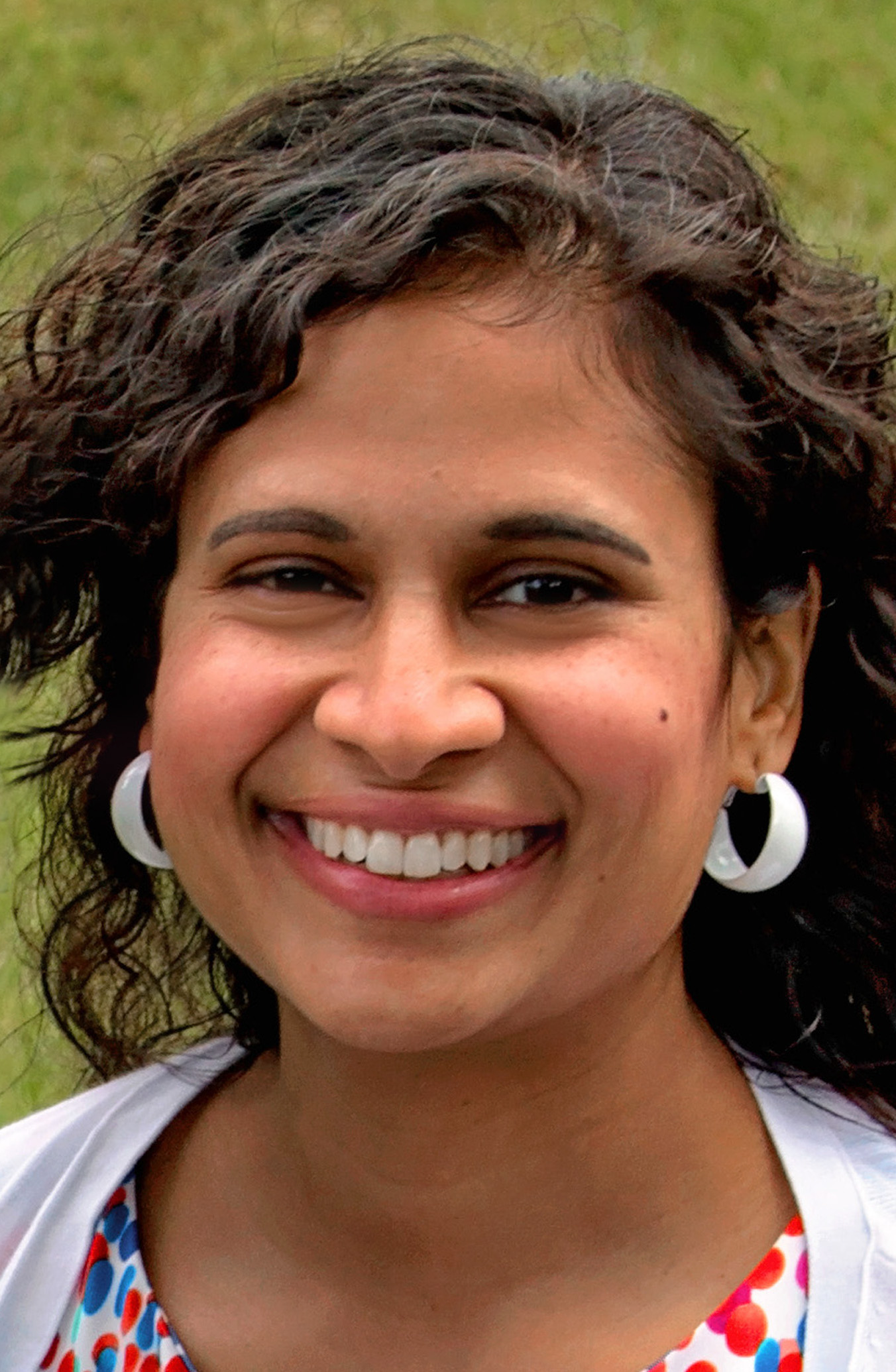}}]
{Elizabeth S. Bentley} has a B.S. degree in Electrical Engineering from Cornell University, a M.S. degree in Electrical Engineering from Lehigh University, and a Ph.D. degree in Electrical Engineering from University at Buffalo.
She was a National Research Council Post-Doctoral Research Associate at the Air Force Research Laboratory in Rome, NY.
Currently, she is employed by the Air Force Research Laboratory in Rome, NY, performing in-house research and development in the Networking Technology branch.
Her research interests are in cross-layer optimization, wireless multiple-access communications, wireless video transmission, modeling and simulation, and directional antennas/directional networking.
\end{IEEEbiography}

\clearpage
\appendices
\section{Proof of Lemma~\ref{lem_ILP}} \label{appdx:lem_ILP1}

\begin{proof}
Consider the probabilities of the following ``bad'' events: 1) $\c^{\top} \hat{\x} > \frac{3G_\delta}{\delta} \c^{\top}\bar{\x}$; 2) $\exists i$ such that $(\A\hat{\x})_{i} < a_{i}$; and 3) $\exists i$ such that $(\B\hat{\x})_{i} > b_{i}$.
Note that events 2) and 3) can be equivalently rewritten as: 2') $\exists i$ such that $\mathbb{E}\{(\A\hat{\x})_{i} \frac{W_{a}}{a_{i}} < W_{a}\}$ and 3') $\exists i$ such that $\mathbb{E}\{(\B\hat{\x})_{i} \frac{W_{b}}{b_{i}} > W_{b}\}$.
Since $\mathbb{E}\{ \hat{\x} \} \!=\! \x' \!=\! G_\delta\bar{\x}$, by linearity of expectation, we have:
\begin{align}
\label{eqn_exp1}&\mathbb{E}\{\c^{\top} \hat{\x} \} = \c^{\top} \mathbb{E}\{ \hat{\x} \} = \c^{\top} G_\delta \bar{\x} = G_\delta \c^{\top}\bar{\x}, \\
\label{eqn_exp2}& \mathbb{E}\left\{ (\A\hat{\x})_{i} \frac{W_{a}}{a_{i}} \right\} = G_\delta \mathbb{E} \left\{ (\A\bar{\x})_{i}\frac{W_{a}}{a_{i}} \right\} \geq G_\delta W_{a}, \\
\label{eqn_exp3}& \mathbb{E}\left\{ (\B\hat{\x})_{i} \frac{W_{b}}{b_{i}} \right\} = G_\delta \mathbb{E} \left\{ (\B\bar{\x})_{i}\frac{W_{b}}{b_{i}} \right\} \leq G_\delta W_{b}.
\end{align}
Then, by the Markov inequality and ~(\ref{eqn_exp1}), we can obtain the probability $
\mathrm{Pr} \left\{ \c^{\top} \hat{\x} > \frac{3G_\delta}{\delta} \c^{\top} \bar{\x} \right\} \leq \frac{\delta}{3}$.

Next, we note that each $\hat{x}_{j}$ can be viewed as a sum of independent random variables in $[0,1]$ as follows:
The fixed part of $\lfloor x'_{j} \rfloor$ is a sum of $\lfloor x'_{j} \rfloor$ random variables with value 1 with probability 1.

Then, we have that $(\B\hat{\x})_{i} \frac{W_b}{b_{i}} = (\sum_{j}[\B]_{ij}\hat{x}_{j}) \frac{W_b}{b_{i}}$ is also a sum of independent random variables in $[0,1]$.
Using the Chernoff bound, we have 
\begin{align*}
  \mathrm{Pr} \Big\{ (\B\hat{\x})_{i} \frac{W_{b}}{b_{i}} > (1+\epsilon) G_\delta W_{b} \Big\} \leq \exp \Big( -\epsilon^{2} \frac{G_\delta W_{b}}{3} \Big).  
\end{align*}
Setting $(1+\epsilon)G_\delta \!=\!1$, i.e., $\epsilon \!=\! \frac{1}{G_\delta}\!-\!1$, we have:
\begin{align} \label{eqn_chernoff_bnd1_1}
\mathrm{Pr} \left\{ (\B\hat{\x})_{i} \frac{W_b}{b_{i}} > W_b \right\} \leq \exp\left(- \Big(\frac{1}{G_\delta}-1\Big)^{2} G_\delta \frac{W_b}{3} \right). \!\!\!
\end{align}
Forcing $\exp\!\left(- (\frac{1}{G_\delta}\!-\!1)^{2} G_\delta \frac{W_b}{3} \right) \!=\! \frac{\delta}{3r}$ and solving $G_\delta$, we have:
\begin{align} \label{eqn_L_def_1}
G_\delta \triangleq 1 + \frac{3\ln(3r/\delta)}{2W_b}\! -\!\sqrt{ \bigg(\frac{3\ln(3r/\delta)}{2W_b} \bigg)^{2} + \frac{3\ln(3r/\delta)}{W_b} }. \!\!\!
\end{align}

Using (\ref{eqn_exp2}), the Chernoff bound, and following similar arguments, we have:
\begin{align*}
\mathrm{Pr} \Big\{ (\A\hat{\x})_{i} \frac{W_{a}}{a_{i}} \leq (1-\epsilon) G_\delta W_{a} \Big\} \leq \exp \Big(-\epsilon^{2}\frac{G_\delta W_{a}}{2} \Big).    
\end{align*}
Forcing $\exp\left( -\epsilon^{2} \frac{G_\delta W_a}{2} \right)=\frac{\delta}{3m}$ and solving for $\epsilon$, we have $\epsilon = (\frac{2}{G_\delta W_a}\ln(\frac{3m}{\delta}))^{\frac{1}{2}}$. 
It follows that 
\begin{align*}
 \mathrm{Pr} \{ (\A\hat{\x})_{i} \frac{W_a}{a_{i}} \leq ( 1 \!-\! (\frac{2}{G_\delta W_a}\ln(\frac{3m}{\delta}))^{\frac{1}{2}} )G_\delta W_a\} \leq \frac{\delta}{3m},
\end{align*}
which implies that: 
\begin{align} \label{eqn_chernoff_bnd2_1}
\!\!\! \mathrm{Pr} \left\{ (\A\hat{\x})_{i} \!\leq\! a_i \left( 1 \!-\! \sqrt{\frac{2}{G_\delta W_a}\ln\Big(\frac{3m}{\delta}\Big)} \right)G_\delta, \exists i \right\} \leq \frac{\delta}{3m}. \!\!\!\!\!
\end{align}
By using union bound and (\ref{eqn_chernoff_bnd1_1}) and (\ref{eqn_chernoff_bnd2_1}), we have that events 1)--3) occur with probability less than $\frac{\delta}{3} + m\cdot\frac{\delta}{3m} + r\cdot\frac{\delta}{3r} = \delta$, 
and the proof is complete.
\end{proof}

\section{Proof of Lemma~\ref{lem_ILP2}} \label{appdx:lem_ILP2}

\begin{proof}
Similar to the case when $0<G_\delta\leq 1$, we can have the expectation equations for the bad cases as in Eqns.~(\ref{eqn_exp1})-(\ref{eqn_exp3}).
We also can view each $\hat{x}_j$ as a sum of independent random variables in $[0,1]$ in the same way.
Then, we have that $(\A\hat{\x})_{i} \frac{W_a}{a_{i}} = (\sum_{j}[\A]_{ij}\hat{x}_{j}) \frac{W_a}{a_{i}}$ is also a sum of independent random variables in $[0,1]$.
Using the Chernoff bound, we have that:
\begin{align*}
 \mathrm{Pr} \{ (\A\hat{\x})_{i} \frac{W_{a}}{a_{i}} \leq (1-\epsilon) G_\delta W_{a} \} \leq \exp(-\epsilon^{2}\frac{G_\delta W_{a}}{2} ).   
\end{align*}
Setting $(1-\epsilon)G_\delta \!=\!1$, i.e., $\epsilon = 1-\frac{1}{G_\delta}$, we have:
\begin{align} \label{eqn_chernoff_bnd1}
\mathrm{Pr} \left\{ (\A\hat{\x})_{i} \frac{W_{a}}{a_{i}} \leq W_{a} \right\} \leq \exp\left(- \Big(1-\frac{1}{G_\delta}\Big)^{2} G_\delta \frac{W_{a}}{2} \right). \!\!\!
\end{align}
Forcing $\exp\!\left(- (1-\frac{1}{G_\delta})^{2} G_\delta \frac{W_a}{2} \right) \!=\! \frac{\delta}{3m}$ and solving $G_\delta$, we have:
\begin{align} \label{eqn_L_def}
G_\delta \triangleq 1 + \frac{\ln(3m/\delta)}{W_a}+\sqrt{ \bigg(\frac{\ln(3m/\delta)}{W_a} \bigg)^{2} + \frac{2\ln(3m/\delta)}{W_a} }. \!\!\!
\end{align}

Using (\ref{eqn_exp2}), the Chernoff bound, and following similar arguments, we have:
\begin{align*}
  \mathrm{Pr} \Big\{ (\B\hat{\x})_{i} \frac{W_b}{b_i} > (1+\epsilon) G_\delta W_b \Big\} \leq \exp \Big( -\epsilon^{2}\frac{G_\delta W_b}{3} \Big). 
\end{align*}
Forcing $\exp\left( -\epsilon^{2} \frac{G_\delta W_b}{3} \right) = \frac{\delta}{3r}$ and solving for $\epsilon$, we have $\epsilon = (\frac{3}{G_\delta W_b}\ln(\frac{3r}{\delta}))^{\frac{1}{2}}$. 
It follows that 
\begin{align*}
 \mathrm{Pr} \Big\{ (\B\hat{\x})_{i} \frac{W_b}{b_i} > ( 1 + (\frac{3}{G_\delta W_b}\ln(\frac{3r}{\delta}))^{\frac{1}{2}} )G_\delta W_b \Big\} \leq \frac{\delta}{3r}, \end{align*}
which implies that: 
\begin{align} \label{eqn_chernoff_bnd2}
\!\!\! \mathrm{Pr} \left\{ (\B\hat{\x})_{i} > b_i \left( 1+ \sqrt{\frac{3}{G_\delta W_b}\ln\Big(\frac{3r}{\delta}\Big)} \right)G_\delta, \exists i \right\} \leq \frac{\delta}{3r}. \!\!\!\!\!
\end{align}
By using union bound and (\ref{eqn_chernoff_bnd1}) and (\ref{eqn_chernoff_bnd2}), we have that events 1)--3) occur with probability less than $\frac{\delta}{3} + m\cdot\frac{\delta}{3m} + r\cdot\frac{\delta}{3r} = \delta$, 
and the proof is complete.
\end{proof}

\section{Proof of the Competitive Ratio}  \label{appdx:thm_Competitive}

We use $OPT$ as the optimal objective value of Problem~R-DMLRS, which is also the optimum to Problem~D-R-DMLRS.
We let $\hat{\pi}_{i}$ denote the approximate schedule obtained by Algorithm~2, which inexactly solves Problem D-R-DMLRS.
Let $P_{i}$ and $D_{i}$ be the primal and dual objective values of Problems R-DMLRS and D-RMLRS after determining the schedule $\hat{\pi}_{i}$ in Algorithm~1.
Let $P_0$ and $D_0$ be the initial values of Problems R-DMLRS and D-RMLRS, respectively, where $P_0=0$ and $D_0=\sum_{t\in\mathcal{T}}\sum_{h\in\mathcal{H}}\sum_{r\in\mathcal{R}}P_h^r[0]C_h^r$.
We also let $P_I$ and $D_i$ be the final primal and dual objective values returned by Algorithm~1.
We present our main result in Lemma~\ref{lem:result}.
\begin{lem}
\label{lem:result}
If there exists constants $\epsilon\geq 1$, $G_\delta> 0$ and $\delta\in(0,1]$ such that $P_i-P_{i-1}\geq \frac{\delta/3G_\delta}{\epsilon}(D_i-D_{i-1}), \forall i\in\mathcal{I}$, and if $P_0=0$ and $D_0\leq \frac{1}{2}OPT$, then Algorithm~1 is $\frac{6G_\delta\epsilon}{\delta}$-competitive.
\end{lem}

\begin{proof}[Proof of Lemma~\ref{lem:result}]
Since $P_I=\sum\limits_{i\in\mathcal{I}}(P_i-P_{i-1})$, and $D_I-D_0=\sum\limits_{i\in\mathcal{I}}(D_i-D_{i-1})$, we can have:
\begin{align*}
    P_I&=\sum_{i\in\mathcal{I}}(P_i-P_{i-1})\geq \frac{\delta/3G_\delta}{\epsilon}\sum_{i\in\mathcal{I}}(D_i-D_{i-1})\\
    &=\frac{\delta/3G_\delta}{\epsilon}(D_I-D_0).
\end{align*}
By weak duality theorem~\cite{Bertsekas99:Nonlinear}, we have
\begin{align*}
    D_I\geq OPT\geq P_I.
\end{align*}
Thus, it follows that:
\begin{align*}
    &D_I-D_0\geq \frac{1}{2}OPT,\\
    &P_I\geq \frac{\delta/3G_\delta}{\epsilon}(D_I-D_0)\geq \frac{\delta/3G_\delta}{2\epsilon}OPT, 
\end{align*}
and the proof is complete.
\end{proof}

Next, following similar arguments in~\cite{Buchbinder09,Bao18:ML_INFOCOM}, 
we introduce the relationship between the cost and resource consumption before and after processing one job.
Let $p_h^{r,i}[t]$ be the unit cost of type-$r$ resource on server $h$ at time $t$ after handling job $i$.
Let $\rho_h^{r,i}[t]$ be the amount of type-$r$ resource allocated to jobs on server $h$ at time $t$ after processing the job $i$.
For ease of our subsequent analysis, we now define the following allocation-cost relation that is implied by Algorithm~1:

\begin{defn} \label{defn:alloc_cost}
The allocation-cost relationship for Algorithm~1 with $\epsilon>1, G_\delta> 0$ and $\delta\in(0,1]$ is
\begin{align*}
    &p_h^{r,i-1}[t](\rho_h^{r,i}[t]-\rho_h^{r,i-1}[t])\geq \frac{\delta/3G_\delta C_h^r}{\epsilon}(p_h^{r,i}[t]-p_h^{r,i-1}[t]),\\ 
    &\hspace{2in}\forall i\in\mathcal{I},h\in\mathcal{H},r\in\mathcal{R}.
\end{align*}
\end{defn}

The allocation-cost relationship shows that the cost in each time slot for scheduling a new job is bounded by the increase of term $C_h^rp_h^r[t]$ in Problem~D-R-DMLRS, and the possible increment introduced by randomized rounding in Algorithm~4.
This is ensured by the update of the price function and the rounding scheme, respectively.

\begin{lem}
\label{lem:relation}If the allocation-cost relationship holds for $\epsilon\geq 1, G_\delta> 0$ and $\delta\in(0,1]$, then Algorithm~1 ensures $P_i-P_{i-1}\geq \frac{\delta/3G_\delta}{\epsilon}(D_i-D_{i-1}),\forall i\in\mathcal{I}$.
\end{lem}
\begin{proof}[Proof of Lemma~\ref{lem:relation}]
For any job $i\in\mathcal{I}$, if job $i$ is rejected, then we have $P_i-P_{i-1}=D_i-D_{i-1}=0$ according to Problems~R-DMLRS and D-R-DMLRS, the inequality must hold.
If job $i$ is accepted with schedule $\pi_i$, i.e., $x_{\pi_i}=1$, then the increment value of the primal objective value $P_i$ is
\begin{align*}
    P_i-P_{i-1}=u_i(t_{\pi_i}-a_i).
\end{align*}
Since $x_{\pi_i}=1$, according to Algorithm~1, the constraint~(\ref{eqn_dualconstr}) is binding.
Then, we can have
\begin{align*}
    &u_i(t_{\pi_i}-a_i)\\
    &=\lambda_i+\sum_{t\in\mathcal{T}(\pi_i)}\sum_{h\in\mathcal{H}(\pi_i[t])}\sum_{r\in\mathcal{R}}(\alpha_i^rw_{ht}^{\pi_i}+\beta_i^rs_{ht}^{\pi_i})p_h^r[t]\\
    &=\lambda_i+\sum_{t\in\mathcal{T}(\pi_i)}\sum_{h\in\mathcal{H}(\pi_i[t])}\sum_{r\in\mathcal{R}}p_h^r[t](\rho_h^{r,i}[t]-\rho_h^{r,i-1}[t]).
\end{align*}
Similarly, we can have the increment value of the dual objective value $D_i$ as follows:
\begin{align*}
    D_i-D_{i-1}\!=\!\lambda_i+\sum_{t\in\mathcal{T}(\pi_i)}\sum_{h\in\mathcal{H}(\pi_i[t])}\sum_{r\in\mathcal{R}}(p_h^{r,i}[t]-p_h^{r,i-1}[t])C_h^r.
\end{align*}
Summing up the allocation-cost relationship over all $t\in\mathcal{T}(\pi_i),h\in\mathcal{H}(\pi_i[t]),r\in\mathcal{R}$, we have
\begin{align*}
    P_i-P_{i-1}&\geq\frac{\delta/3G_\delta}{\epsilon}(D_i-D_{i-1}-\lambda_i)+\lambda_i\\
    &=\frac{\delta/3G_\delta}{\epsilon}(D_i-D_{i-1})+(1-\frac{\delta/3G_\delta}{\epsilon})\lambda_i.
\end{align*}
As $\lambda_i\geq0$, $\epsilon, G_\delta> 0$ and $\delta\in(0,1]$, we have
\begin{align*}
    P_i-P_{i-1}\geq\frac{\delta/3G_\delta}{\epsilon}(D_i-D_{i-1}).
\end{align*}
This completes the proof of Lemma~\ref{lem:relation}.
\end{proof}

For specific $h\in\mathcal{H},r\in\mathcal{R}$, we define $\epsilon_h^r$ as the corresponding parameter in the allocation-cost relationship for any job $i\in\mathcal{I}$ and any time slot $t\in\mathcal{T}$.
Then, it holds that $\epsilon=\max_{h,r}\{\epsilon_h^r\}$.
Without loss of generality, we assume that the resource demand of each worker or parameter server is much smaller compared to the capacity of that resource on one server, i.e., $\alpha_i^r\ll C_h^r, \beta_i^r\ll C_h^r$.
This is common in real-world machine learning system as large percentage of resources in the whole server.
As $\rho_h^r[t]$ increases from 0 to $C_h^r$, then we can claim that $d\rho_h^r[t]=\rho_h^{r,i}-\rho_h^{r,i-1}$, and derive a differential version of the allocation-cost relationship, which is defined as follows:

\begin{defn} \label{defn:diff_alloc_cost}
The differential allocation-cost relationship for Algorithm~1 with $\epsilon_h^r\geq 1$ is
\begin{align*}
    p_h^r[t]d\rho_h^r[t]\geq \frac{C_h^r}{\epsilon_h^r}dp_h^r[t],\forall t\in\mathcal{T},h\in\mathcal{H},r\in\mathcal{R}.
\end{align*}
\end{defn}
Next we show that a feasible $\epsilon_h^r$ satisfies the differential allocation-cost relationship with price function $p_h^r[t]$ defined in (\ref{eqn_cost_update_fn}).

\begin{lem}\label{lem:diff_alloc_cost}
$\epsilon_h^r=\ln\frac{U^r}{L}$, and the price function defined in (\ref{eqn_cost_update_fn}) satisfy the differential allocation-cost relationship.
\end{lem}

\begin{proof}[Proof of Lemma~\ref{lem:diff_alloc_cost}]
The derivation of the marginal cost function is
\begin{align*}
    dp_h^r[t]=p_h^{r'}(\rho_h^r[t])d\rho_h^r[t]=L(\frac{U^r}{L})^{\frac{\rho_h^r[t]}{C_h^r}}\ln(\frac{U^r}{L})^{\frac{1}{C_h^r}}d\rho_h^r[t].
\end{align*}
The differential allocation-cost relationship is
\begin{align*}
    L(\frac{U^r}{L})^{\frac{\rho_h^r[t]}{C_h^r}}d\rho_h^r[t]\geq \frac{C_h^r}{\epsilon_h^r}L(\frac{U^r}{L})^{\frac{\rho_h^r[t]}{C_h^r}}\ln(\frac{U^r}{L})^{\frac{1}{C_h^r}}d\rho_h^r[t],
\end{align*}
which holds for $\epsilon_h^r\geq \ln(\frac{U^r}{L})$. 
Then, we can set $\epsilon\!=\!\max\limits_{r\in\mathcal{R}}(1,\ln(\frac{U^r}{L}))$, which satisfies the differential allocation-cost relationship.
This completes the proof.
\end{proof}

With the aforementioned lemmas, we are now in a position to prove Theorems~\ref{thm_Competitive_1} and \ref{thm_Competitive_2}.
Note that the Theorems provide probabilistic guarantees.
We first analyze the performance ratio, which is followed by the probabilistic feasibility discussion for both cases.
\begin{proof}[Proof of Theorems~\ref{thm_Competitive_1} and \ref{thm_Competitive_2}]
According to Lemma~\ref{lem:diff_alloc_cost}, the marginal cost function used in Algorithm~1 satisfies the differential allocation-cost relationship with $\epsilon=\max_r(1,\ln\frac{U^r}{L})$.
Since the resource demand in a job $i$ is much smaller than the capacity, we can derive
\begin{align*}
    &d\rho_h^r[t]=\rho_h^{r,i}-\rho_h^{i-1}[t], \\
    &dp_h^r[t]=p_h^{r'}(\rho_h^r[t])(\rho_h^{r,i}[t]-\rho_h^{r,i-1}[t])=p_h^{r,i}[t]-p_h^{r,i-1}[t].
\end{align*}
So, the differential allocation-cost relationship in Definition~\ref{defn:diff_alloc_cost} implies the allocation-cost relationship in Definition~\ref{defn:alloc_cost} holds for $\epsilon=\max_r(1,\ln\frac{U^r}{L})$.

According to Algorithm~1, we note that
\begin{align*}
 \frac{1}{\mu}\leq \frac{\ceil{E_{i}K_{i}(\tau_{i}+2g_{i}\gamma_i/(b_{i}^{(e)}F_i))}\sum_{r\in\mathcal{R}}(\alpha_{i}^{r}+\beta_{i}^{r})}{T\sum_{h\in\mathcal{H}}\sum_{r\in\mathcal{R}}C_{h}^{r}},\forall i\in\mathcal{I}.  
\end{align*}
Then, the minimum amount of overall resource consumption of  job $i$ can be computed as:
\begin{multline*}
\frac{T\sum_{h\in\mathcal{H}}\sum_{r\in\mathcal{R}}C_{h}^{r}}{\mu}\\
   \leq \ceil{E_{i}K_{i}(\tau_{i}+2g_{i}\gamma_i/(b_{i}^{(e)}F_i))}\sum_{r\in\mathcal{R}}(\alpha_{i}^{r}+\beta_{i}^{r}).
\end{multline*}
Then, it follows that:
\begin{align*}
    &D_0=\sum_{t,h,r}LC_h^r\\
    &=\sum_{t,h,r}\min_{i\in\mathcal{I},\pi_i\in\Pi_i}\frac{1/(2\mu) u_{i}(t_{\pi_i}-a_{i})C_h^r}{\sum_{r\in\mathcal{R}}\ceil{E_{i}K_i(\tau_{i}+2g_{i}\gamma_i/(b_{i}^{(e)}F_i)}(\alpha_{i}^{r}+\beta_{i}^{r})}\\
    &=\frac{T\sum_{h,r}C_h^r}{2\mu}\min_{i,\pi_i}\frac{u_{i}(t_{\pi_i}-a_{i})}{\sum_{r\in\mathcal{R}}\ceil{E_{i}K_i(\tau_{i}+2g_{i}\gamma_i/(b_{i}^{(e)}F_i)}(\alpha_{i}^{r}+\beta_{i}^{r})}\\
    &\leq \frac{1}{2}\ceil{E_{i}K_{i}(\tau_{i}+2g_{i}\gamma_i/(b_{i}^{(e)}F_i))}\sum_{r\in\mathcal{R}}(\alpha_{i}^{r}+\beta_{i}^{r})\\
    &\min_{i\in\mathcal{I},\pi_i\in\Pi_i}\frac{u_{i}(t_{\pi_i}-a_{i})}{\sum_{r\in\mathcal{R}}\ceil{E_{i}K_i(\tau_{i}+2g_{i}\gamma_i/(b_{i}^{(e)}F_i)}(\alpha_{i}^{r}+\beta_{i}^{r})}, \forall i\in\mathcal{I}.\\
    &\overset{(a)}{\leq}\frac{1}{2}\ceil{E_{i}K_{i}(\tau_{i}+2g_{i}\gamma_i/(b_{i}^{(e)}F_i))}\sum_{r\in\mathcal{R}}(\alpha_{i}^{r}+\beta_{i}^{r})\\
    &\frac{u_{i}(t_{\pi_i}-a_{i})}{\sum_{r\in\mathcal{R}}\ceil{E_{i}K_i(\tau_{i}+2g_{i}\gamma_i/(b_{i}^{(e)}F_i)}(\alpha_{i}^{r}+\beta_{i}^{r})}\\
    &\leq \frac{1}{2}u_i(t_{\pi_i}-a_i)\overset{(b)}{\leq} \frac{1}{2}OPT,
\end{align*}
where (a) follows by selecting $(i,\pi)=\arg\min_{i\in\mathcal{I},\pi_i\in\Pi_i}u_i(t_{\pi_i}-a_i)$, and (b) follows from the assumption that the offline optimal solution accepts at least one job, which is reasonable in real-world machine learning system.
Then we have $OPT\geq \min_{i,\pi}u_i(t_{\pi_i}-a_i)$.
According to Lemmas~\ref{lem:result} and \ref{lem:relation}, the competitive ratio is proved.

Recall that the randomized rounding algorithm is a key component in our algorithm.
Toward this end, we show the probability of obtaining a feasible solution with the proved competitive ratio.
Here, we consider the following two cases:

{\em 1) } When $0<G_\delta\leq 1$ (Theorem~\ref{thm_Competitive_1}): According to Theorem~\ref{thm_Alg4_1}, the probability of violating the cover constraint is no greater than $\delta/3$ at each randomized rounding iteration.
Recall that our Algorithm~1 runs a predetermined number of $S$ iterations to find a feasible integer solution.
Thus the probability that no feasible integer solution returned after $S$ iterations rounding is at most $(\delta/3)^S$.
It then follows that the probability of at least one feasible integer solution found is at least $1-(\delta/3)^S$.
Moreover, the number of states $(t,v)$ in the dynamic programming for each job $i$ is $O(TK_iE_i)$.
Therefore, with probability greater than $(1-(\delta/3)^S)^{T_iK_iE_i}$, PD-ORS in Algorithm~1 returns a feasible integer solution with the proved competitive ratio.

{\em 2) } When $G_\delta>1$ (Theorem~\ref{thm_Competitive_2}): According to Theorem~\ref{thm_Alg4_2}, the probability of violating the packing constraint is no greater that $\delta/3(HR+1)$ at each randomized rounding iteration.
Following the similar arguments in {\em 1)}, we can show that with probability greater than $(1-(\delta/3(HR+1))^S)^{T_iK_iE_i}$, PD-ORS in Algorithm~1 returns a feasible integer solution with the proved competitive ratio, and the proof is complete.
\end{proof}
\end{document}